\documentclass{uwstat518}

\usepackage{amsmath}
\usepackage[square,sort,comma,numbers]{natbib}

\usepackage{times}
\usepackage{bm}
\usepackage{natbib}
\usepackage{mathtools}
\usepackage{amsmath}
\usepackage{amssymb}
\usepackage{url}
\usepackage{stmaryrd}
\usepackage{amsthm}

\usepackage{graphicx}

\usepackage{multirow}

\usepackage{tikz}

\def\checkmark{\tikz\fill[scale=0.4](0,.35) -- (.25,0) -- (1,.7) -- (.25,.15) -- cycle;}

\newcommand{\citeU}{\citep}
\newcommand{\Ll}{L}
\newcommand{\A}{A}
\newcommand{\D}{D}
\newcommand{\B}{B}
\newcommand{\F}{F}
\newcommand{\C}{C}
\newcommand{\W}{W}

\newcommand{\bff}{}

\newtheorem{theorem}{Theorem}
\newtheorem{corollary}{Corollary}

\newtheorem{proposition}{Proposition}
\newtheorem{lemma}{Lemma}

\newtheorem{example}{Example}

\usepackage[plain,noend]{algorithm2e}

\makeatletter
\renewcommand{\algocf@captiontext}[2]{#1\algocf@typo. \AlCapFnt{}#2} 
\def\@algocf@capt@plain{top}
\renewcommand{\algocf@makecaption}[2]{%
  \addtolength{\hsize}{\algomargin}%
  \sbox\@tempboxa{\algocf@captiontext{#1}{#2}}%
  \ifdim\wd\@tempboxa >\hsize
    \hskip .5\algomargin%
    \parbox[t]{\hsize}{\algocf@captiontext{#1}{#2}}
  \else%
    \global\@minipagefalse%
    \hbox to\hsize{\box\@tempboxa}
  \fi%
  \addtolength{\hsize}{-\algomargin}%
}
\makeatother

\begin{document}







\title{Identifiability and Estimation of Structural Vector Autoregressive Models for Subsampled and Mixed Frequency Time Series}

\author{
	Alex Tank \\
       Department of Statistics \\
       University of Washington \\
       alextank@uw.edu
	\and 
	Emily Fox \\
       Department of Statistics \\
       University of Washington \\
       ebfox@uw.edu
       \and
       Ali Shojaie \\
Department of Biostatistics\\
University of Washington\\
ashojaie@uw.edu
}

\maketitle

\begin{abstract}
Causal inference in multivariate time series is challenging due to the fact that the sampling rate may not be as fast as the timescale of the causal interactions. In this context, we can view our observed series as a \emph{subsampled} version of the desired series. Furthermore, due to technological and other limitations, series may be observed at different sampling rates, representing a \emph{mixed frequency} setting. To determine instantaneous and lagged effects between time series at the true causal scale, we take a model-based approach based on structural vector autoregressive (SVAR) models. In this context, we present a unifying framework for parameter identifiability and estimation under both subsampling and mixed frequencies when the noise, or shocks, are non-Gaussian. Importantly, by studying the SVAR case, we are able to both provide identifiability and estimation methods for the causal structure of both lagged and instantaneous effects at the desired time scale. We further derive an exact EM algorithm for inference in both subsampled and mixed frequency settings. We validate our approach in simulated scenarios and on two real world data sets.
\end{abstract}
 \section{Introduction}\label{sec:intro}
Classical approaches to multivariate time series and Granger causality assume that all time series are sampled at the same sampling rate. However, due to data integration across heterogeneous sources, many data sets in econometrics, health care, environment monitoring, and neuroscience are composed of multiple time series sampled at different rates, referred to as \emph{mixed frequency} time series. Furthermore, due to the cost or technological challenge of data collection, many time series may be sampled at a rate lower than the true causal scale of the underlying physical process. For example, many econometric indicators, such as GDP and housing price data, are recorded at quarterly and monthly scales \citeU{moauro:2005}. However, there may be important interactions between these indicators at the weekly or bi-weekly scales \citeU{moauro:2005,stram:1986,boot:1967}. In neuroscience, imaging technologies with high spatial resolution, like functional magnetic resonance imaging or fluorescent calcium imaging, have relatively low temporal resolutions. On the other hand, it is well-known that many important neuronal processes and interactions happen at finer time scales \citeU{zhou:2014}. A causal analysis rooted at a slower time scale than the true causal time scale may both miss true interactions and add spurious ones \citeU{zhou:2014,silvestrini:2008,boot:1967,breitung:2002}. A comprehensive approach to Granger causality in multivariate time series should be able to simultaneously accommodate both mixed frequency and subsampled data. 

Recently, the problem of causal discovery in subsampled time series has been studied drawing from methods in causal structure learning using graphical models \citeU{danks:2013,pils:2015,plis:2015,hyttinen:2016}. These methods are model free, and automatically infer a sampling rate for causal relations most consistent with the data. We maintain a similar goal, but take a model-based approach and examine the  identifiability of structural vector autoregressive models (SVAR) under both subsampling and mixed frequency settings. SVARs are an important tool in time series analysis \citeU{Lutkepohl:2005,harvey:1990} and are a mainstay in econometrics and macro-economic policy analysis. SVAR models combine classical linear autoregressive models with structural equation modeling \citeU{ullman:2003} to allow analysis of both instantaneous and lagged causal effects between time series. However, SVAR models are commonly applied to regularly sampled data, where each series is observed at the same, discrete regular intervals. Moreover, the time scale of a causal SVAR analysis is typically restricted to this shared sampling scale. 

\cite{Gong:2015} recently explored identifiability and estimation for VAR models under subsampling with independent innovations, i.e.,  \emph{no} instantaneous causal effects or error correlations. They show that with non-Gaussian errors, the transition matrix is identifiable under subsampling, implying that Granger causality estimation under subsampling is possible. Unfortunately, their results do not cover the case of correlated errors, a common and important aspect of many real world time series and their respective models \citeU{Lutkepohl:2005}. Interestingly, non-Gaussian errors have also been shown to aide model identifiability in SVAR models with standard sampling assumptions \citeU{zhang:2009,lanne:2015,hyvarinen:2010,hyvarinen:2008,peters:2013}. 
This line of work applies techniques originally developed for both structural equation modeling with non-Gaussian errors and  independent component analysis (ICA) \citeU{hyvarinen:2004} to the SVAR context. Importantly, non-Gaussian errors allow identification of the SVAR model without any other identifying restrictions \citeU{lanne:2015}, and further allow identification of the causal ordering of the instantaneous effects if these are known to follow a directed acyclic graph (DAG) \citeU{hyvarinen:2010}.

Our approach to subsampling unifies existing approaches to identifiability along two complimetary directions.

\begin{enumerate}
\item Our work concretely connects the non-Gaussian subsampled VAR with independent innovations method \citeU{Gong:2015} with the now extensive non-Gaussian SVAR framework \citeU{zhang:2009,lanne:2015,hyvarinen:2010,hyvarinen:2008,peters:2013} by proving identifiability of an  SVAR model of order one under arbitrary subsampling. As a result, we find that not only can one identify the causal structure of lagged effects from subsampled data with correlated errors, but also the DAG of the instantaneous effects without prior knowledge of the causal ordering. 
\item We generalize our results to the mixed frequency setting with arbitrary subsampling, where the subsampling level may be different for each time series.  In doing so, we provide a unified theoretical approach and estimation methodology for subsampled and mixed frequency cases. Precise identifiability conditions on the model parameters in the mixed frequency case is notoriously difficult \citeU{anderson:2015} and has only been studied based on the first two moments of the mixed frequency process. Our work takes a complimentary direction by leveraging higher order moments and provides the first set of specific model conditions for mixed frequency SVAR models needed for identifiability. Furthermore, previous approaches to mixed frequency SVAR have assumed a causal ordering, while our results indicate this may be estimated by leveraging non-Gaussianity. Finally, our approach to identifiability allows us to move beyond the classical mixed frequency setting where the time scale is fixed at the most finely sampled series \citeU{anderson:2015}, and instead consider identifiability and estimation in more general mixed frequency cases. We display the four sampling types our approach covers in Figure \ref{sampling_types}.
\end{enumerate}

We introduce an exact EM algorithm for inference for both subsampled and mixed frequency cases. \cite{Gong:2015} also utilize an EM algorithm, but because they formulate inference directly on the subsampled process by marginalizing out the missing data, the approach requires an extra layer of approximation.  Our approach instead casts inference as a missing data problem and utilizes a Kalman filter to exactly compute the E-step for both subsampled and mixed frequency cases. We validate our estimation and identifiability results via extensive simulations and apply our method to evaluate causal relations in a subsampled climate data set and a mixed frequency econometric dataset. Taken together, we present a unified theoretical analysis and unified estimation methodology for both subsampled and mixed frequency SVAR cases, areas that have been traditionally studied separately. A summary of our contributions are presented in Table \ref{estimation_contributions}.

\begin{figure}
\centering
\includegraphics[width = .6\textwidth]{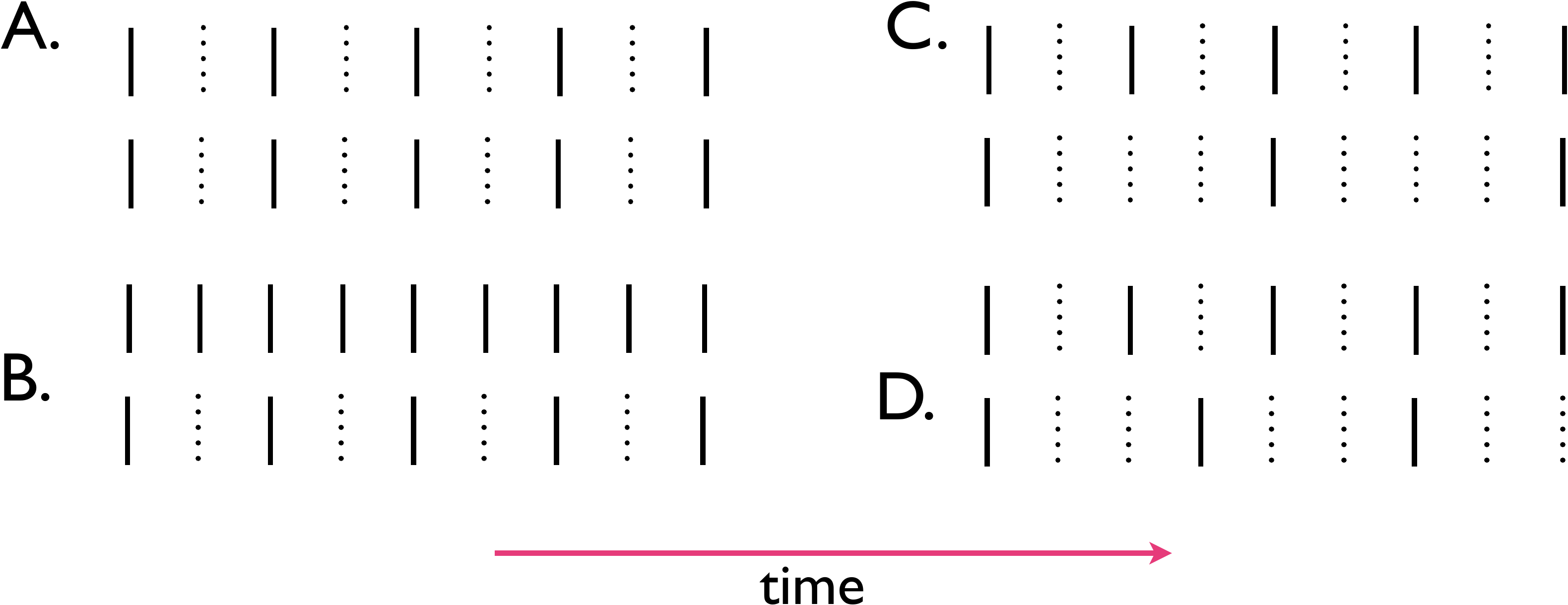}
\caption{Three different types of structured sampling. Black lines indicate observed data and dotted lines indicate missing data. A) Both series are subsampled at a rate of two. B) The standard mixed frequency example \citeU{anderson:2015}, with the first having no subsampling while the second series is subsampled. C) A mixed frequency subsampled version of B where each series is subsampled, but at different rates. D) Another subsampled mixed frequency series, but where there is no common factor across sampling rates and is thus not a subsampled version of B. } \label{sampling_types}
\end{figure}

\begin{table}[ht]
\begin{center}
\caption{Tabular summary of the contributions of our work to identifiability and estimation in mixed frequency sampling SVAR models. The subsampling types are as in Figure \ref{sampling_types}. Citations indicate previous work and the check marks indicate our contributions. The notation ce indicates `computationally expensive';  see the discussion at the end of Section \ref{real_data}. Hyv08 represents \cite{hyvarinen:2008}, Gong15 represents \cite{Gong:2015} and Lut06 represents \cite{Lutkepohl:2005}.}
\begin{tabular}{c c c c c l l}
    sampling type & & none & A & B & \multicolumn{1}{c}{C} & \multicolumn{1}{c}{D}\\
    \multirow{2}{*}{$\C = I$}&ident.& cf. Lut06 & Gong15 & \checkmark & \checkmark & \checkmark \\ 
    &est.& cf. Lut06 & Gong15 (approx), \checkmark & \checkmark & \checkmark (ce) & \checkmark (ce) \\
    \multirow{2}{*}{$\C$ free}&ident.& Hyv08  & \checkmark & \checkmark & \checkmark & \checkmark\\ 
    &est.& Hyv08  & \checkmark & \checkmark & \checkmark (ce)
 & \checkmark (ce)\\

\end{tabular}
 \label{estimation_contributions}
\end{center}
\label{tab:multicol}
\end{table}




 \section{Background}
 \label{svar_b}
Let $x_t \in \mathbb{R}^{p}$ be a $p$-dimensional multivariate time series for $t = 1, \ldots, T$ generated at a fixed sampling rate. We collect the entire set of $x_t$s into the matrix ${\bff X} = \left(x_1, \ldots,x_T \right)$. We assume the dynamics of $x_t$ follow a combination of instantaneous effects, lagged autoregressive effects, and independent noise
\begin{align} \label{SVAR}
x_t &= \B x_t + \D x_{t-1} + e_t,
\end{align}
where $\B \in \mathbb{R}^{p \times p}$ is the structural matrix that determines the instantaneous time linear effects, $\D \in \mathbb{R}^{p \times p}$ is an autoregressive matrix that specifies the lag one effects conditional on the instantaneous effects, and $e_t \in \mathbb{R}^{p}$ is a white noise process such that $E(e_t) = 0 \,\, \forall t$, and $e_{ti}$ is independent of $e_{t'j}$ $\forall i, j,t,t'$ such that $(i,t) \neq (j,t)$ We assue $e_{tj}$ is distributed as  $e_{tj} \sim p_{e_j}$. Solving Eq. (\refeq{SVAR}) in terms of $x_t$ gives the following equation for the evolution of $x_t$:
\begin{align}
x_t &= (I - \B)^{-1} \D x_{t-1} + (I - \B)^{-1} e_t \nonumber\\
&= \A x_{t-1} + \C e_t \label{svar_f}
\end{align}
Under the representation in Eq. (\refeq{svar_f}), each $\A_{ji}$ element denotes the lag one linear effect of series $i$ on series $j$ and $\C \in \mathbb{R}^{p \times p}$ is the structural matrix. Element $e_{tj}$ is refered to as the \emph{shock} to series $j$ and element $\C_{ji}$ is the linear instantaneous effect of shock $j$ on series $j$ to series $i$.

Conditions on $\C$, or equivalently $\B$, for model identifiability and estimation have been heavily explored \citeU{harvey:1990}.
The most typical condition is that $\C$ is a lower triangular matrix with ones on the diagonal, implying a known causal ordering to the instantaneous effects. In this case, one may interpret the instantaneous effects as a directed acyclic graph (DAG) \citeU{lauritzen:1996}. A DAG is a directed graph, $G = (V, E)$, with vertices $V = \{1, \ldots, p\}$ and directed edge set $E$, with no directed cycles. A causal ordering for a DAG is an ordering of the vertices into a sequence, $\pi$, such that if $j$ comes before $i$ in $\pi$ then $E$ does not contain a path of edges from $i$ to $j$; see, e.g.,  \citeU{shojaie2010biometrika} for more details.
In the context of SVARs, for $i \neq j$ there exists a directed edge $i \to j$ from $x_i$ to $x_j$ in $E$, if and only if $\C_{ji}$ is nonzero. Classical estimation for SVAR models with known causal ordering typically proceeds by simultaneously fitting $\A$ and $\C$ with the identifiability constraint that $\C$ be lower triangular.  

A recent line of work \citeU{zhang:2009,lanne:2015,hyvarinen:2010} focuses on estimating $\A$ and $\C$ when $\pi$ is unknown. They show that when the errors, $e_t$, are \emph{non-Gaussian}, both the causal ordering and instantaneous effects $\C$, or $\B$, may be inferred directly from the data using techniques common in independent component analysis (ICA) \citeU{hyvarinen:2010}. Alternatively, one may dispense with causal orderings and lower triangular restrictions all together and directly estimate $\C$ \citeU{lanne:2015} under non-Gaussian errors. Our analysis continues this direction of work, leveraging non-Gaussianity of SVAR with subsampling and/or mixed frequency sampling.

\section{Subsampled SVAR} \label{subsamp_sec}
\subsection{The subsampled process}
Subsampling occurs when, due to low temporal resolution, we only observe $x_t$ every $k$ time steps, as displayed graphically as case A in Figure \ref{sampling_types}. In this case, we only have access to the observations ${\bff \tilde{X}} = \left(\tilde{x}_1, \tilde{x}_2, \ldots, \tilde{x}_{\tilde{T}}\right) \equiv \left(x_1, x_{1 + k},\ldots,x_{1 + (\tilde{T}-1)k}\right)$, where $\tilde{T}$ is the number of subsampled observations. We may marginalize out the unobserved $x_t$ to obtain the evolution equations for $\tilde{x}_t$:
\begin{align}
\tilde{x}_{t+1} &= x_{1 + tk} = \A x_{1 + tk - 1} + \C e_{1 + tk} \nonumber \\ 
&= \A \left( \A x_{1 + tk - 2} + \C e_{1 + tk - 1} \right) + \C e_{1 + tk} \nonumber \\
&= \ldots \nonumber \\
&= \A^k \tilde{x}_{t-1} + \sum_{l = 0}^{k-1} \A^l \C e_{1 + tk - l} \\
&= \A^k \tilde{x}_{t-1} + \Ll \tilde{e}_t \label{svar_sub},
\end{align} 
where $\tilde{e}_t = \left(e_{1 + tk}^T, \ldots e_{2 + (t - 1)k}^T\right)^T$ and $\Ll = \left(\C, \ldots, \A^{k-1} \C\right)$.
Eq. (\refeq{svar_sub}) appears to take a similar form to the SVAR process in Eq. (\refeq{SVAR}); however, now the vector of shocks, $\tilde{e}_t$, is of dimension $kp$ with special structure on both the structural matrix $\Ll$ and the distributions of the elements in $\tilde{e}_t$. Unfortunately, this representation no longer has the interpretation of instantaneous causal effects as described in Section \ref{svar_b} since there are now multiple shocks per individual time series. We will refer to the full parametrization of the subsampled SVAR model in Eq. (\refeq{svar_sub}) as $\left(\A, \C, p_e; k\right)$. Identifiability of the SVAR model means that there is a unique pair of $\A$ and $\C$ for the SVAR model consistent with the joint distribution of ${\bff \tilde{X}}$ at subsampling rate $k$.

\subsection{Lagged and Instantaneous Causality Confounds of Subsampling}
\label{cs}
A classical SVAR analysis on the $\tilde{x}_t$ that does not account for subsampling would incorrectly estimate lagged Granger causal effects in $\A^k$; this is because, $\A_{ij}$ being zero does not imply that $(\A^k)_{ij} = 0$, and vice versa \citeU{Gong:2015}. Zeros in the estimated structural matrix may also be incorrect if subsampling is ignored. Furthermore, classical SVAR estimation methods that assume a known causal ordering to the instantaneous shocks simply estimate the covariance of the error process, $\Sigma = E(\C e_t e_t^T \C^T) = \C \Lambda \C^T$, and let the estimated structural matrix be the Cholesky decomposition of $\Sigma$. Under subsampling, the covariance of the error process is 
\begin{align}
E(\Ll\tilde{e}_t \tilde{e}_t^T \Ll^T) = \Ll \left(I_{k} \otimes \Lambda \right) \Ll^T, \label{sub_cov}
\end{align}
where $\otimes$ is the Kronecker product and $I_{k}$ is the identity matrix of size $k$. The causal structure given by zeros in the Cholesky decomposition of Eq. (\refeq{sub_cov}) need not be the same as those implied by $\C$. 

\begin{example} As an example, consider the following process \citeU{Gong:2015}:
\begin{align}
\A = \left(\begin{array}{c c}
.8 & .5 \\
0 & -.5
\end{array}\right) \,\,\,\, \C = \left( \begin{array}{c c}
1 & 0 \\
0 & 1 
\end{array} \right) \,\,\,\, \Lambda = \left( \begin{array}{c c}
1 & 0 \\
0 & 1 
\end{array} \right)\nonumber
\end{align}
so that $\C \Lambda \C^T = I_p$. Then for a subsampling of $k = 2$,
\begin{align}
\A^k = \left(\begin{array}{c c}
.64 & 0 \\
0 & .64
\end{array} \right) \nonumber \,\,\,\,
\Ll \left(I_{k} \otimes \Lambda \right) \Ll^T = \left(\begin{array}{c c}
1.89 & -.4 \\
-.4 & 1.64
\end{array} \right), \nonumber
\end{align}
implying \emph{no} lagged causal effect between $x_1$ and $x_2$, but a relatively large instantaneous interaction; this is the opposite of the true data generating model! A graphical depiction of this example is given in Figure~\ref{confounds}. \label{ex1}
\end{example}

\begin{figure} 
\centering 
\includegraphics[width=.8\textwidth]{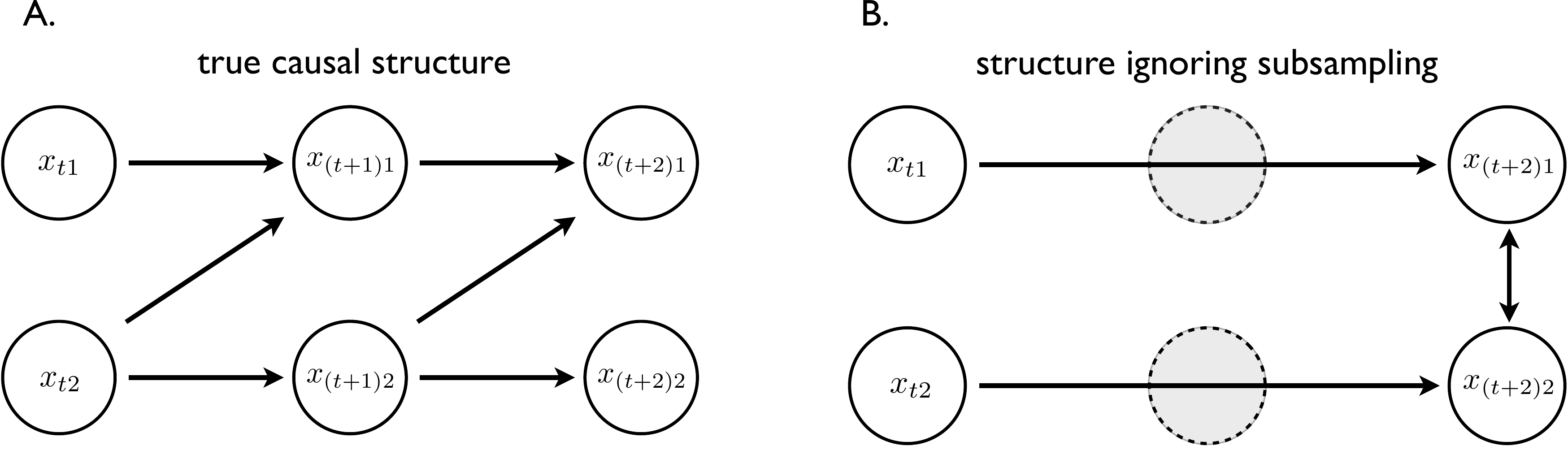}
\caption{Graphical depiction of how subsampling confounds both causal analysis of lagged and instantaneous effects. A) The true causal diagram for the regularly sampled data. B) The estimated causal structure of the subsampled process when the effects of subsampling are ignored.} \label{confounds}
\end{figure}

\subsection{Identifiability of L under subsampling}
While both lagged Granger causality and instantaneous structural interactions are confounded by subsampling, we show here that by accounting for subsampling, and under some conditions, we may still estimate the $\A$ and $\C$ matrices of the underlying SVAR process directly from the subsampled data. 
As a first step to proving the identifiability result of $\A$ and $\C$, we show that the matrix $\Ll = \left(\C, \ldots, \A^{k-1} \C\right)$ in Eq. (\refeq{svar_sub}) is identifiable up to permutation and scaling of columns when the $p_{e_j}$ (the distribution of $e_{tj}$) are all non-Gaussian.
\begin{proposition} \label{L_ident}
Suppose that all $p_{e_j}$ are non-Gaussian. Given a known subsampling factor $k$ and subsampled data $\tilde{X}$ generated according to Eq. (\refeq{svar_sub}), $\Ll$ may be determined up to permutation and scaling of columns.
\end{proposition} 
The proof closely follows the proof of Proposition 1 in \citeU{Gong:2015} and reposes upon the following foundational result in the field of ICA \citeU{eriksson:2004}.
\begin{lemma} \label{lemma_ica}
Let $\hat{e} = J r$ and $\hat{e} = M s$ be two representations of the n-dimensional random vector $\hat{e}$, where $J$ and $M$ are constant matrices of orders $n \times l$ and $n \times m$, respectively, and $r = (r_1, \ldots, r_l)^T$ and $s = (s_1, \ldots, s_m)^T$ are random vectors with independent components. Then the following assertions hold.
\begin{enumerate}
\item If the $i$th column of $J$ is not proportional to any column of $M$, then $r_i$ is Gaussian.
\item If the $i$th column of $J$ is proportional to the $j$th column of $M$, then the logarithms of the characteristic functions of $r_i$ and $s_j$ differ by a polynomial in a neighborhood of the origin.
\end{enumerate}
\end{lemma}
Intuitively, this result states that if $r$ is non-Gaussian with independent elements, and if $Jr = Ms$, it must be that $M$ and $J$ are equal up to permutation and scaling of columns. This implies that one may estimate $M$ from only observations of $\hat{e}$ and that the estimate of $M$ should be equal up to permutations and scalings of the true data generating $M$.

To apply Lemma \ref{lemma_ica} to Proposition \ref{L_ident}, note that $\A^k$ is identifiable by linear regression. Thus, the error component $\hat{e} = \tilde{x}_t - \A^k \tilde{x}_{t - 1} = \Ll \tilde{e}_t$ satisfies condition 1 of Lemma \ref{lemma_ica} and $\Ll$ is identifiable up to permutations and scalings since $\tilde{e}_t$ are non-Gaussian. 

\subsection{Complete identifiability of SVAR under $\C = I$}

 Using the identifiability result for $\Ll$ in Proposition \ref{L_ident} we can derive identifiability statements and conditions for $\C$ and $\A$ for the subsampled SVAR. We first require a few mild assumptions.

\begin{description}
\item[A1] $x_t$ is stationary so that all singular values of $\A$ have modulus less than one. 
\item[A2] The distributions $p_{e_j}$ are distinct for each $j$ after rescaling $e_j$ by any non-zero scale factor, their characteristic functions are all analytic (or they are all non-vanishing), and none of them has an exponent factor with polynomial of degree at least 2.
\item[A3] All $p_{e_j}$ are asymmetric.
\end{description}
Assumption A1 is standard in time series modeling \citeU{Lutkepohl:2005} and A2 is also common in non-Gaussian, ICA-type models. \cite{Gong:2015} provide identifiability results under A1 and A2 for the subsampled VAR with no instantaneous correlations, $\C = I$. We restate their result in our framework, both for comparison with the SVAR results and for use in Section~\ref{mf}, where we consider the mixed frequency setting. 

\begin{theorem} \label{Gong_theorem}
(Gong et al. 2015) Suppose $e_{tj}$ is non-Gaussian for all $t,j$, and that the data $\tilde{x}_t$ are generated by Eq. (\refeq{svar_f}) with $\C = I_p$. Further assume that the process admits another $k$th order subsampling representation $(\A', I_p, p'_E;k)$. If assumptions A1 and A2 hold, the following statements are true.
\begin{enumerate}
\item $\A'$ can be represented as $\A = \A D_1$, where $D_1$ is a diagonal matrix with $1$ or $-1$ on the diagonal. If we constrain the self influences to be positive, represented by the diagonal entries, then $\A' = \A$.
\item If A3 also holds, then $\A' = \A$.
\end{enumerate}   
\end{theorem}

\subsection{Complete identifiability of general SVAR model}

For identifiability of the full SVAR model under subsampling, we require two additional assumptions: 
\begin{description}
\item[A4] The variance of each $p_{e_j}$ is equal to one, i.e., $\Lambda = I_p$.
\item[A5] $\C$ is full rank.
\end{description}
Assumption A4 is common in  SVAR models due to the inherent non-identifiability between scaling the $e_{tj}$s and scaling the columns of $\C$. Assumption A5 is mild, and contains the more restrictive assumptions in non-Gaussian SVAR models (necessary to infer the instantaneous structural DAG \citeU{shimizu:2006}), that $\C$ may be row and column permuted to a lower triangular matrix with non-zeros on the diagonal. Under these assumptions, we have the following identifiability result for general subsampled SVAR models: 

\begin{theorem} \label{mixed_corr_theorem}
 Suppose $e_{tj}$ are all non-Gaussian and independent, and that the data are generated by a SVAR(1) process with representation $(\A, \C,p_e;k)$ that also admits another subsampling representation $(\A', \C',p'_e;k)$. If assumptions A1,A2 and A4 hold, the following statements are true:
 \begin{enumerate}
 \item $\C$ is equal to $\C'$ up to permutation of columns and scaling of columns by $1$ or $-1$, i.e. that $\C' = \C P$ where $P$ is a scaled permutation matrix with $1$ or $-1$ elements. This implies $\Sigma = \C\C^T = \C' \C'^T = \Sigma^T$.
 \item If A3 and A5 also hold then $\A$ is equal to $\A'$.
 \end{enumerate}
 \end{theorem}
 
 The requirement that $\C$ be full rank is due to the structure of $\Ll$. Since one may identify $\C$ as the first $p$ columns of $\Ll$, to obtain $\A$ we must premultiply the second set of $p$ columns of $\Ll$ by $\C^{-1}$. The asymmetry assumption is needed since the scaling of the columns of $\C$ and $\A \C$ by 1s or -1s is ambiguous if the distributions are symmetric;  the asymmetry assumption ensures that the unit scalings are identifiable. See the Appendix for a full proof.  

If the instantaneous causal effects follow a directed acyclic graph (DAG), we may identify the DAG structure without any prior information about causal ordering of the variables in the DAG. 
 \begin{corollary} \label{corcor}
 Suppose assumptions A1, A2, and A4 hold. Suppose also that the true SVAR process corresponds to a DAG $G$, i.e. it has a lower triangular structural matrix $\C$ with positive diagonals, and it also admits another representation with structural matrix $\C'$. Then $\C = \C'$. This implies that the structure of $G$ is identifiable without prior specification of the causal ordering of $G$. 
 \end{corollary}
 This result follows from the fact that $\C$ may be identified up to column permutations. Based on the identifiability results of \citeU{shimizu:2006}, if $\C$ follows a DAG structure, it may be row and column permuted to a unique lower triangular matrix. The row permutations identify the causal ordering, and the nonzero elements below the diagonal identify the edges in $G$. See \citeU{shimizu:2006} for more details on identifiability and estimation of the DAG from $\C$. 
 
 Taken together, the results of Theorem \ref{mixed_corr_theorem} and Corollary \ref{corcor} imply that when the shocks, $e_t$, are independent and non-Gaussian, a complete causal diagram of the lagged effects and the instantaneous effects are fully identifiable from the subsampled time series, $\tilde{\bff X}$.

\section{Mixed Frequency SVAR} \label{mf}
Estimation and forecasting of mixed frequency time series are commonly approached using both standard VAR and SVAR models \citeU{foroni:2013,schorfheide:2015}. 
Typically, the VAR model is fit at the same scale as the fastest sampled time series, setting C in Figure~\ref{sampling_types}. Due to costly data collection, especially for large macroeconomic indicators like GDP, this scale is generally arbitrary and may not reflect the true causal dynamics, leading to confounded Granger and instantaneous causality judgements \citeU{zhou:2014,breitung:2002}. In particular, if the true causal time scale, or one of interest to an analyst, is at a lower rate as in setting D in Figure~\ref{sampling_types}, then a causal analysis at the observed rate will run into the same problems as those for the single frequency subsampling case as discussed in Section \ref{cs}. We provide an example at the end of Section \ref{mfsvar}.

Identifiability conditions for mixed frequency VAR models with no subsampling at the fastest scale (Figure~\ref{sampling_types}B) was an open problem for many years \citeU{chen:1998} . Anderson et al. \citeU{anderson:2015}  recently showed the mixed frequency VAR (MF-VAR) of type B in Figure~\ref{sampling_types} is \emph{generically} identifiable from the first two observed moments of the MF-VAR, meaning that unidentifiable models make up at most a set of measure zero of the parameter space. However, no explicit identifiability conditions of the VAR process were given.

In this section, we generalize our identifiability results from Section \ref{subsamp_sec} to the mixed frequency case with arbitrary levels of subsampling for each time series. Our analysis indicates that Granger and instantaneous causal effects can be accurately estimated from mixed frequency time series. Specifically,  we use the results from Section \ref{subsamp_sec} to provide explicit identifiability conditions for MF-SVAR models under arbitrary subsampling (cases B, C, and D in Figure~\ref{sampling_types}) with non-Gaussian error assumptions. Together, our framework provides a unified way of deriving explicit identifiability conditions for both subsampling and mixed frequency cases. We note that while case C in Figure~\ref{sampling_types} is a subsampled version of the standard mixed frequency case, our results also cover mixed frequency subsampling like case D. To our knowledge, this is the first identifiability result for subsampled mixed frequency cases like C and D. 
\subsection{Mixed Frequency SVAR} \label{mfsvar}
For simplicity of presentation, we assume each time series in $x_t \in \mathbb{R}^p$ is sampled at one of two sampling rates, slow subsampling rates $k_s$ and fast subsampling rates $k_f$. We then write $x_t = (x_t^s, x_t^f)$ where $x_t^s$ are those series subsampled at $k_s$ and $x_t^f$ are those subsampled at $k_f$. Let ${\bff k} \in \{k_s,k_f\}^p$ be the list of subsampling rates for each time series. In Figure~\ref{sampling_types}B, $k_f = 1$ and $k_s = 2$, whereas in Figure~\ref{sampling_types}C, $k_f = 2$ and $k_s = 4$. Analogous to the subsampled case, we refer to a parameterization of a MF-SVAR model as $(\A, \C, p_e;{\bff k})$, where ${\bff k}$ is now a $p$-vector. Let $k^{*}$ be the smallest multiple of both $k_s$ and $k_f$; for example, in Figure~\ref{sampling_types}C, $k^* = 4$. 


We may derive a similar representation to Eq. (\refeq{svar_sub}) for mixed frequency series. Fix a time point $t$ such that all series are observed.  Let $I^{(q)}$ be a modified  $p \times p$ identity matrix where all rows $i$ such that $x_{ti}$ is not observed at time $t - q$ are set to zero. Further, let $I^{(\bar{q})} = I - I^{(q)}$, $A^{(q)} = I^{(q)} A$, and $A^{(\bar{q})} = I^{(\bar{q})} A$. Then
\begin{align}
x_t &= \A x_{t-1} + \C e_t \nonumber \\
&= \A I^{(1)} x_{t-1} + \A I^{(\bar{1})} x_{t-1} + \C e_t \nonumber \\
&= \A I^{(1)} x_{t-1} + \A(\A^{(\bar{1})} x_{t-2} +  \C^{(\bar{1})} e_{t-1}) + \C e_t \nonumber \\
&= \ldots \nonumber \\
& =  \F \tilde{x}_{t-1} + \Ll \tilde{e}_t \label{mf_rep}, 
\end{align}
where 
\begin{align}
&\F = (\A , \A \A^{(\bar{1})}, \ldots, \A \A^{(\bar{1})}\ldots\A^{(\overline{k^* - 1})}), \,\,\, 
\Ll = (\C, \A \C,\A \A^{(\bar{1})}\C,\ldots,\A \A^{(\bar{1})}\ldots \A^{(\overline{k^* - 1})}\C), \nonumber \\ 
&\tilde{x}_{t-1} = (I^{(1)} x_{t-1}, \ldots,I^{(k)} x_{t-k^*}) \nonumber,\,\ \text{and  } \tilde{e}_t = (e_t, I^{(1)} e_{t-1}, \ldots, I^{(k^* - 1)} e_{t-k^* - 1}). \nonumber
\end{align}
Eq. (\refeq{mf_rep}) takes the same form as Eq. (\refeq{svar_sub}), namely some matrix $\F$ times the observed time series samples $\tilde{x}_{t-1}$ plus a matrix $\Ll$ times a vector of non-Gaussian errors $\tilde{e}_t$. This intuitively suggests that similar identifiability results will hold. 

In a \emph{subsampled} mixed frequency setting where the fastest rate is greater than one (Figure \ref{sampling_types} C), not accounting for subsampling may lead to not only the same kind of mistaken inferences as discussed in Section \ref{cs}, but also to some further mistakes unique to the mixed frequency case.

\begin{example}
Consider a subsampled mixed frequency SVAR process generated by Eq. (\ref{mf_rep}) with the same $(\A, \C)$ parameters given by Example \ref{ex1}. Suppose subsampling is not taken into account and $\tilde{\bff X}$ is analyzed instead as a classical mixed frequency series (case B) using MF-VAR methods based on the first two moments \citeU{anderson:2015}. Consider two cases:
\begin{enumerate}
    \item[] Case 1: true sampling rate is ${\bff k} = (2, 4)$. In this case, if $\tilde{\bff X}$ is analyzed at the rate $(1,2)$ using the first two moments, then $\A$ and ${\bff \Sigma}$ are not identifiable at this rate since both off diagonal elements of $\A$ are zero \citeU{anderson:2015}. Thus, no inference of both the instantaneous correlations and lagged effects are even possible.
    \item[] Case 2: true sampling rate is ${\bff k} = (2, 6)$. In this case, if $\tilde{\bff X}$ is analyzed at the rate $(1,3)$ using the first two moments, the estimated $\A$ and covariance ${\bff \Sigma}$ will be the same as that in Example \ref{ex1} \citeU{anderson:2015}, leading to an incorrect inference that there is an instantaneous effect but not any directed lagged effect.  
\end{enumerate}
\end{example}

\subsection{Identifiability of MF-SVAR}
We provide generalizations of both Theorems \ref{Gong_theorem} and \ref{mixed_corr_theorem} to the mixed frequency case. 

\begin{theorem} \label{Gong_theorem_mf}
Suppose the $e_{ti}$ are non-Gaussian and independent for all $t$ and $i$, and that the data $\tilde{x}_t$ are generated by Eq. (\refeq{svar_f}) with $\C = I_p$. Further suppose that the process also admits another mixed frequency subsampling representation $(\A', I_p, p'_e;{\bff k})$. If assumptions A1 and A2 hold, the following statements are true.
\begin{enumerate}

\item $\A'$ can be represented as $\A' = \A D_1$, where $D_1$ is a diagonal matrix with $1$ or $-1$ on the diagonal. \label{i1}
\item If any multiple of $k_i$ is $1$ smaller than some multiple of $k_j$, then $\A_{ij} = \A'_{ij}$. If $\A_{ij} \neq 0$ this implies that $(D_1)_{jj} = 1$, i.e. the $j$th columns of $\A$ and $\A'$ are equal: $\A_{:j} = \A'_{:j}$.
\item If each $p_{e_i}$ is asymmetric, we have $\A' = \A$. \label{i3}
\end{enumerate}   
\end{theorem}
\begin{proof} Points \ref{i1}. and \ref{i3}. follow since we may further subsample all series in $x_t$ to a subsampling rate of $k^*$. This gives a subsampled ${\bff \tilde{X}}$ with representation $(\A, I, p(e) ;k^*)$. Applying Theorem \ref{Gong_theorem} gives the result. Furthermore, we note that if some multiple of $k_i$ is one less than some multiple of $ k_j$, then there exists a set of $t$s for Eq. (\refeq{mf_rep}), where series $i$ is observed at time $t - 1$ and series $j$ is observed at time $t$. By identifiability of linear regression, $A'_{ij} = A_{ij}$. This resolves the sign ambiguity of the columns in \ref{i1}, so that $\A_{:j} = \A'_{:j}$.
\end{proof}

\begin{theorem} \label{mixed_corr_theorem_mf}
 Suppose the $e_{ti}$ are non-Gaussian and independent for all $t$ and $i$, and that the data is generated by an SVAR(1) process with representation $(\A, \C,p_e; {\bff k})$ that also admits another mixed frequency subsampling representation $(\A', \C',p'_e; {\bff k})$. If assumptions A1, A2, and A4 hold, the following statements are true:
 \begin{enumerate}
 \item $\C$ is equal to $\C'$ up to permutation of columns and scaling of columns by $1$ or $-1$, ie $\C' = \C P$ where $P$ is a scaled permutation matrix with $1$ or $-1$ elements. This implies that $\Sigma = \C\C^T = \C' \C'^T = \Sigma'$. \label{one}
 \item If $\C$ is lower triangular with positive diagonals, i.e. the instantaneous interactions follow a DAG, and if for all $i$ there exists a $j$ such that any multiple of $k_i$ is $1$ smaller than some multiple of $k_j$ with $A_{j:} C_{:i} \neq 0$, then $\A = \A'$. \label{2}
 \item If A3 and A5 also hold, then $\A  = \A'$. \label{three}
 \end{enumerate}
 \end{theorem} 
 The proofs of points \ref{one} and \ref{three} follow the same subsampling logic as the proof given for Theorem \ref{Gong_theorem_mf}. The proof of point \ref{2} is given in the Appendix. 

Taken together, Theorems \ref{Gong_theorem_mf} and \ref{mixed_corr_theorem_mf} demonstrate that identifiability of SVAR models still holds for mixed frequency series with subsampling under non-Gaussian errors. Note that points 1 and 3 in both Theorem \ref{Gong_theorem_mf} and Theorem \ref{mixed_corr_theorem_mf} are the same as their subsampled counterparts; point 2 in both Theorems shows how the mixed frequency setting provides additional information to resolve parameter ambiguities in the non-Gaussian setting. Specifically, in the SVAR(1) model when there is one time step difference between when series $x_j$ and $x_i$ is sampled, then $A_{ij}$ is identifiable. We can then use this information to resolve sign ambiguties in columns of $\A$, which leads to point 2 in both Theorems \ref{Gong_theorem_mf} and \ref{mixed_corr_theorem_mf}. This result applies directly to the standard mixed frequency setting \citeU{anderson:2015,schorfheide:2015} where one series is observed at every time step (Fig. \ref{sampling_types} B). It also applies to case D since there exists certain time steps where one series is observed directly before a latter series.
\section{Estimation}
We take a model-based approach to estimation. Specifically, we model the non-Gaussian error terms as a mixture of Gaussians with $m$ components. This approach has been used widely in econometrics and other fields as a flexible and tractable way of modeling non-Gaussianity in innovations \citeU{lanne:2015, Gong:2015}. Formally, we assume that $e_{tj}$ is drawn from the mixture distribution:
\begin{align}
z_{tj} \sim \text{Categorical}(\pi_j), \,\,\,\,\,e_{tj} \sim \mathcal{N}(\mu_{j z_{tj}}, \sigma^2_{j z_{tj}}) \nonumber
\end{align} 
 where $\mu_j, \sigma^2_j$ and $\pi_j$ are length $m$ vectors specifying the mean, variance, and mixing weight of each mixture component. The $z_{tj}$ component indicators are auxilliary variables introduced to facillitate tractable inference. Together the full set of parameters for the non-Gaussian structural VAR model is given by $\Theta = (\A,\C,\mu,\sigma^2,\pi)$ where $\mu, \sigma^2, \pi$ concatenate the mixture parameters of the errors across series. For example, $\mu_{ji}$ is the mean of the $i$th mixture component for the $j$th error distribution, and likewise for $\sigma^2$ and $\pi$.

\subsection{EM algorithm}

We develop an EM algorithm for joint maximum likelihood estimation of the full set of parameters $\Theta$ based only on the observed subsampled/mixed frequency data $\tilde{X}$. Importantly, our method is the same for both subsampled and mixed frequency data, unlike that of \cite{Gong:2015}, which is tailored specifically to the subsampled case. Furthermore, the VAR-specific (i.e. $\C = I$) EM algorithm of \cite{Gong:2015} introduces auxiliary noise terms to facilitate inference, rendering their resulting algorithm non-exact; in constrast, our algorithm introduces no such approximations. Since the the log-likelihood surface is non-convex, we employ multiple random restarts to avoid poor local optima. For the subsampled case, the local optima problem is particularly severe due to the nonidentifiability under the first two moments, implying that many $(\A, \C)$ parameter values tend to do a decent job at approximately fitting the data. Finally, the basic EM algorithm also suffers from slow convergence due to the large amount of missing data. To ameliorate this problem, we deploy the adaptive-overrealxed EM method \citeU{Salakhutdinov:2003}.

Let $\W = \C^{-1}$. Further, let $z_{tji} = 1$ if error $e_{tj}$ was generated by mixture component $i$ and $z_{tji} = 0$ otherwise. The complete log-likelihood of the SVAR(1) model with mixture of normal errors may be written as:
\begin{align}
\log p(X_{1:T},z_{1:T}|\Theta) = T \log |\W| + \sum_{t = 1}^T \sum_{j = 1}^p \sum_{i = 1}^m z_{tji} \left(  \log \pi_{ji} - \frac{1}{2}\log 2 \pi \sigma_{ji}^2 -\log \frac{\left(\W_j x_t - \W_j \A x_{t-1} + \mu_{ji} \right)^2}{2\sigma^2_{ji}} \right), 
\end{align}
where $\W_j$ is the $j$th row vector of $\W$. The EM algorithm alternates between the $E$-step, where we compute the conditional expectation $E\left(\log p(X_{1:T},z_{1:T}|\Theta) | \tilde{X} \right)$, and the $M$-step, where that expectation is maximized with respect to the parameters $\Theta$. We first provide the specific updates in the $M$-step, and then explain how the partiular conditional expectations used in the $M$-step are computed using a Kalman filter.

\subsection{M-step}
In the M-step, we maximize the expected complete log-likelihood conditional on the observed data, \\
$E\left(\log p(X_{1:T},z_{1:t}|\Theta) | \tilde{X}\right)$, with respect to $\Theta$. We perform this maximization via coordinate ascent, cycling through $\A$, $\W$, and $(\mu, \sigma^2, \pi)$ until convergence. The specific updates are given below.
\begin{description}
\item[$\bullet$] A update: Each row of $\A$, $\A_j$,  may be updated independently,
 \begin{align}
 \hat{\A}_j = \left(\sum_{t = 1}^T \sum_{i = 1}^m \frac{E(z_{tji} x_{t - 1} x_{t-1}^T | \tilde{X})}{\sigma^2_{ji}} \right)^{-1} \left( \sum_{t = 1}^T \sum_{i=1}^m \frac{- \mu_{ji} E(z_{tji} x_{t - 1}|\tilde{X}) + E(z_{tji} x_{t-1} x_t^T | \tilde{X}) \W_j^T}{\sigma^2_{ji}} \right).
 \end{align}

\item[$\bullet$] $\mu$, $\sigma^2$, and $\pi$ update:  \,\, These may be optimized jointly in one step using

 \begin{align}
 \hat{\mu}_{ji} = \frac{\sum_{t = 1}^T E(z_{tji} x_t|\tilde{X}) - \W_j \A E(z_{tji} x_{t-1} | \tilde{X})}{\sum_{t = 1}^T  E(z_{tji}|\tilde{X})},  \,\,\,\,\,\, \hat{\pi}_{ji} = \frac{\sum_{t = 1}^T E(z_{tji}|\tilde{X})}{T} \nonumber
 \end{align}
 \vspace{-.2in}
 \begin{align}
 \hspace{-.2in}\hat{\sigma}^2_{ji} = \frac{1}{\sum_{t = 1}^T  E(z_{tji}|\tilde{X})} \bigg( \sum_{t = 1}^T \W_j E(z_{tji} x_t x_t^T | \tilde{X}) \W_j^T + \W_j^T \A E(z_{tji} x_{t-1} x_{t-1}^T | \tilde{X} ) \A^T \W_j^T  + \hat{\mu}_{ji}^2 E(z_{tji}|\tilde{X}) \nonumber \\
  - 2 \mu_{ji} \W_j E(z_{tji} x_t | \tilde{X})  
 - 2 \W_j E(z_{tji} x_t x_{t-1}^T) \A^T \W_j^T + 2 \mu_{ji} \W_j \A E(z_{tji} x_{t-1})\bigg) \nonumber
 \end{align}

 \item[$\bullet$] W update: \,\,\,
 The maximization with respect to $\W$ is not given in closed form. Instead, we utilize the Newton-Raphson method. Let ${\bff w} = \text{vec}(\W)$ be the column-wise vectorization of $\W$. At each step, the next ${\bff w}$ iterate is given by
 \begin{align}
 {\bff w}^{l + 1} = {\bff w}^l - H({\bff w}^l)^{-1} \nabla l({\bff w}^l)  
 \end{align}
where $l({\bff w}) = E(\log p(X_{1:T},z_{1:t}|\Theta) | \tilde{X})$ and $H({\bff w})$ is the Hessian of $l({\bff w})$ with respect to ${\bff w}$. We provide explicit expressions for the gradient and Hessian in the Appendix. 
\end{description}

\subsection{E-step} \label{e_step}
All conditional expectations in the $M$-step above are computed using the Kalman filtering-smoothing algorithm. For simplicity of presentation, consider only one block of data, so that $X = x_{1:t}$, where $x_1$ and $x_t$ are fully observed but $x_{t'}$, $1 < t' < t$, have some missing data, and hence are not included $\tilde{X}$.  Any subsampled/mixed frequency time series can be broken into independent blocks of this type. The conditional expectation $E(z_{tji} x_{t}x_{t-1}^T | \tilde{X})$ under the past parameter values can be computed by noticing that
\begin{align} \label{comb}
E(z_{tji} x_{t}x_{t-1}^T | \tilde{X}) = E_{z_{1:t}} \left( z_{tji} E_{x}\left(x_{t}x_{t-1}^T | \tilde{X}, z_{1:t} \right) \right).
\end{align}
Now, for a fixed $z_{1:t}$, $E_{x}\left(x_{t}x_{t-1}^T | \tilde{X}, z_{1:t} \right)$ may be computed using the Kalman filtering-smoothing algorithm since for fixed $z_{1:t}$, $\tilde{x}_t$ follows a linear Gaussian state-space model with latent observations $x_t$. Thus, to compute the expectations in Eq. (\refeq{comb}) we compute $E_{x}\left(x_{t}x_{t-1}^T | \tilde{X}, z_{1:t} \right)$ for each $z_{1:t}$ combination, then average them together weighted by $p(z_{1:t} | \tilde{X}) z_{tji}$. The $p(z_{1:t} | \tilde{X})$ terms used in the averaging step may be computed by:
\begin{align}
p(z_{1:t} | \tilde{X}) \propto p(\tilde{X} | z_{1:t}) p(z_{1:t})
\end{align} 
where $p(z_{1:t})$ is given by the prior mixture component weights, $\pi$, and $p(\tilde{X} | z_{1:t})$ is the likelihood of the observed data, which may also be computed by one pass of the Kalman filtering algorithm. This processes is repeated for all expectations in the $E$-step. The computational complexity of this exact EM algorithm scales as $2^{(k + 1)p}$, since the Kalman filter must be run for all combinations of $z_{1:t}$ for each block. The approximate EM algorithm of \cite{Gong:2015} has the same computational complexity. Similar to \cite{Gong:2015}, we have explored approximate inference methods based on variational EM and Markov Chain Monte Carlo (MCMC) methods but found their performance to be quite poor; we discuss this further in the Discussion.

\section{Simulations}

We investigate the estimation performance of the SVAR under subsampling. We simulate data with $p = 2$ time series and $m=2$ mixture components. The asymmetric error distributions are given by: $\pi_1 = (.7,.3)$, $\sigma_1 = (.2,1)$, $\mu_1 = (.36,-.84)$ for $e_{t1}$ and $\pi_2 = (.7,.3)$, $\sigma_2 = (.2,1)$, $\mu_2 = (-.36,.84)$ for $e_{t2}$. We look at two cases each for $\A$ and $\C$:
\begin{align}
\A^{(1)} &= \left(\begin{array}{c c}
.98 & 0 \\
.2 & .98 \end{array} \right) \,\, &\A^{(2)} &= \left(\begin{array}{c c}
.98 & .31 \\
-.31 & .98 \end{array} \right) ,\,\,\,\,\, \C^{(1)} &= \left(\begin{array}{c c}
1 & 0 \\
0 & 1 \end{array} \right) \,\, &\C^{(2)} &= \left(\begin{array}{c c}
1 & 0 \\
-.2 & 1 \end{array} \right)  \nonumber
\end{align}
Simulations are performed for two subsampling factors, $k \in \{2, 3\}$, and three sample sizes, $T \in \{205,403, 805 \}$. Note that due to subsampling, the actual sample sizes are reduced. Data from each parameter configuration is generated 10 times and the EM algorithm is run on each realization using 1000 random restarts. Box plots of the estimates of two scenarios are shown in Figures \ref{f1} and \ref{f4}. Similar plots for the other scenarios are shown in the Appendix.

We next investigate estimation performance in subsampling and mixed frequency sampling as a function of the signal to noise ratio. In these experiments we use $\A^{(1)}$ and $\C^{(2)}$. We scale $\A$ by a factor to set its maximum eigenvalue to the desired level. We perform these experiments for both full subsampling of $k = 2 \text{ and } 3$ and mixed frequency subsampling where one series is observed at every time point and the other is subsampled. Data from each parameter configuration is generated 40 times. In Figure \ref{eigen} we plot the average absolute error of estimating the $\A$ and $\C$ matrices as a function of the maximum eigenvalue of $\A$. Estimation under subsampling is stable until the maximum eigenvalue falls to about 0.6-0.5, and estimation becomes dramatically worse. The individual boxplots for 40 simulation runs per configuration are given in Figures \ref{box2} and \ref{box3} in the Appendix, where it is clear that many outliers in estimation appear for both $\A$ and $\C$ estimation at a maximum eigenvalue of $.5$. The results also show that the standard errors are also dramatically larger, further indicating unstable estimation in this regime. The increasing error in the estimation of A as a function of signal to noise ratio is also observed in the mixed frequency case. However, estimation remains stable and the variability of estimates increases less dramatically than in the subsampled case. This is partly due to the presence of significantly less local optima in the mixed frequency case. We further note that in the mixed frequency case, the error in $\C$ estimation appears to be constant across the maximum eigenvalue range we considered.

Unstable estimation arises from a combination of two factors. First, under subsampling, the transition matrix of the subsampled process is $\A^k$, indicating that the signal strength between observations scales exponentially as a function of subsampling. Furthermore, the likelihood surface is highly multi-modal where the other high probability modes all have approximately the same $\A^k$ value. As the signal to noise ratio falls, $\A^k$ estimation becomes more difficult due to subsampling, and thus the multimodal estimation becomes more severe leading to modes far from the true $\A$ occasionally obtaining higher likelihood. Overall, these simulations indicate that in the subsampling case there appears to be a threshold on the maximum eigenvalue, below which inference becomes unstable and unreliable.

We note that the simulations above cover cases A and B in Figure \ref{sampling_types}. Unfortunately, the computational complexity of the E-step of the EM algorithm forbids performing simulations in a reasonable time on cases C and D. Future work will explore computational speed ups to make inference in these cases tractable; see the discussion at the end of Section \ref{real_data}.

\begin{figure}
\centering
\includegraphics[width=.47\textwidth]{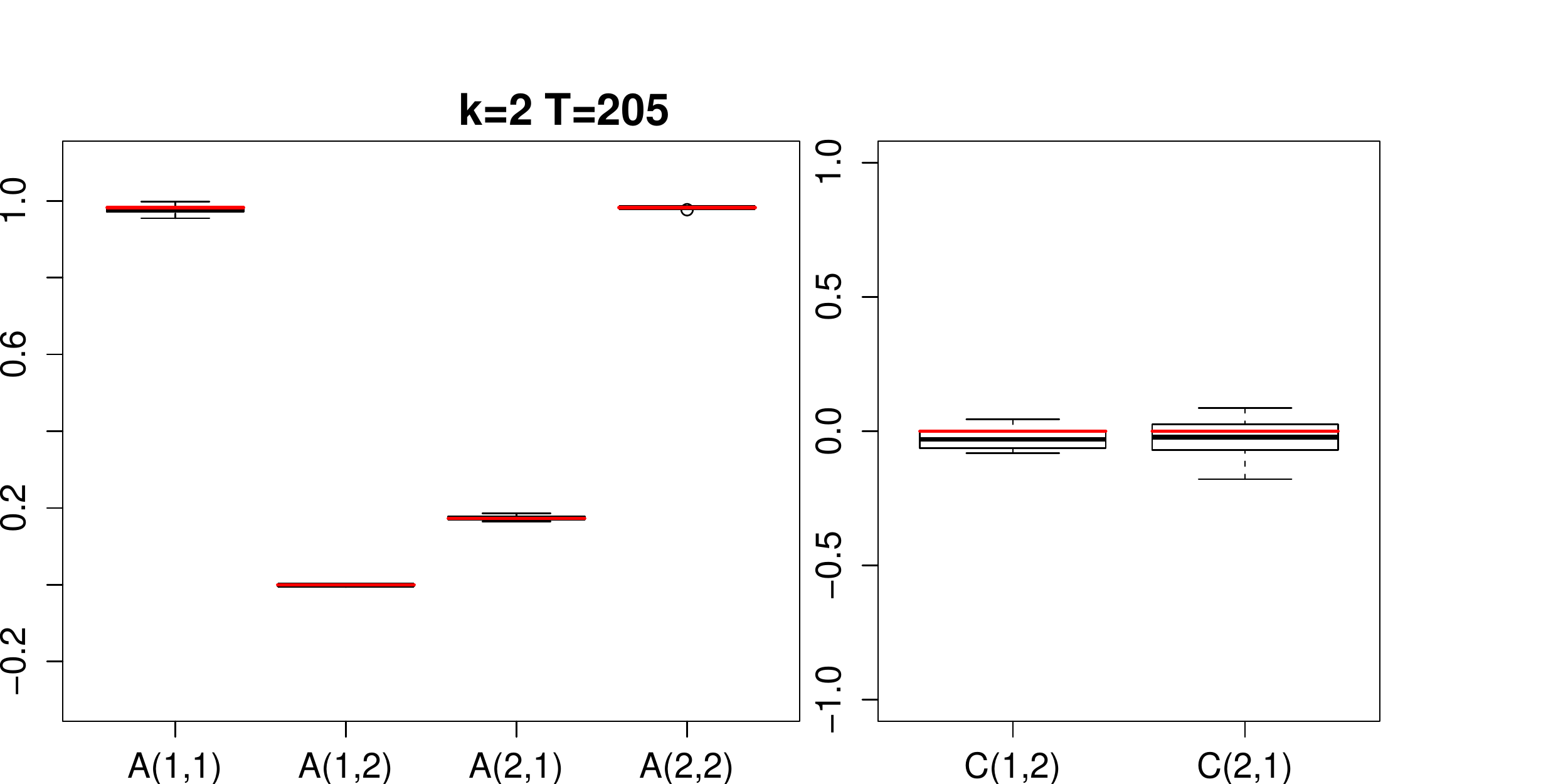} \hspace{.2in}
\includegraphics[width=.47\textwidth]{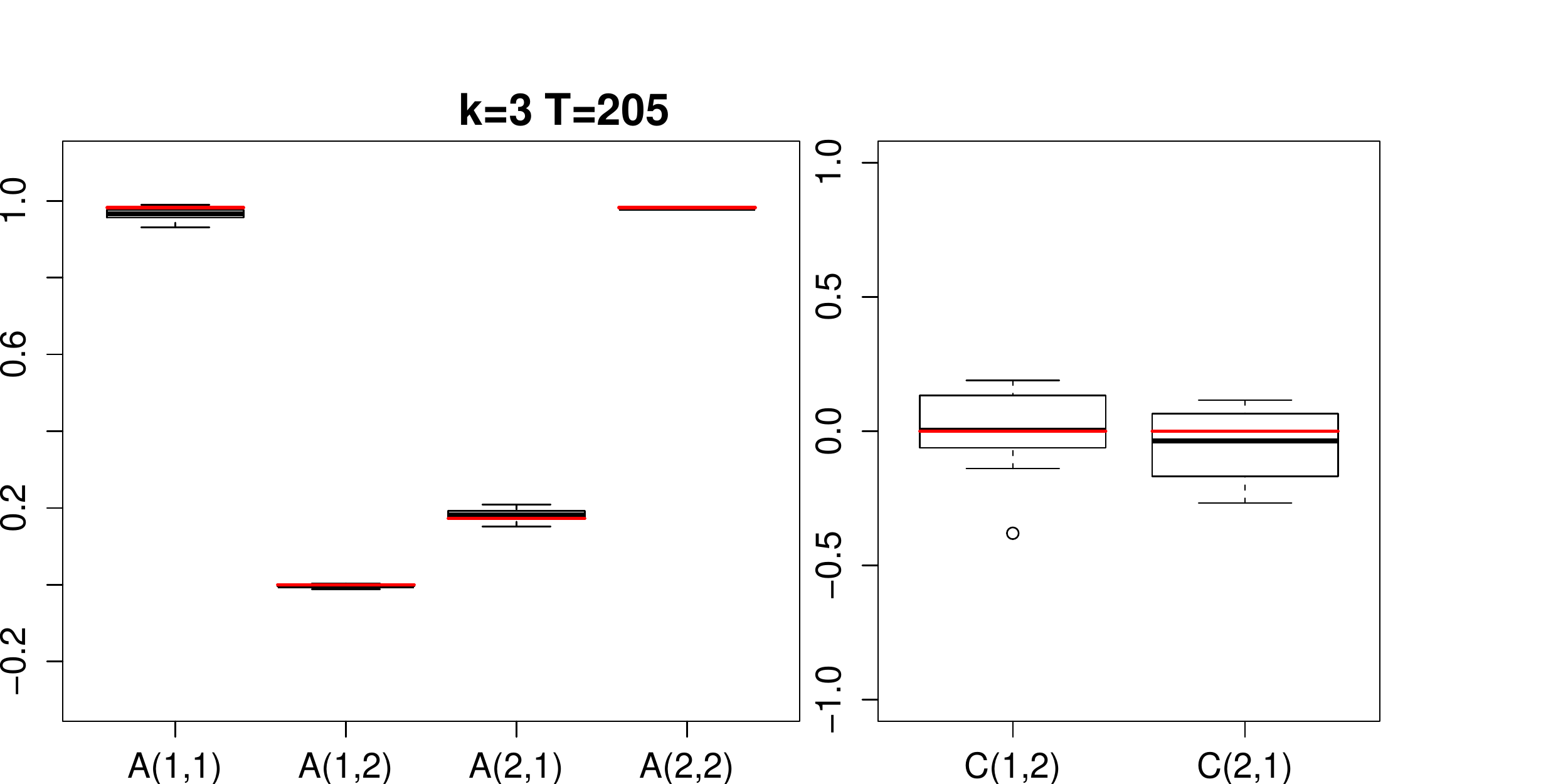}
\includegraphics[width=.47\textwidth]{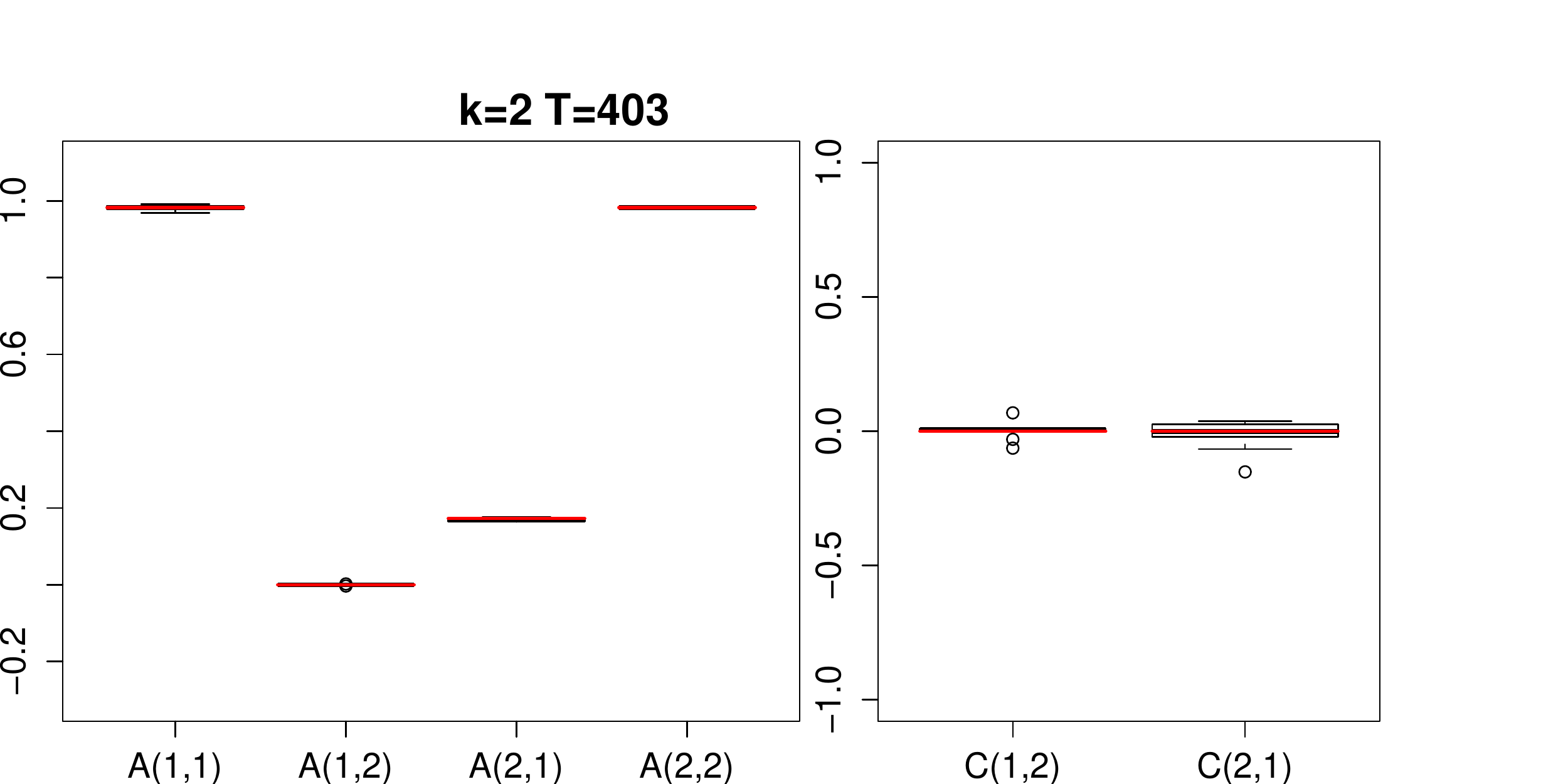} \hspace{.2in}
\includegraphics[width=.47\textwidth]{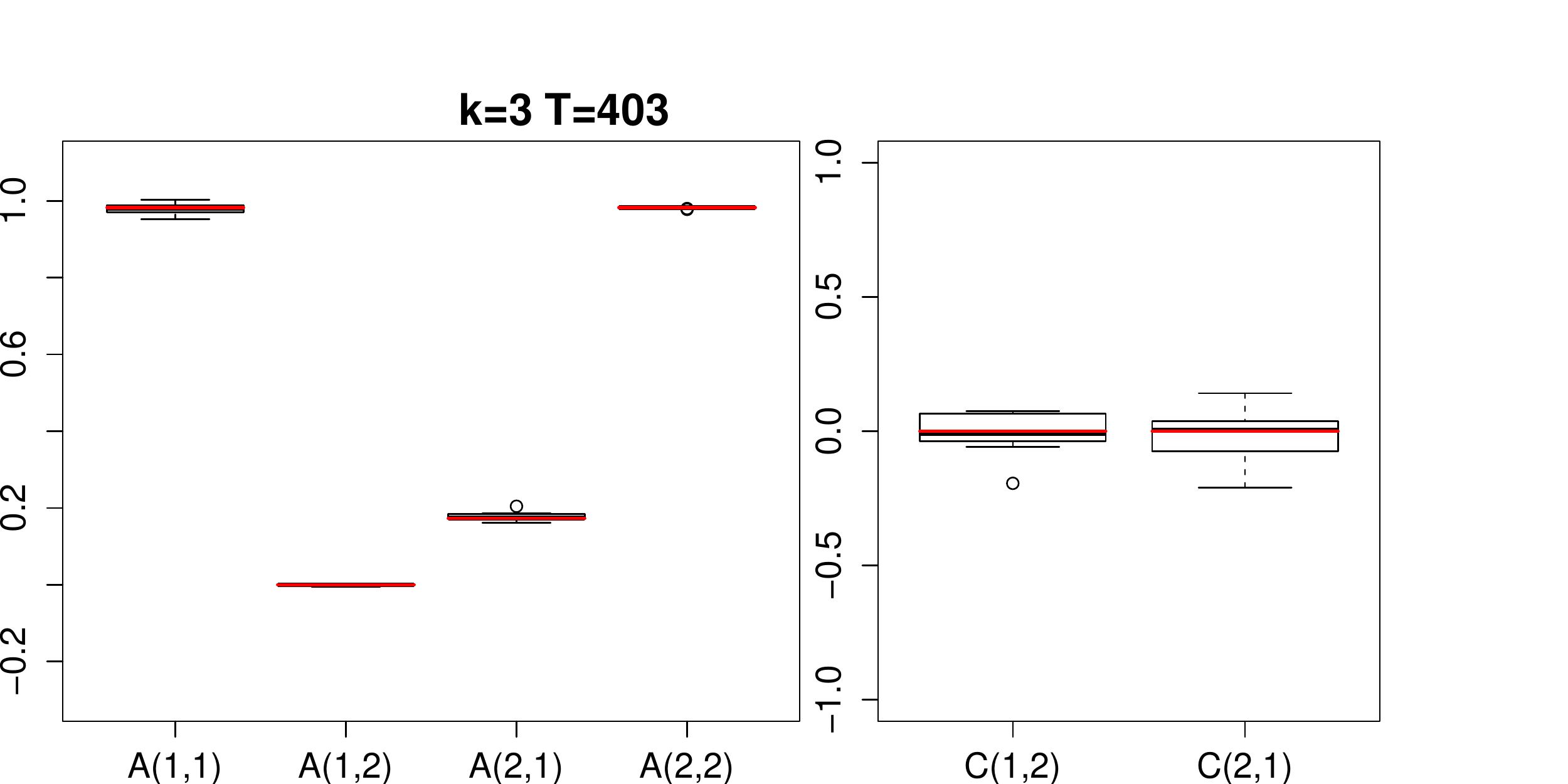}
\includegraphics[width=.47\textwidth]{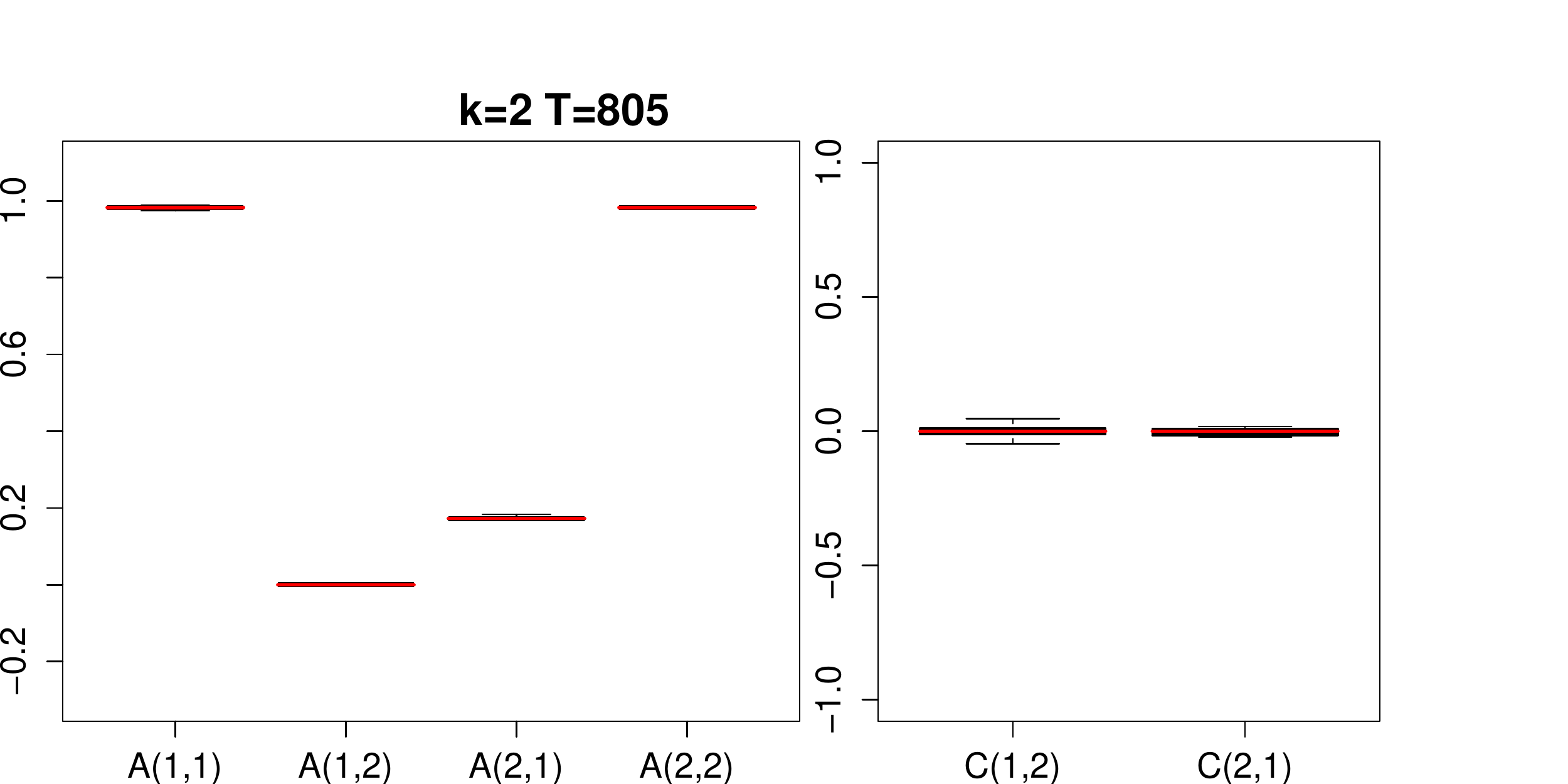} \hspace{.2in}
\includegraphics[width=.47\textwidth]{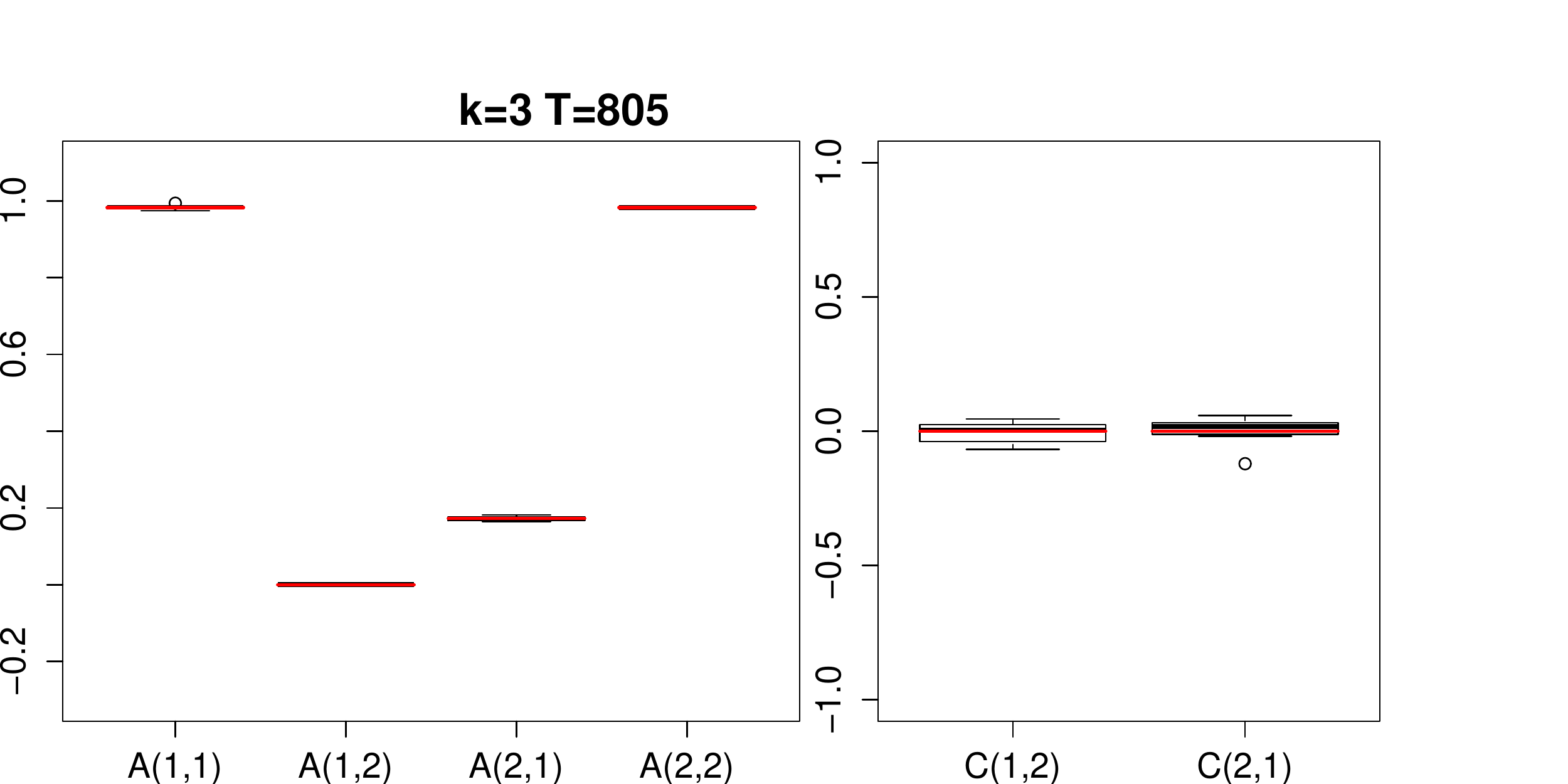}
\caption{{\bf Subsampled estimation performance simulations.} Histogram plots of $\A^{(1)}$ and $\C^{(1)}$ parameter estimates over 10 random data samplings. The original series is either of length $203$ (\emph{top}), $403$ (\emph{middle}) or $805$ (\emph{bottom}) and then subsampled at (\emph{left}) $k = 2$  and (\emph{right}) $k = 3$.}\label{f1}
\end{figure}

\begin{figure}
\includegraphics[width=.47\textwidth]{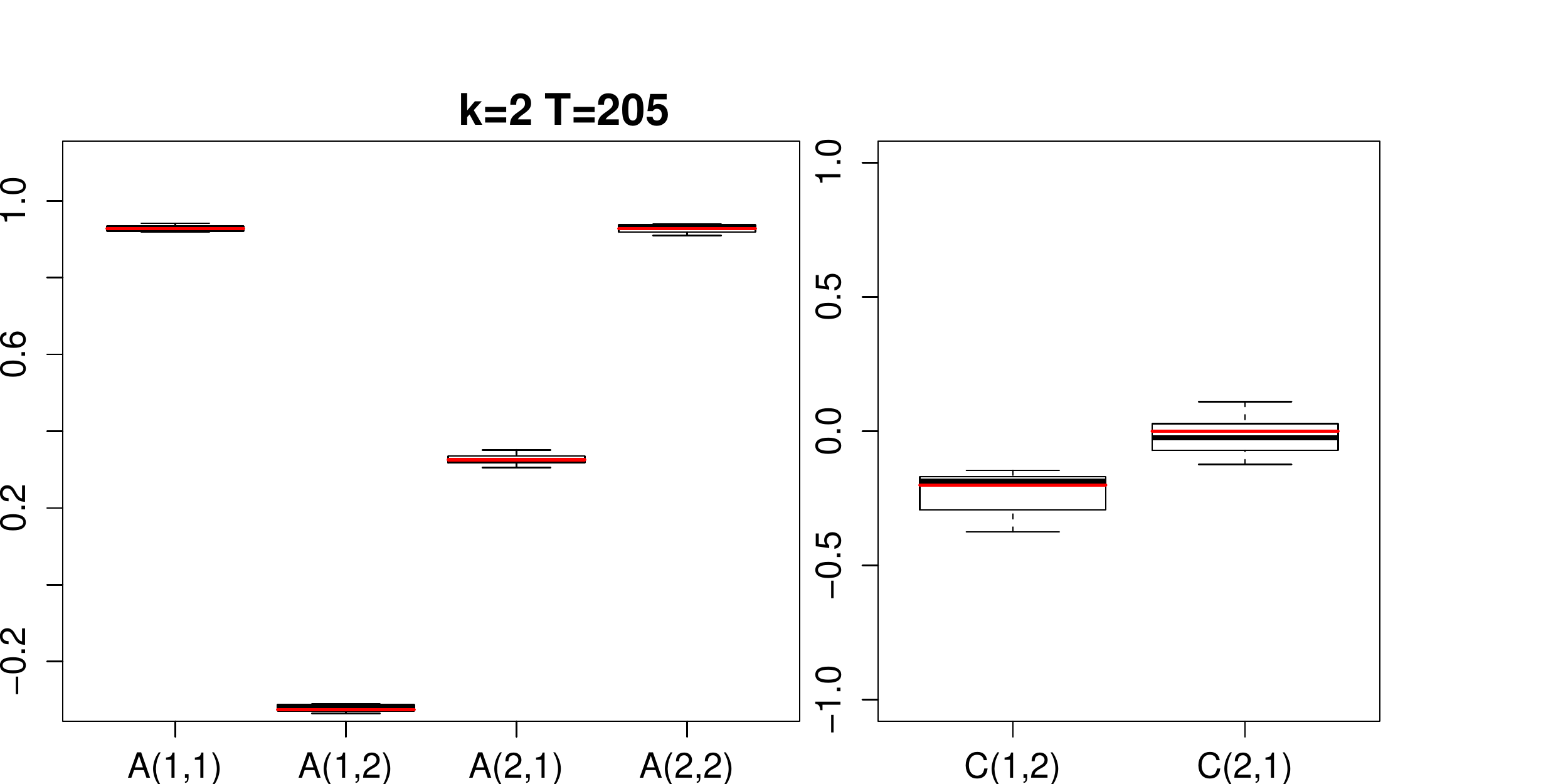} \hspace{.2in}
\includegraphics[width=.47\textwidth]{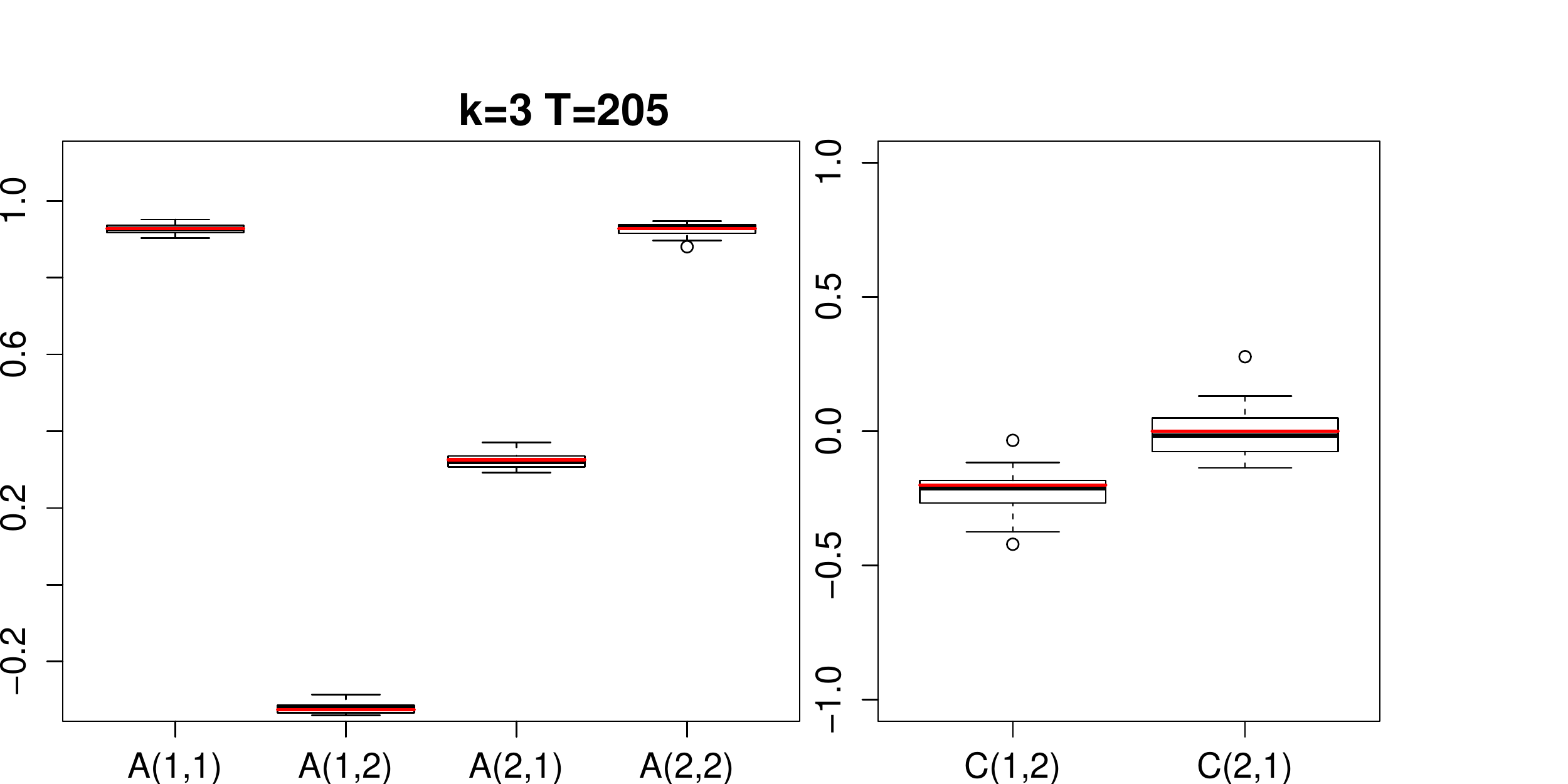}
\includegraphics[width=.47\textwidth]{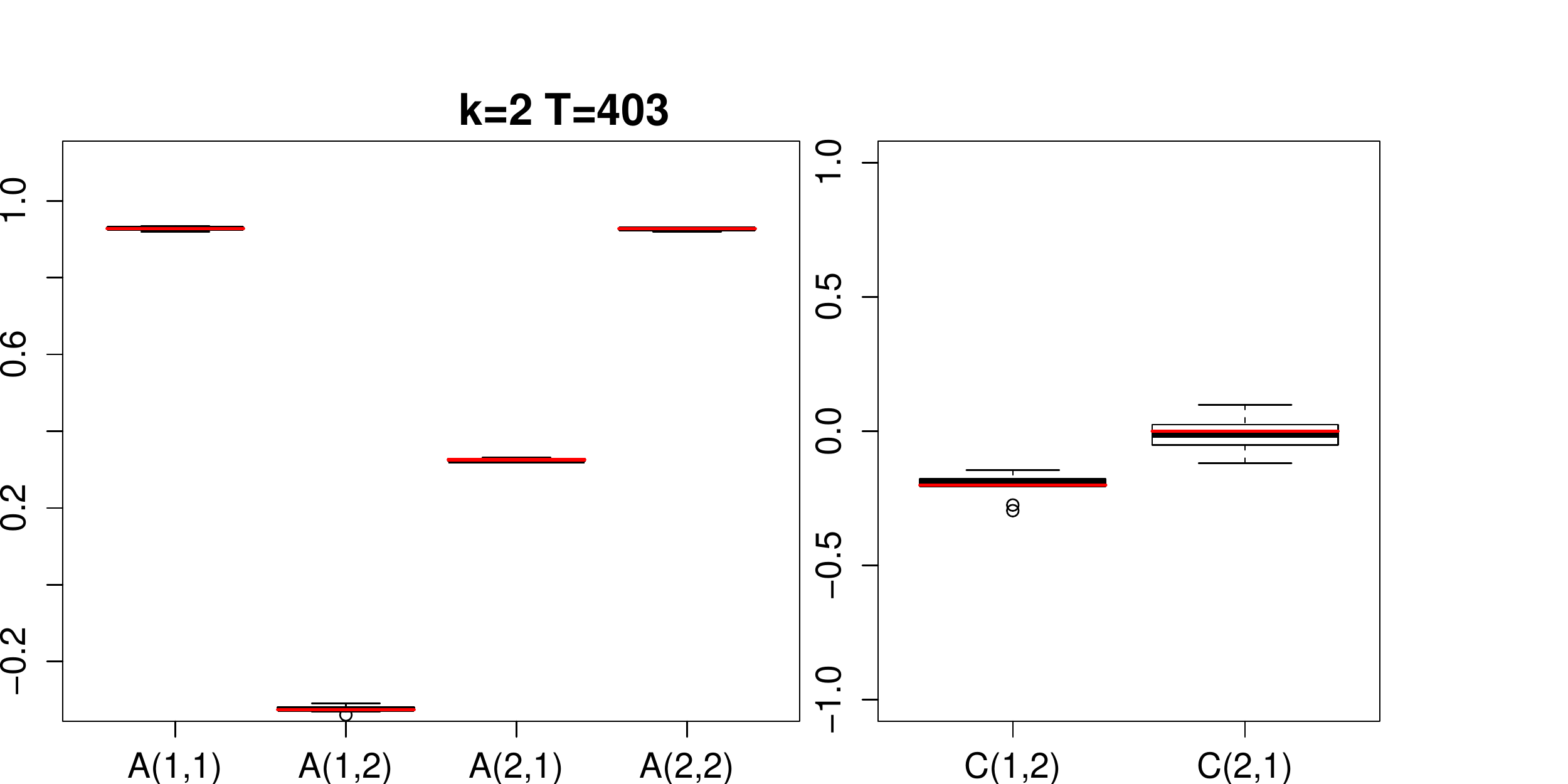} \hspace{.2in}
\includegraphics[width=.47\textwidth]{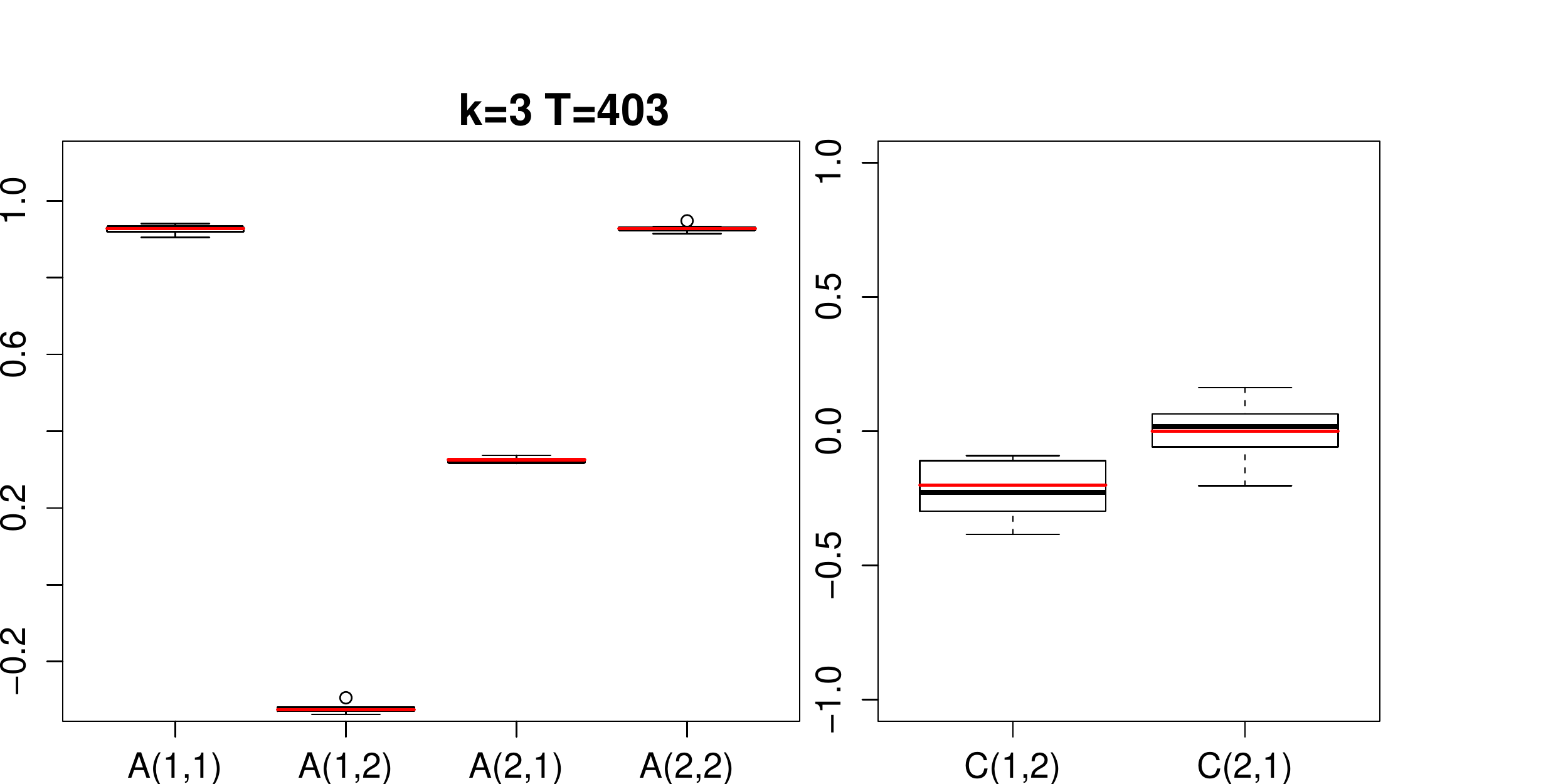}
\includegraphics[width=.47\textwidth]{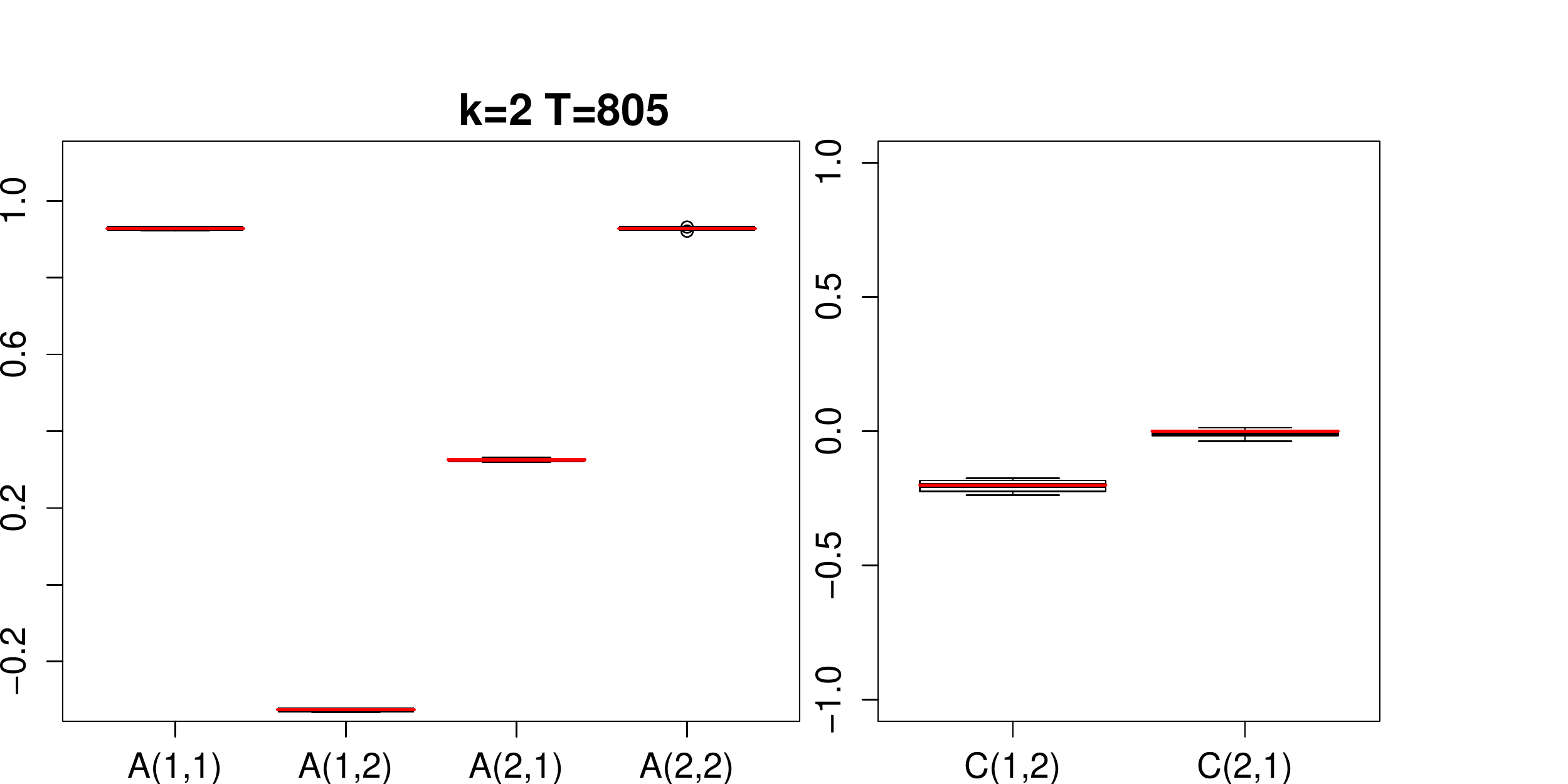} \hspace{.2in}
\includegraphics[width=.47\textwidth]{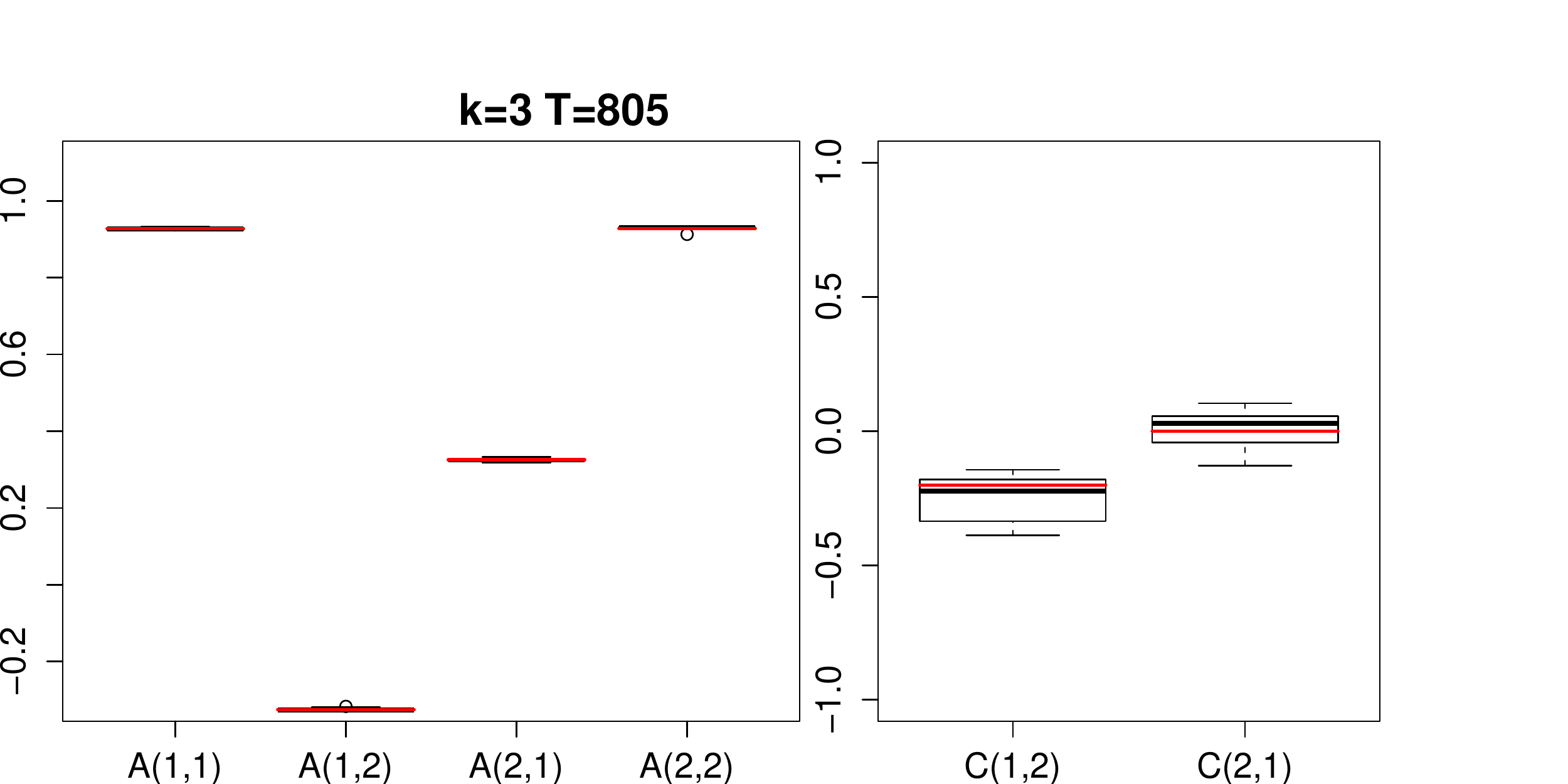}
\caption{{\bf Subsampled estimation performance simulations.} As in Fig. \ref{f1} for $\A^{(2)}$ and $\C^{(2)}$.}\label{f4}\end{figure}

\begin{figure}
\includegraphics[width=.52\textwidth]{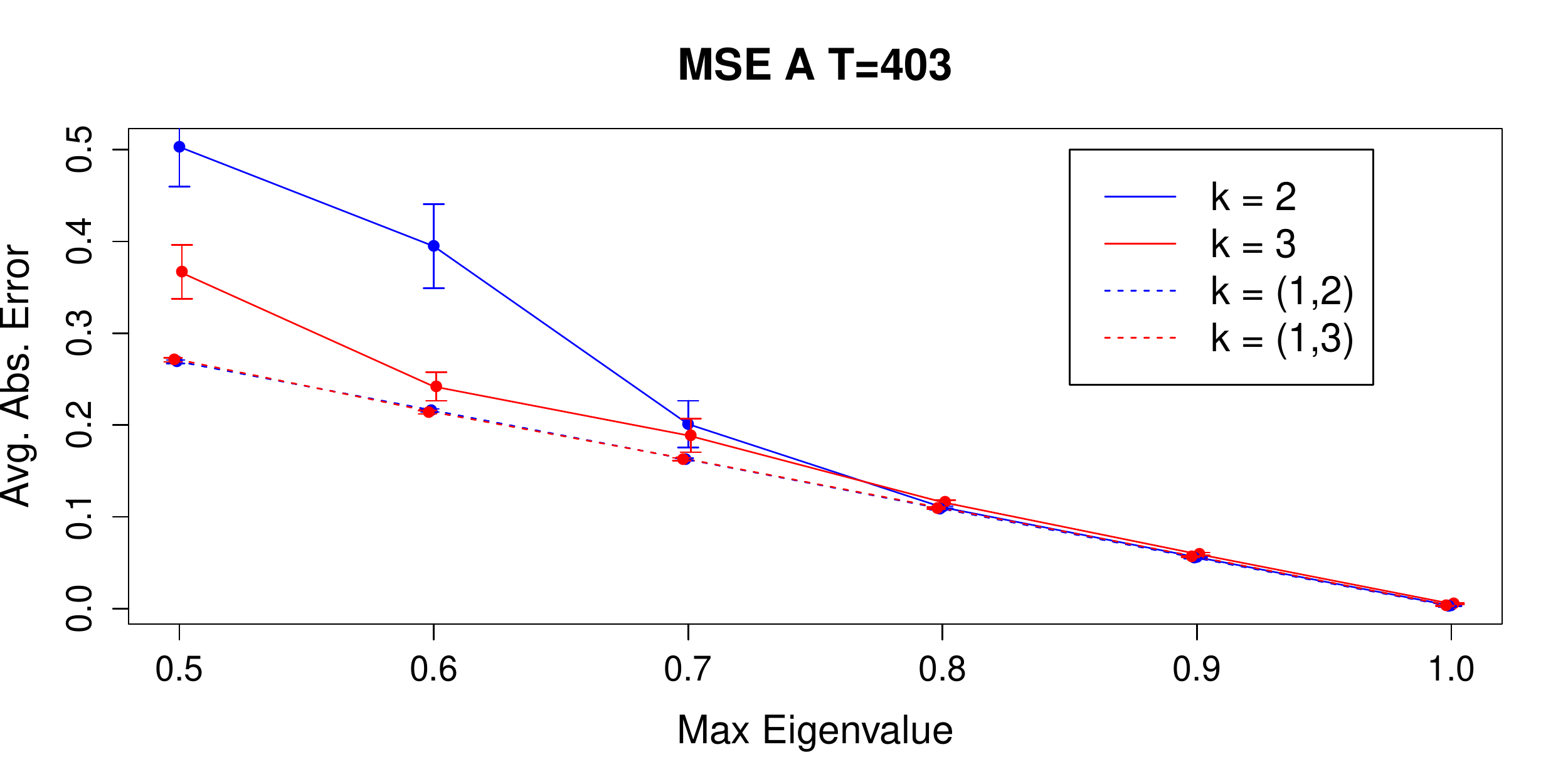}
\includegraphics[width=.52\textwidth]{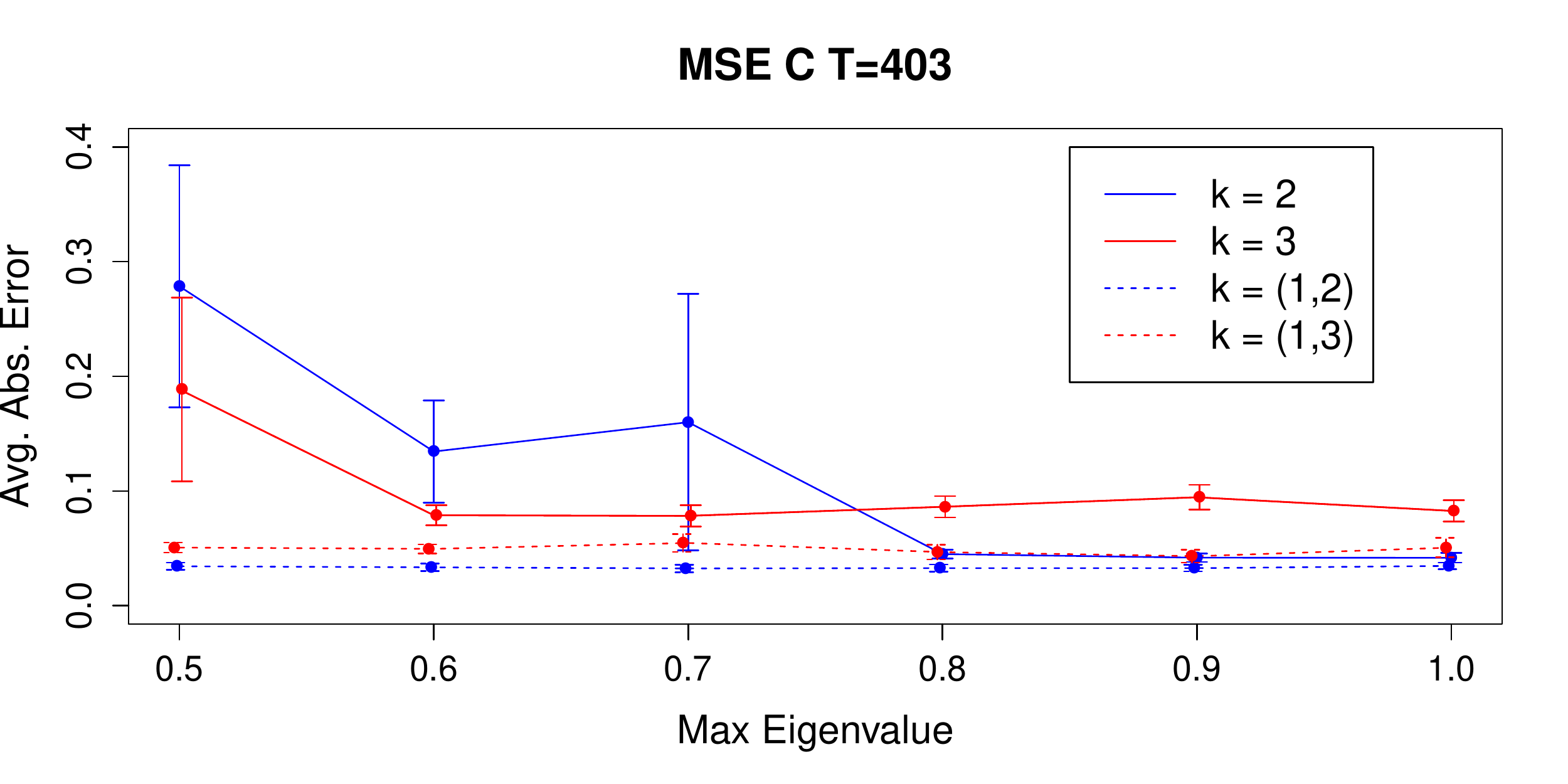}
\caption{{\bf Subsampled and mixed frequency signal to noise simulations.} Average MSE in estimation of $\A$ (\emph{left}) and $\C$ (\emph{right}) as a function of maximum eigenvalue of $\A$. Error bars indicate one standard error from 40 simulation runs.}
\label{eigen}
\end{figure}

\section{Real Data} \label{real_data}
\subsection{Subsampled Ozone Data}
We use the subsampled SVAR to analyze the causal scale and pathways in an ozone and temperature data set. The Temperature Ozone data is the 50th causal-effect pair from the website \url{https://
webdav.tuebingen.mpg.de/cause-effect/}, and was also considered in \citeU{Gong:2015}. The dataset consists of two time series, temperature and ozone concentration, sampled daily. First we standardize each time series to mean zero and unit variance. We fit the subsampled SVAR to the preprocessed series for $k = (1,2,3,4)$ subsampling regimes under both independent errors, $C = I$, and structural covariance in the instantaneous errors, $\C$ free. To ensure that good optima are found we perform 30,000 random restarts and run the adaptive-overrelaxed EM algorithm until the relative change in log-likelihood is less than $10^{-6}$.

 We first note that the estimated $\hat{\A}$ for $k = 1$ is given by $\hat{\A} = \left(\begin{array}{c c}
0.669 & 0.175 \\
-0.050 & 0.992
\end{array} \right)$, with maximum eigenvalue of $.962$, suggesting that accurate estimation of subsampled parameters is possible. The Bayesian Information Criterion (BIC) score for all models is displayed in Table \ref{ozone_table}. Across all subsampling rates, the structural model ($\C$ free) has substantially lower BIC, indicating that the two extra parameters of the structural model (off diagonal elements of $\C$) provide necessary flexibility. Furthermore, the best performing model is the structural matrix with subsampling rate $k = 2$. The transition matrix at $k = 2$ is given by $\hat{\A} = \left(\begin{array}{c c}
 0.849 & 0.058 \\
-0.027 & 0.981 
\end{array} \right)$, a similar result as that given by \citeU{Gong:2015} for $C = I$. After normalizing columns, we obtain a $\hat{\C} = \left(\begin{array}{c c}
1.00 & .206 \\
.29 & 1.00 
\end{array} \right)$ and an instantaneous error covariance of $\hat{\Sigma} = \hat{\C} \hat{\Lambda}(e_t) \hat{\C}^T = \left(\begin{array}{c c}
.1993 & .0535 \\
.0535 & .0539
\end{array} \right)$. Together, these results indicate the prevalence of relatively weak lagged effects at the subsampled scale, but stronger instantaneous effects between temperature and ozone. Furthermore, we see that the temperature time series obtains most of its power from a stronger error variance, while the ozone series is driven relatively more by the autoregressive component.  

\begin{table}
\centering
\begin{tabular}{c c c c c }
Model / k & 1 & 2 & 3 & 4 \\
$\C = I$ & 901.96 & 791.02 & 839.56 & 797.0066 \\
$\C$ free & 784.53 & {\bf 777.78} & 790.46 & 791.23 \\
\end{tabular}
\caption{BIC score for the SVAR model under different subsampling and covariance types on the Temperature Ozone Dataset.} \label{ozone_table}

\end{table}

\subsection{Mixed Frequency: GDP and Treasury Bonds}
We perfrom an SVAR analysis on the mixed frequency data set of quarterly Gross Domestic Product (GDP) and monthly price of treasury bonds (TB). The data set has been previously compiled and analyzed in the mixed frequency setting by \cite{schorfheide:2015} and is available on the author's website. We follow \citeU{schorfheide:2015} and log transform both quarterly GDP and monthly TB. Furthermore, as is common in mixed frequency analysis of econometric indicators \citeU{chen:1998,zadrozny:1990}, we compute first differences to remove first order non-stationarities. 

We fit the SVAR model to the preprocessed data at the monthly rate. In the traditional approaches to mixed frequency VARs analyses, $\A$ and the instantaneous covariance $\Sigma$ are generically identifiable from the first two moments \citeU{anderson:2015}. What sets our non-Gaussian approach apart in this mixed frequency domain with no further subsampling is the ability to uniquely identify the ordering of the instantaneous causal effects in the structural matrix $\C$. To highlight this ability, we perform model selection on the zero entries in $\C$ to determine the causal ordering of the instantaneous effects. Specifically, we calculate the BIC score for the nested models $M:\C_{2,1} = \C_{2,1} = 0$, $M_{GDP \to TP}:\C_{1,2} = 0$, $M_{TP \to GDP}: \C_{2,1} = 0$, and $M_{GDP \to TB, TB \to GDP}$. Models $M$, $M_{GDP \to TB}$ and $M_{TB \to GDP}$ represent DAG structures on the instantaneous effects while the unrestricted model $M_{GDP \to TB, TB \to GDP}$ does not. The BIC scores for all models are given in Table \ref{mf_realdat}. We see that the $M_{TB \to GDP}$ model performs best. The estimated $\C$ matrix is given by 
$\hat{\C} = \left( \begin{array} {c c}
.950 & 0.000 \\
.2800 & .695
\end{array} \right)$, suggesting an instantaneous interaction at the monthly scale from TB to GDP. The inferred transition matrix is given by $\hat{\A} = \left( \begin{array} {c c}
0.297 & -0.068 \\
0.01194395 & 0.658
\end{array} \right)$ suggesting a slight negative lagged interaction from GDP to TB. 

The above analysis fits an SVAR model at the time scale of months, the same sampling rate as the TB time series.  The results from Section \ref{mf} indicate that we could uniquely identify models at bi-monthly, or even more granular, time scales. However, even at the bi-monthly rate, the computational complexity of the E-step of the EM algorithm becomes large due to the large number of combinations of error mixture components in a `block', as discussed in Section \ref{e_step}. We note, however, that the E-step requires running the forward backward algorithm many times. The marginalization over mixture assignments can be run in parallel and massive computational gains could be gleamed from a GPU implementation; we leave this for future work since implementing the forward backward algorithm on a GPU is nontrivial. 

\begin{table}
\centering
\begin{tabular}{ c c c c c } 
Model & $M$ & $M_{GDP \to TB}$ & $M_{TB \to GDP}$ & $M _{GDP \to TB, TB \to GDP}$  \\ 
BIC & 1984.004 & 1983.409 & {\bf 1981.082} & 1987.550 \\
\end{tabular}
\caption{SVAR BIC on the GDP and TB data set for different instantaneous causality structures.}
\label{mf_realdat}
\end{table}

\section{Discussion} \label{discussion}
Our results provide sufficient conditions for identifiability of structural VAR models for both subsampled and mixed frequency time series. Importantly, the complete causal diagram of both lagged effects and instantaneous causal effects is fully identifiable under arbitrary subsampling schemes and non-Gaussian errors. 

For estimation, we developed an exact EM algorithm for maximum likelihood estimation and analyzed its performance via simulations. Our EM estimation approach has two drawbacks: 1) high computational complexity due the evaluation of the Kalman filter over all local mixture error assigments within a subsampled block and 2) many local optima due to weak identifiability and general nonidentifiability from the first two moments. Our simulations show that the local mode problem is more severe under even subsampling factors and low signal to noise regimes.

An ongoing line of work is to develop approximate inference for these models using MCMC or variational methods. Unfortuntely, we have found that the local optima problem makes MCMC approaches particularly difficult in this domain. A Gibbs sampler we have explored gets stuck in one local mode and requires the same number of random restarts as our EM algorithm to find a good solution. Perhaps incorporating recent MCMC advances \citeU{Ma:2016} may prove beneficial. We have also attempted a variational EM algorithm for this problem but found that performance was excessively poor. \cite{Gong:2015} also reported significantly worse results for a variational EM approach as compared to their approximate EM algorithm. By breaking the dependence between the unobserved, subsampled $x_t$ and the auxiliary $z_t$s, the variational approach avoids the combinatorial evaluation of a Kalman filter; however, this dependence is critical for correctly evaluating the probable trajectories of the latent $x_t$, without which inference of $\A$ suffers. As an alternative to approximate methods, exploring parallel GPU implementations of the E-step in our EM algorithm would allow scaling to both more time series and greater subsampling factors.

As a future research direction, it would be interesting to specify what order moments of the process are required for identifiability. This line of work may aide in developing a method of moments estimation procedure based on third order moments for this problem. A method of moments approach may side step both the local optima problem and the combinatorial computational complexity of the EM algorithm.

\paragraph{\bf Acknowledgements}  We thank Mathias Drton for a helpful discussion. AT and EF work was supported in part by
ONR Grant N00014-15-1-2380, NSF CAREER Award IIS-1350133 and  AFOSR Grant FA9550-16-1-0038. AT was also
partially funded by an IGERT fellowship. AS acknowledges the support from NSF grants  DMS-1161565 \& DMS-1561814 and NIH grants 1K01HL124050-01 \& 1R01GM114029-01.

\bibliographystyle{unsrt}
\bibliography{ng_subsamp}
\section{Appendix}
\subsection{Proof of Theorem \ref{mixed_corr_theorem}}
 We prove it for the subsampled case.
 The structural VAR model can be decomposed as:
 \begin{align}
 \tilde{x}_t &= \A^k \tilde{x}_{t-k} + \Ll \tilde{e}_t \\
 &=  \A^k \tilde{x}_{t-1}  + \vec{e}_t,
 \end{align}
 where $\Ll = (\C, \A \C, \ldots, \A^{k-1}\C)$ and $\vec{e}_t  = \Ll \tilde{e}_t$. We may determine $A^{k}$ uniquely by linear regression and thus determine the distribution of $\vec{e}_t$. Proposition 1 states that each column of $L'$ is a scaled version of a column of $L$. Denote by $\Ll_{lp + i}$, $l = 0, \ldots, k - 1$, $i = 1,\ldots, p$ the $(lp + i)$th column of $\Ll$, and similarly for $L'_{lp + i}$. From the Uniqueness Theorem in Erikson and Koivunen 2004 \cite{eriksson:2004}, we know that under condition A2, for each $i$, there exists one and only $j$ such that the distribution of $e_{(t-l)i}, l = 1, \ldots, k - 1$ is the same as the disitrbution of $e'_{(t-l)j}$, $l = 1,\ldots, k - 1$ up to changes in location and scale. This implies that each column in $\Ll_{lp + i}$, $l = 0, \ldots, k - 1$, is proportional to at least one of the nonzero columns in $\Ll_{lp+j}$, $l=1, \ldots, k-1$, and vice versa. The proportionality must be either $1$ or $-1$ since we have standardized the $p_e$ to have unit variance. Furthermore, it must be the case that $\L_{lp + i}$ is proportional to column $\Ll'_{lp + j}$ for $j$ and $i = 1,\ldots, p$ since the columns are ordered in magnitude in both $\Ll$ and $\Ll'$, ie $||\Ll_{lp + i}||_2 > ||\Ll_{(l + 1)p + i}||_2$,
 \begin{align}
 ||\Ll_{(l+1)p + i}||_{2} &=  ||\A \A^{l} \C_{:i}||_2 \\
 &< ||\A||_2 ||\A^{l} \C_{:i}||_2 \\
 &<  ||\A^{l} \C_{:i}||_2 \\
 &= ||\Ll_{lp + i}||_{2}.
 \end{align}
 This implies that $\Ll'$ may be written as:
 \begin{align}
 \Ll' &= \Ll P \\
 &= \left( \C P_0, \A \C P_1, \ldots \A^{k-1} \C P_{k-1} \right),
 \end{align}
 where $P_i$ is a scaled permutation matrix with either $1$s or $-1$ scaling factors where $P_i$ and $P_j$ have the same permutation pattern but potentially different scaling factors. This proves the first assertion, ie $\C' = \C P_0$ and $\Sigma' = \C' \C'T = \C P_0 P_0^T \C^T = \C \C^T = \Sigma$. Now, if the $p_e$ are restricted to be nonsymmetric then the scaling factors must all be $1$ so that all the $P_i$ are equal. 
 \begin{align}
 \A' \C' &= \A' \C P \\
 &= \A \C P
 \end{align}
 and since $\C$ is full rank, $\C P$ is full rank so that $\A' = \A$, as desired.

 \subsection{Theorem \ref{mixed_corr_theorem_mf} part 2}
 If $\C$ is lower triangular then $\C = \C'$. Now, $\A \C = \A' \C' P_1 = \A' \C D$ where $D$ is diagonal with either $1$ or $-1$ on the diagonal. This implies that $\Ll'_{p+1:2p} = \A \C D$. We procceed by induction. Since the last column of $\C$, $\C_{:p}$, is zeros everywhere except the last element, we must have that $\C_{pp} \A_{:p} D_{pp}=\Ll'_{2p} =  \C_{pp} \A'_{:p}$, so that $\A_{:p} D_{pp} = \A'_{:p}$. Following the same logic as the proof to item 2 of Theorem \ref{Gong_theorem_mf}, if there exists some $j$ such that a multiple of $k_p$ is one less than a multiple of $k_j$ and $\A_{pj} \neq 0$, then we can identify $\A_{pj}$, and hence its sign, implying $\A_{:p} = \A_{:p}'$. 

 Assume that $\A_{:i} = \A'_{:i}$ for $i > j$. Since $\C$ is lower diagonal we must have that 
\begin{align}
 \Ll'_{p + j} &= \left(\C_{jj} \A'_{:j} + \sum_{i > j} \C_{ij} \A_{:i}\right) \\
 &= D_{jj} \left(\C_{jj} \A_{:j} + \sum_{i > j} \C_{ij} \A_{:i}\right).
\end{align}
Since $\A_{lj} = \A'_{lj}$ with $\C_{jj} \A_{lj} + \sum_{i > j} \C_{ij} \A_{li} \neq 0$ for some $l$, this implies $D_{jj} = 1$, so that $\A_{:j} = \A'_{:j}$. Taken together, $\A = \A'$.

 \subsection{EM algorithm details}

The gradient of the expected joint log probability given in the main text with respect to  $\W = \C^{-1}$ is given by:
\begin{align}
\nabla l(\W) = T \W^{-T} + \sum_{t = 1}^T \sum_{j = 1}^p \sum_{i = 1}^m \frac{1}{\sigma^2_{ji}}\bigg(- E(z_{tji} x_t x_t^T | \tilde{X}) \W_j^T - \A E(z_{tij} x_{t-1} x_{t-1}^T| \tilde{X}) A^T W_j^T \\
+ \left(E(z_{tji} x_t x_{t - 1}^T | \tilde{X}) \A^T + \A E(z_{tjj} x_{t-1} x_{t} | \tilde{X}) \right) \W_i^T + E(z_{tji} x_{t}| \tilde{X}) \mu_{ji} - \A E(z_{tji} x_{t-1}|\tilde{X}) \mu_{ji} \bigg) 
\end{align}

and the Hessian with respect to ${\bf w} = \text(vec)(\W)$ is given by
\begin{align}
H({\bf w}) = -T {\bf \Omega} (\W^{-T} \otimes \W^{-1}) + \sum_{t = 1}^T \sum_{j = 1}^p \sum_{i = 1}^m {\bf \Gamma}_{tji} \otimes {\bf D}^{(j)}
\end{align}
where 
\begin{align}
{\bf \Gamma}_{tji} = \frac{1}{\sigma^2_{ji}} \bigg(- E(z_{tji} x_t x_t^T | \tilde{X}) - \A E(z_{tji} x_{t-1} x_{t-1}^T) \A^T + E(z_{tji} x_t x_{t-1}^T | \tilde{X}) \A^T + \A E(z_{tji} x_{t-1} x_t^T| \tilde{X})  \bigg)
\end{align}
and ${\bf D}^{(j)}$ is a $p \times p$ matrix with ${\bf D}^{(j)}_{jj} = 1$ and all other entries zero. ${\bf \Omega}$ is a permutation matrix with all zero entries except with ${\bf \Omega}_{nm} = 1$ $\forall n \in (1 \ldots p^2)$ and $m = (n - 1) \text{mod}(p) + \lfloor (n - 1)/p \rfloor + 1 $. Finally, note there is a nonidentifiability between the scale of the errors, $e_t$, and the magnitude of $\C$. For algorithmic stability we fix the first mixture componenet for each $e_t$ to have variance set to one, $\sigma^2_{j1} = 1$ $\forall j$.

\subsection{Additional Simulation Plots}
Here we provide additional histogram plots from simulations in the main text. Figures \ref{f2} and \ref{f3} provide estimates for the remaining $(\A^{(1)}, \C^{(2)})$ and $(\A^{(2)},\C^{(1)})$ simulation parameter configurations. Figures \ref{box2} and \ref{box3} contain similar histogram plots but for the maximum eigenvalue experiments. 
\begin{figure}[!htb]
\includegraphics[width=.47\textwidth]{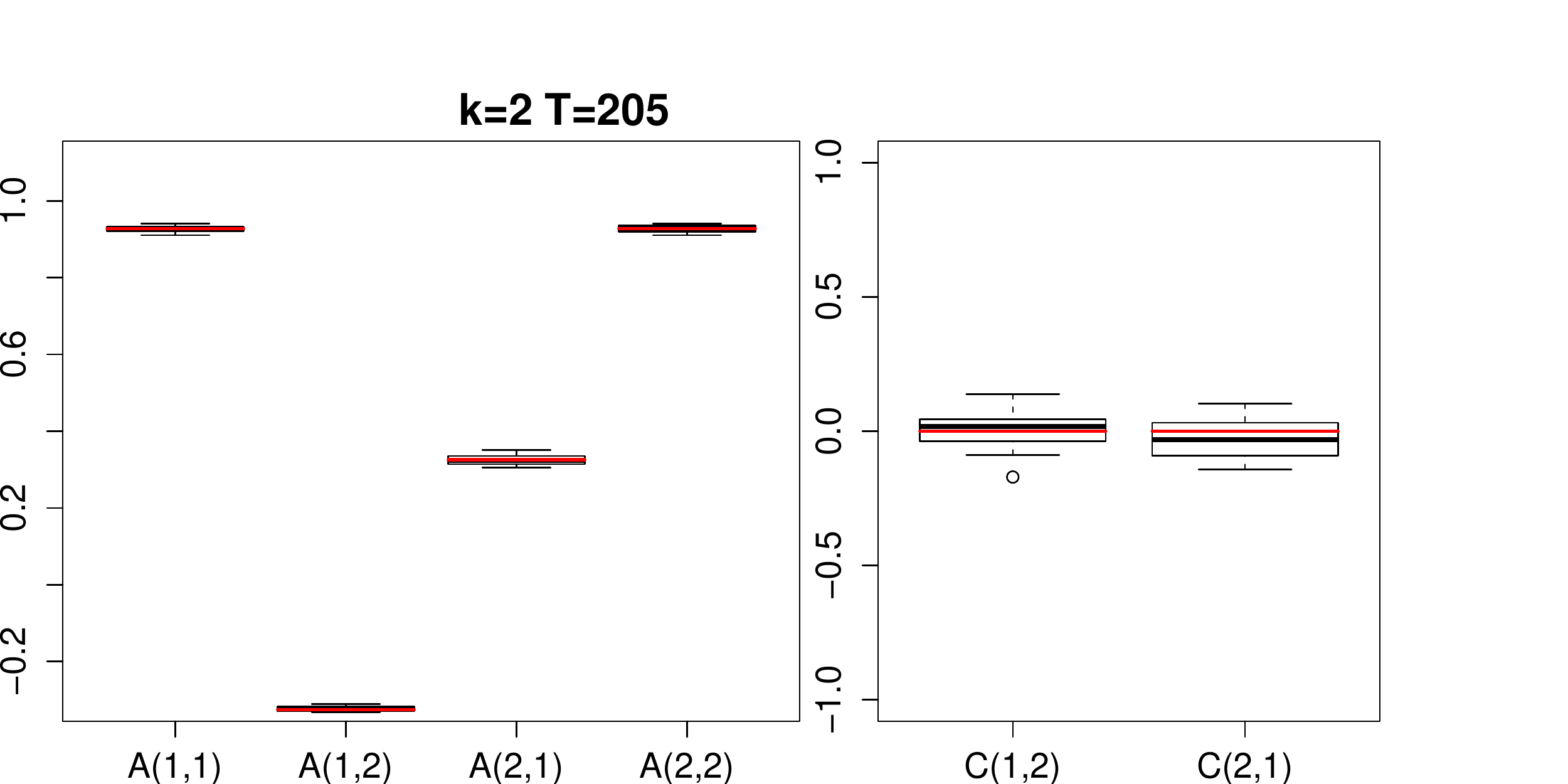} \hspace{.2 in}
\includegraphics[width=.47\textwidth]{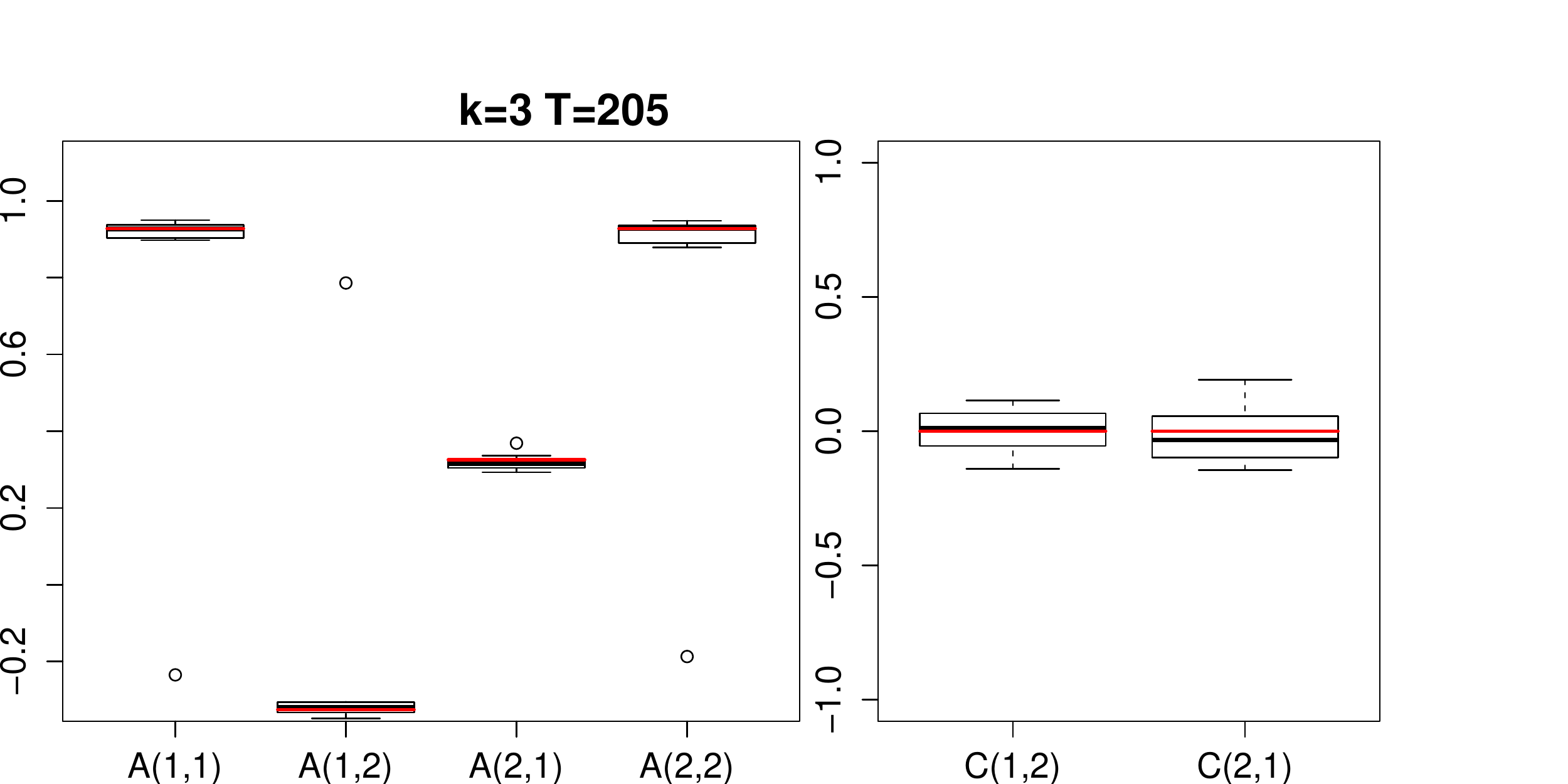}
\includegraphics[width=.47\textwidth]{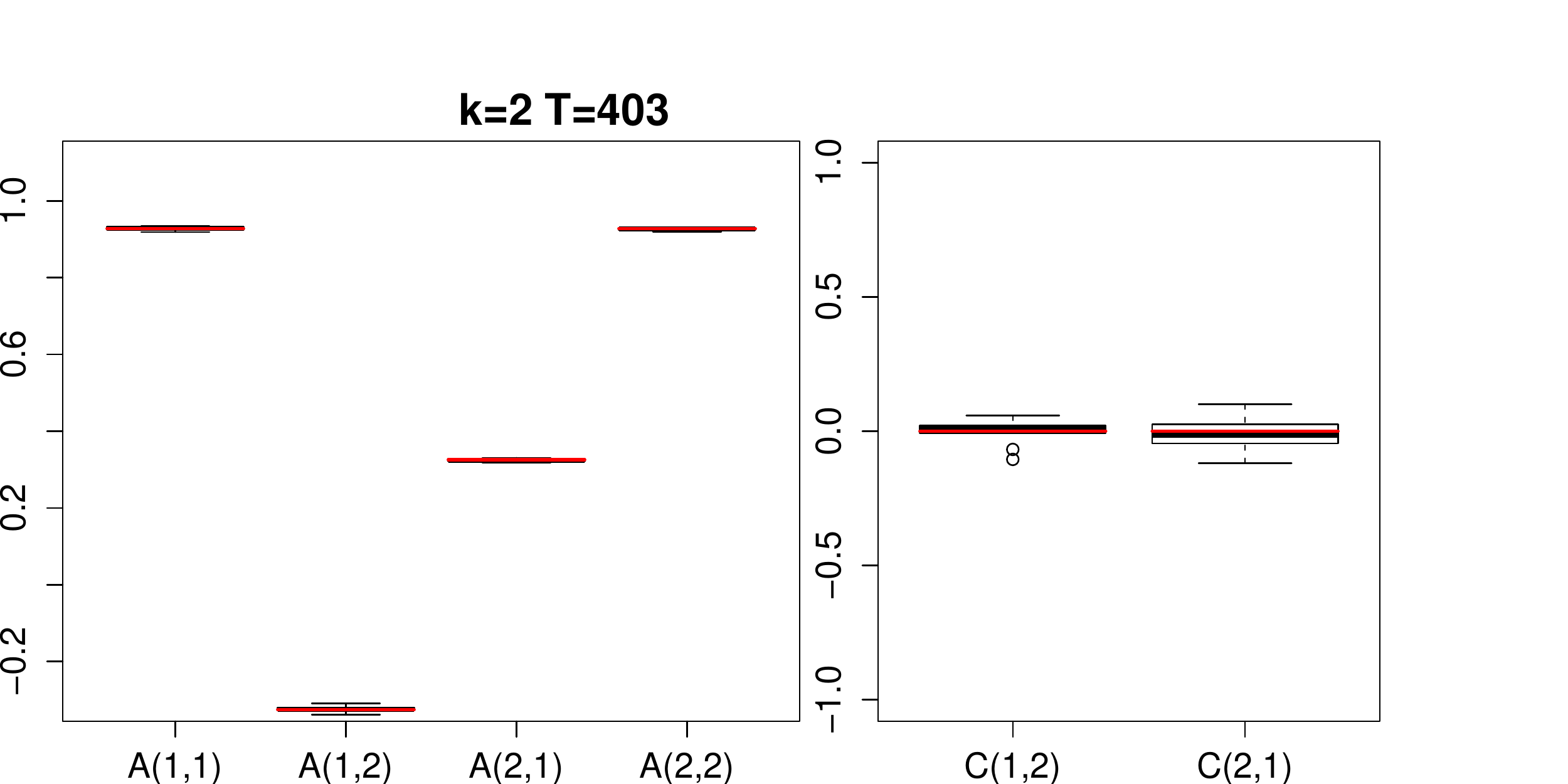} \hspace{.2 in}
\includegraphics[width=.47\textwidth]{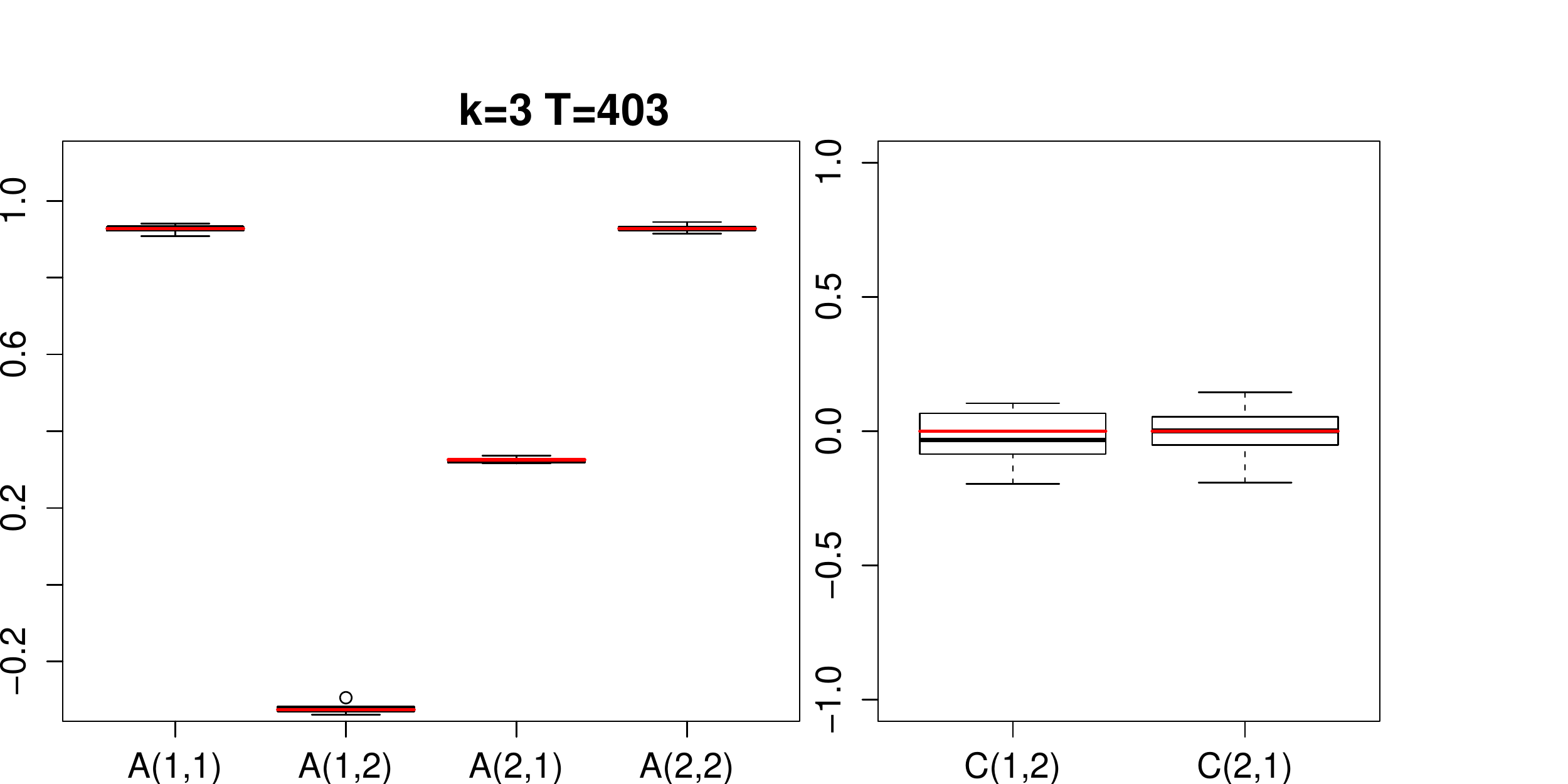}
\includegraphics[width=.47\textwidth]{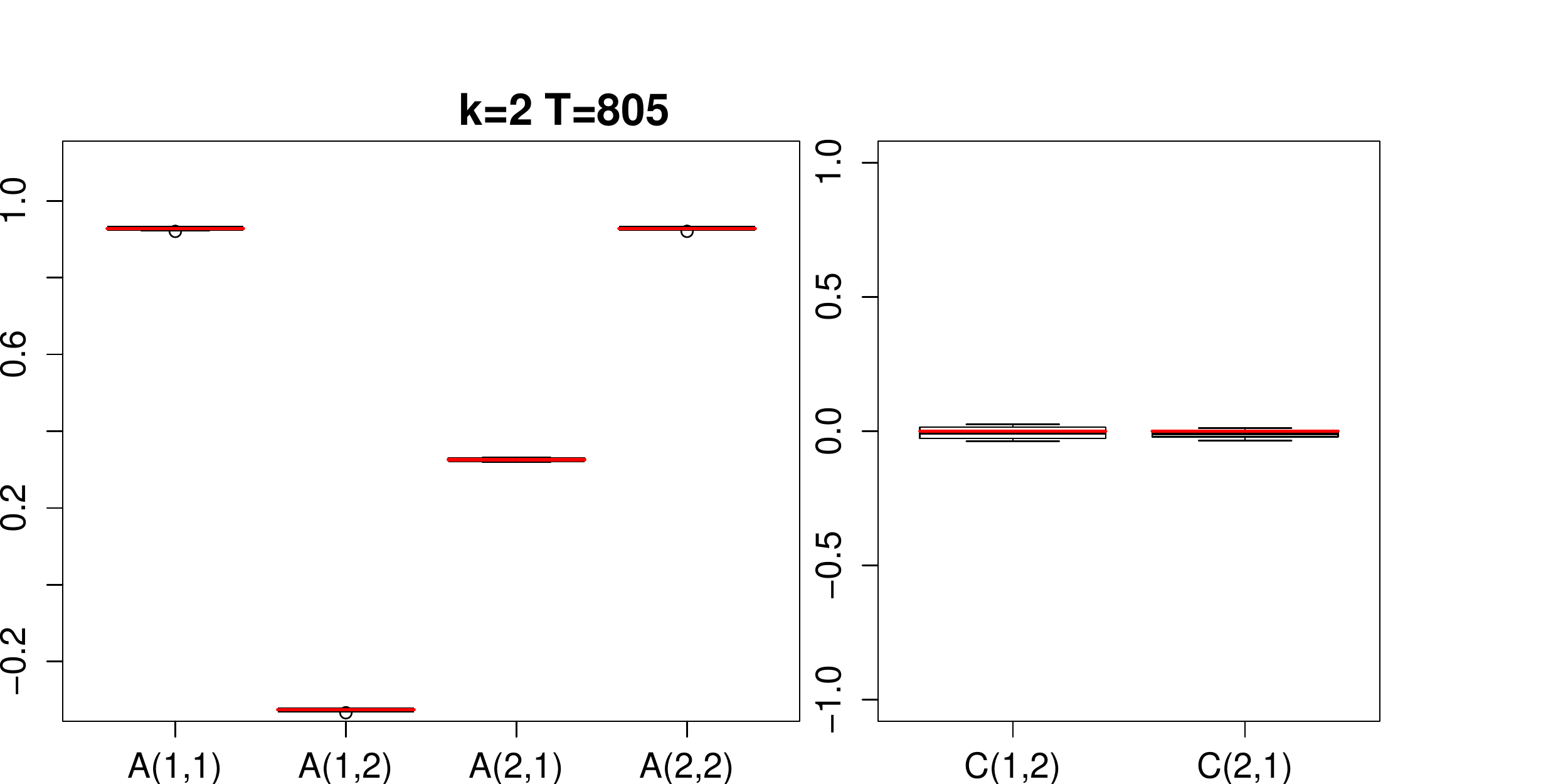} \hspace{.2 in}
\includegraphics[width=.47\textwidth]{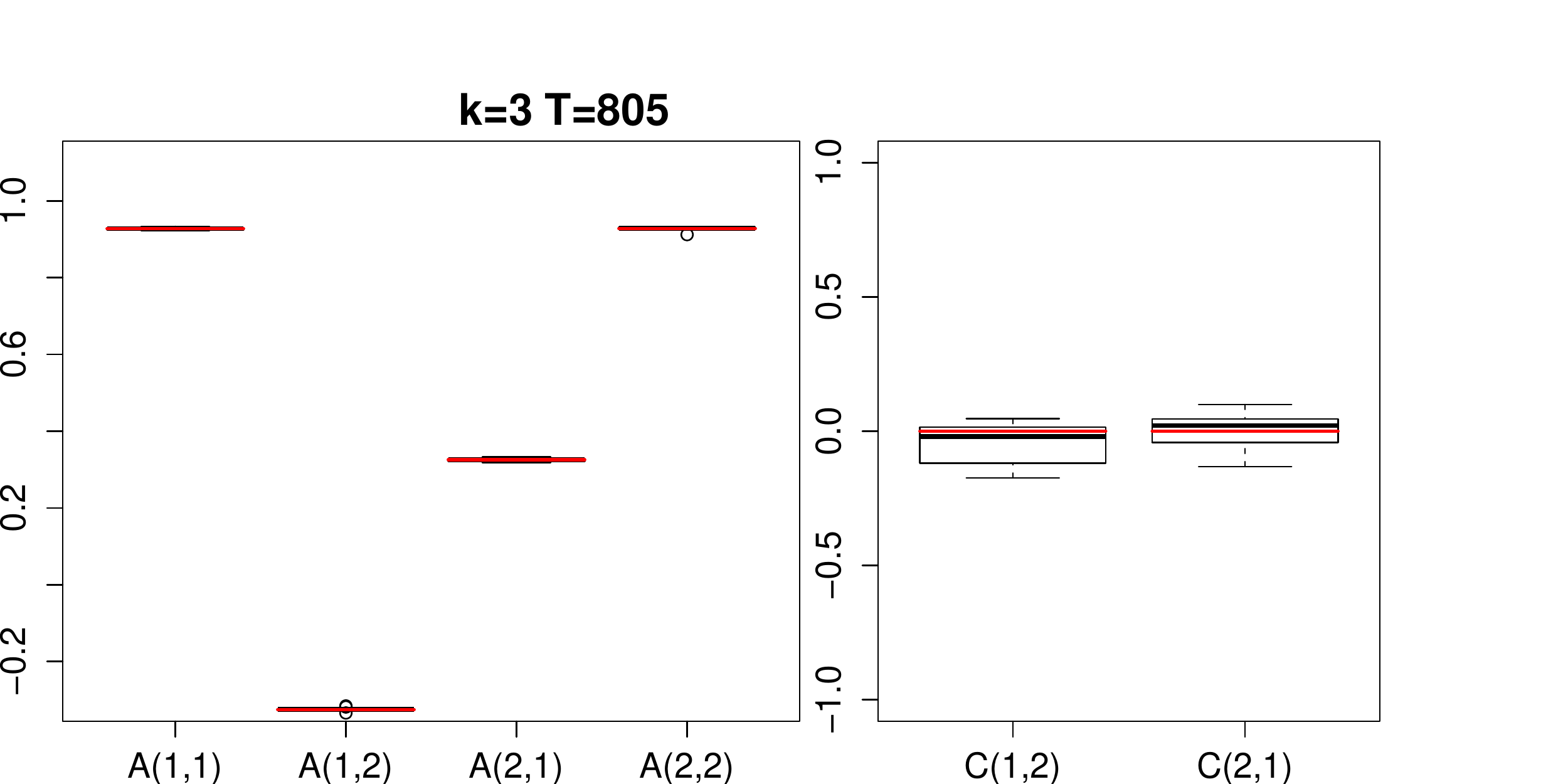}
\caption{Histogram plots of $\A^{(2)}$ and $\C^{(1)}$ parameter estimates as in \ref{f1}.}\label{f2}
\end{figure}
\begin{figure}[!htb]
\includegraphics[width=.47\textwidth]{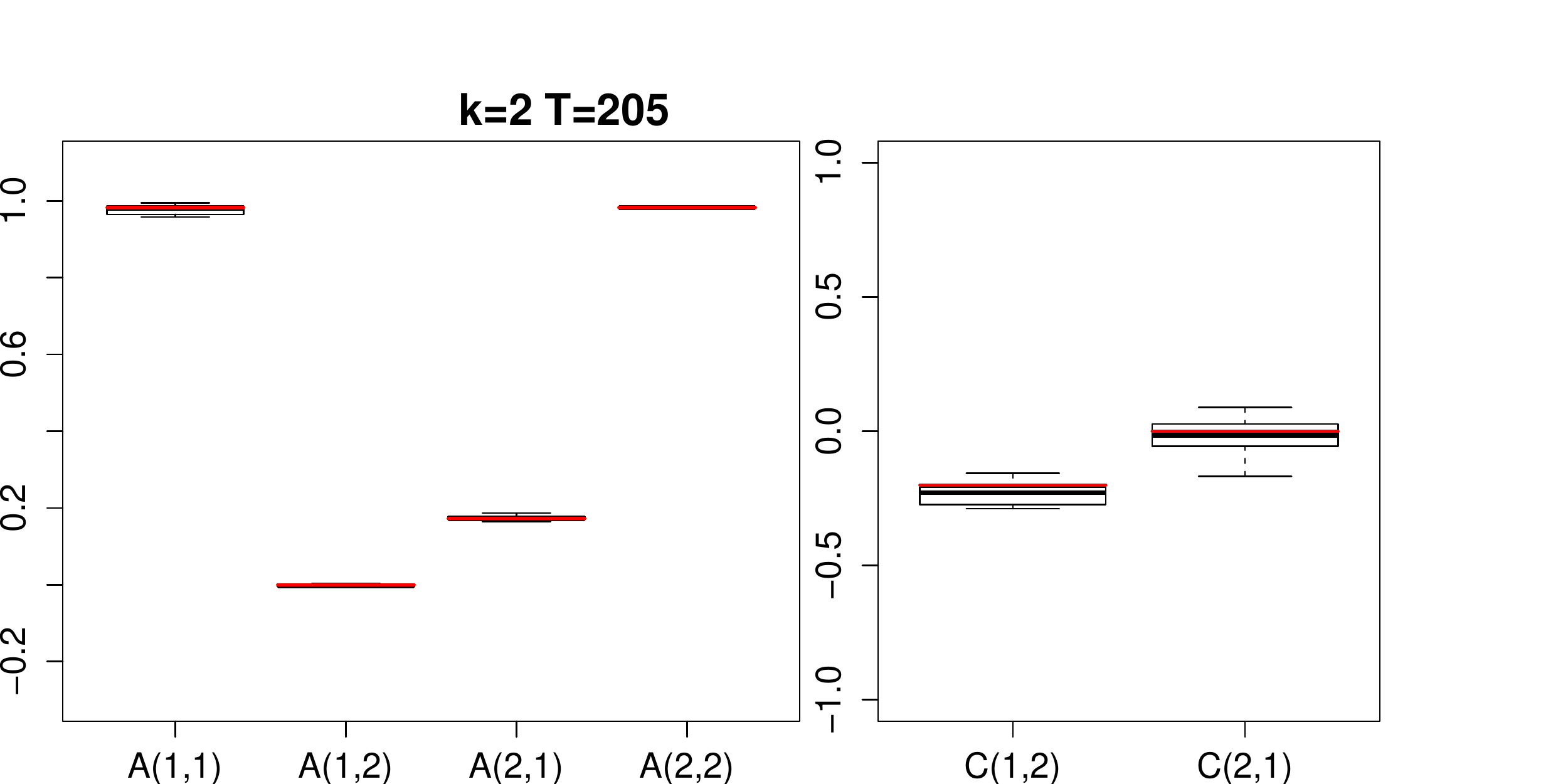} \hspace{.2 in}
\includegraphics[width=.47\textwidth]{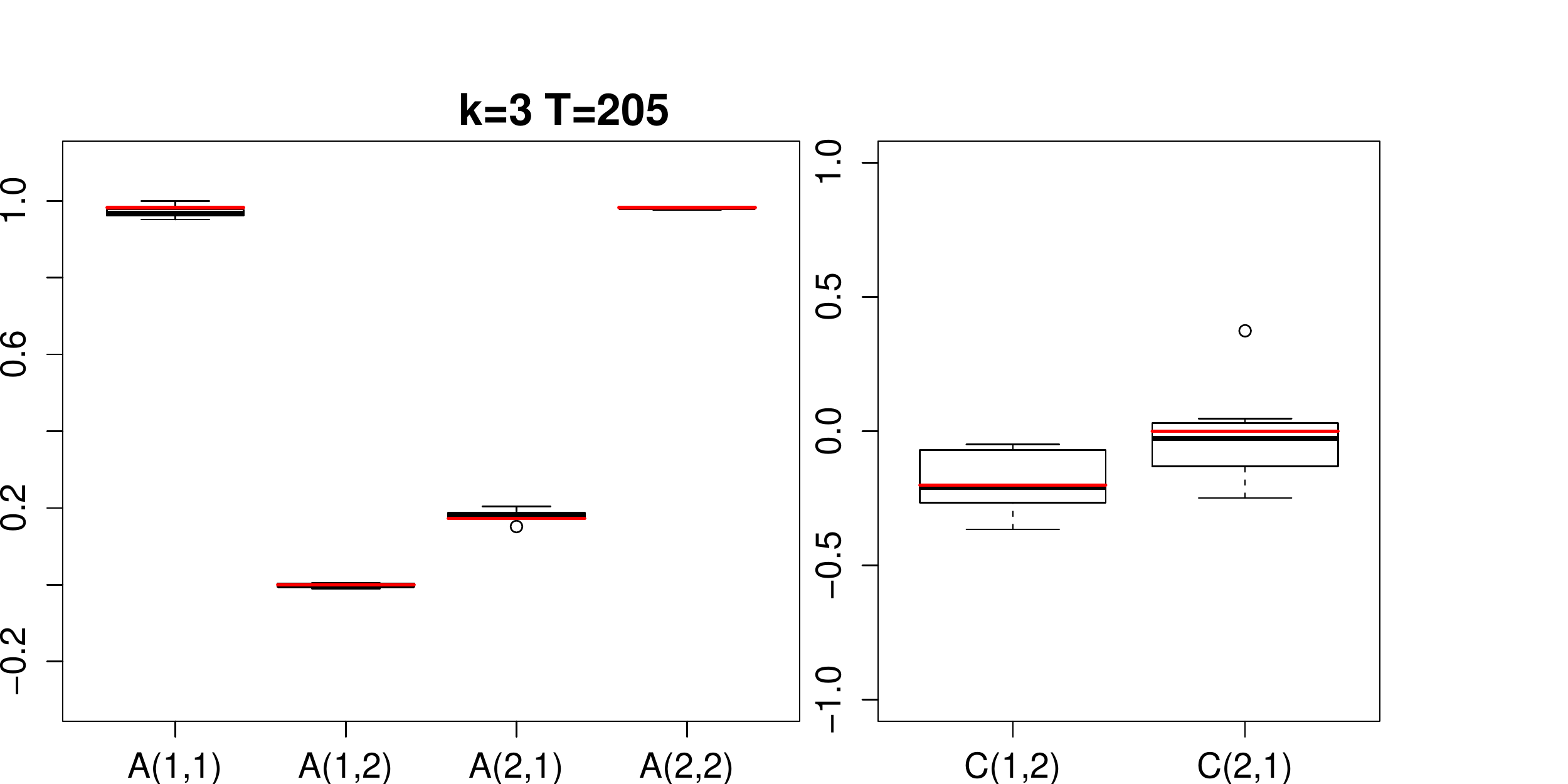}
\includegraphics[width=.47\textwidth]{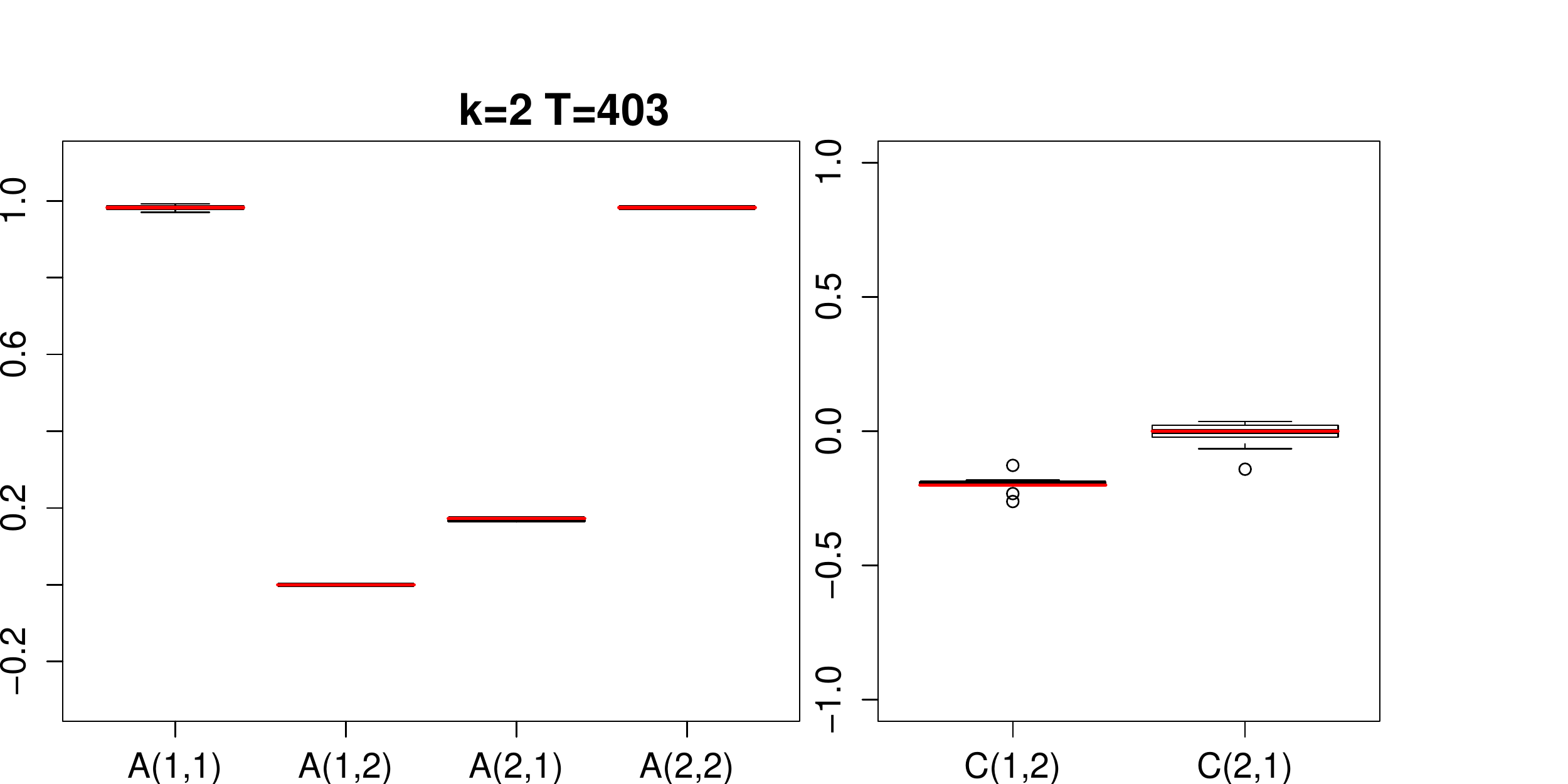} \hspace{.2 in}
\includegraphics[width=.47\textwidth]{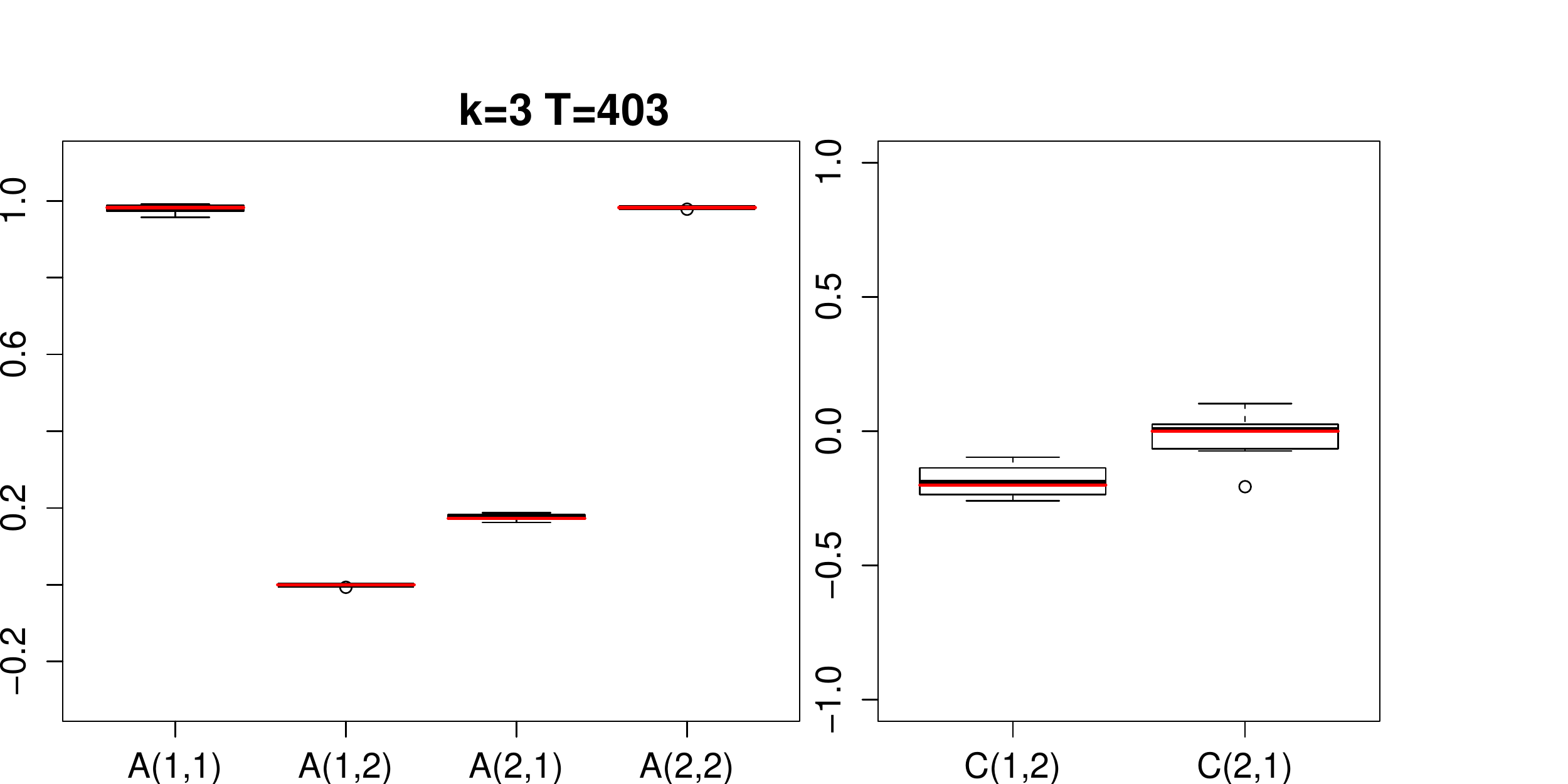}
\includegraphics[width=.47\textwidth]{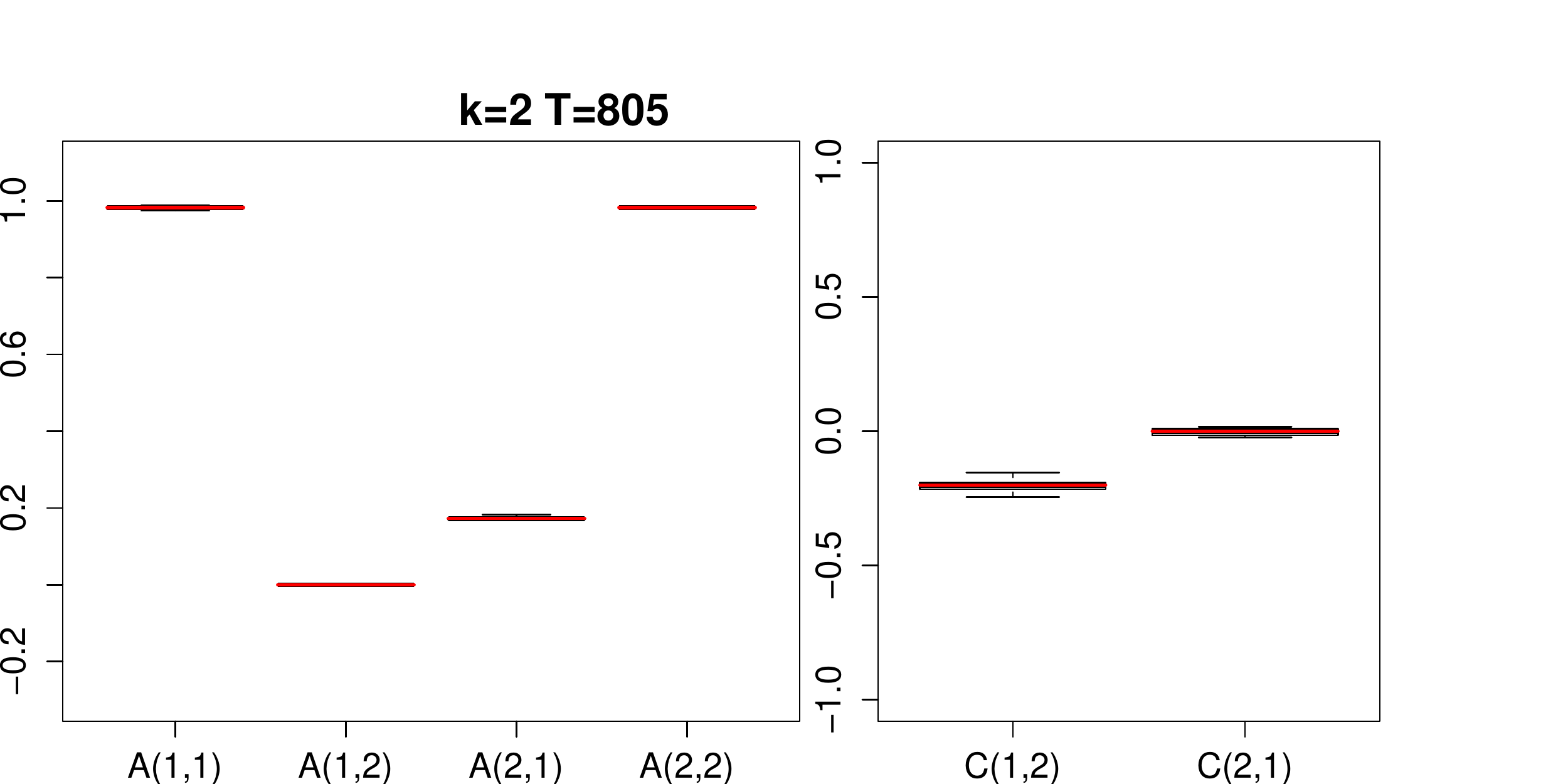} \hspace{.2 in}
\includegraphics[width=.47\textwidth]{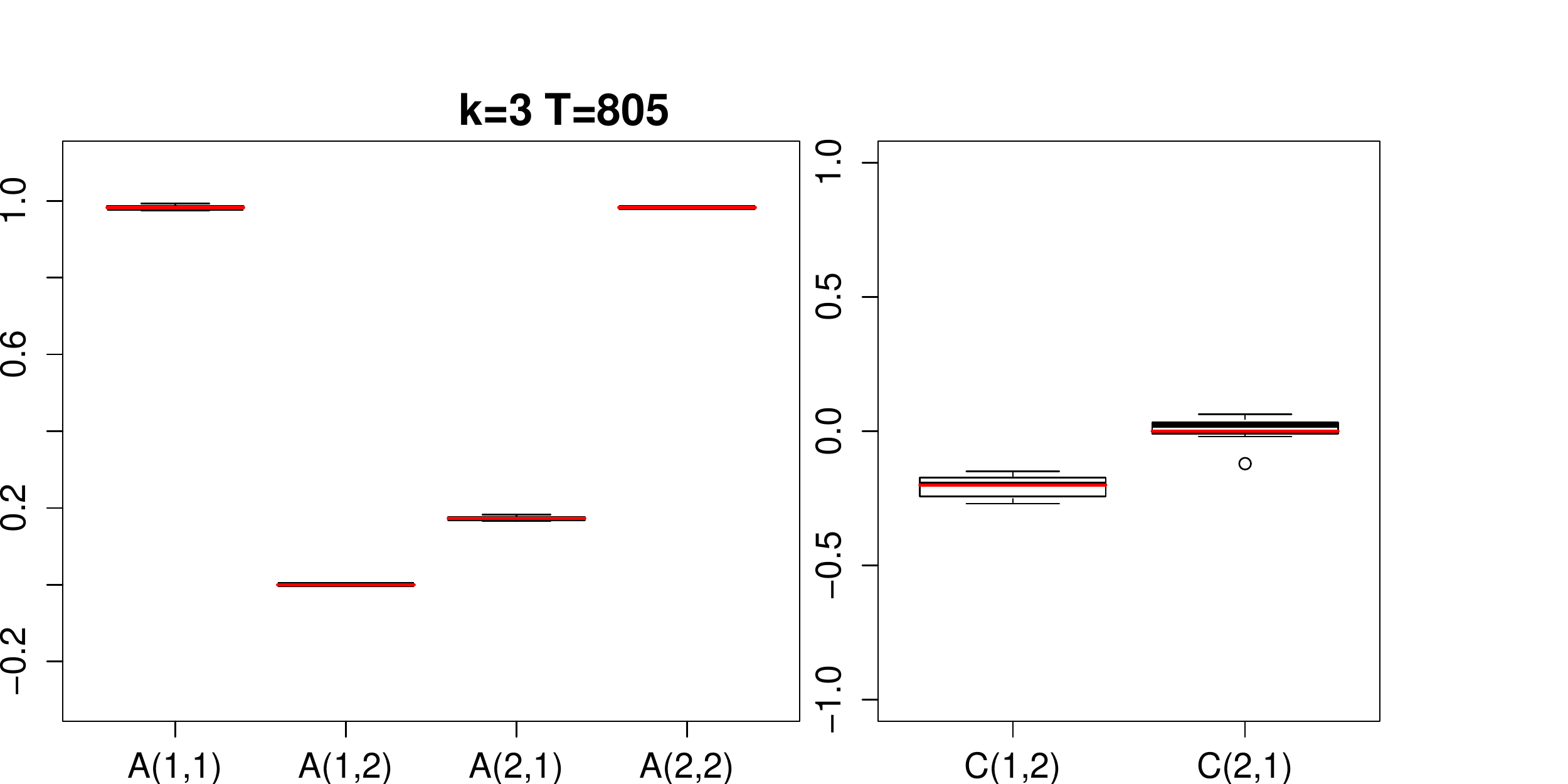}
\caption{Histogram plots of $\A^{(1)}$ and $\C^{(2)}$ parameter estimates as in Figure \ref{f1}.}\label{f3}
\end{figure}

\begin{figure}[!htb]
\includegraphics[width=.47\textwidth]{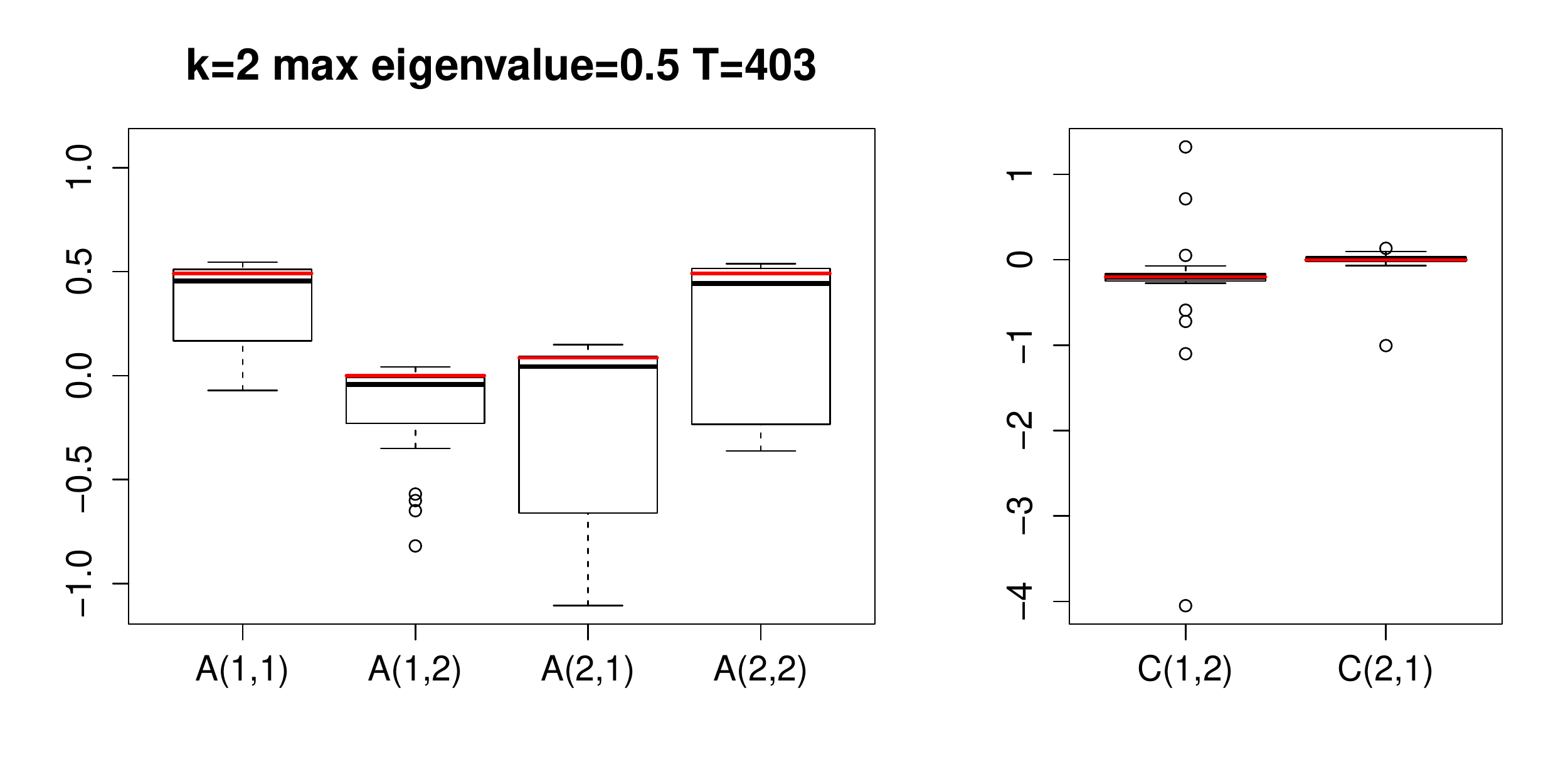} \hspace{.2 in}
\includegraphics[width=.47\textwidth]{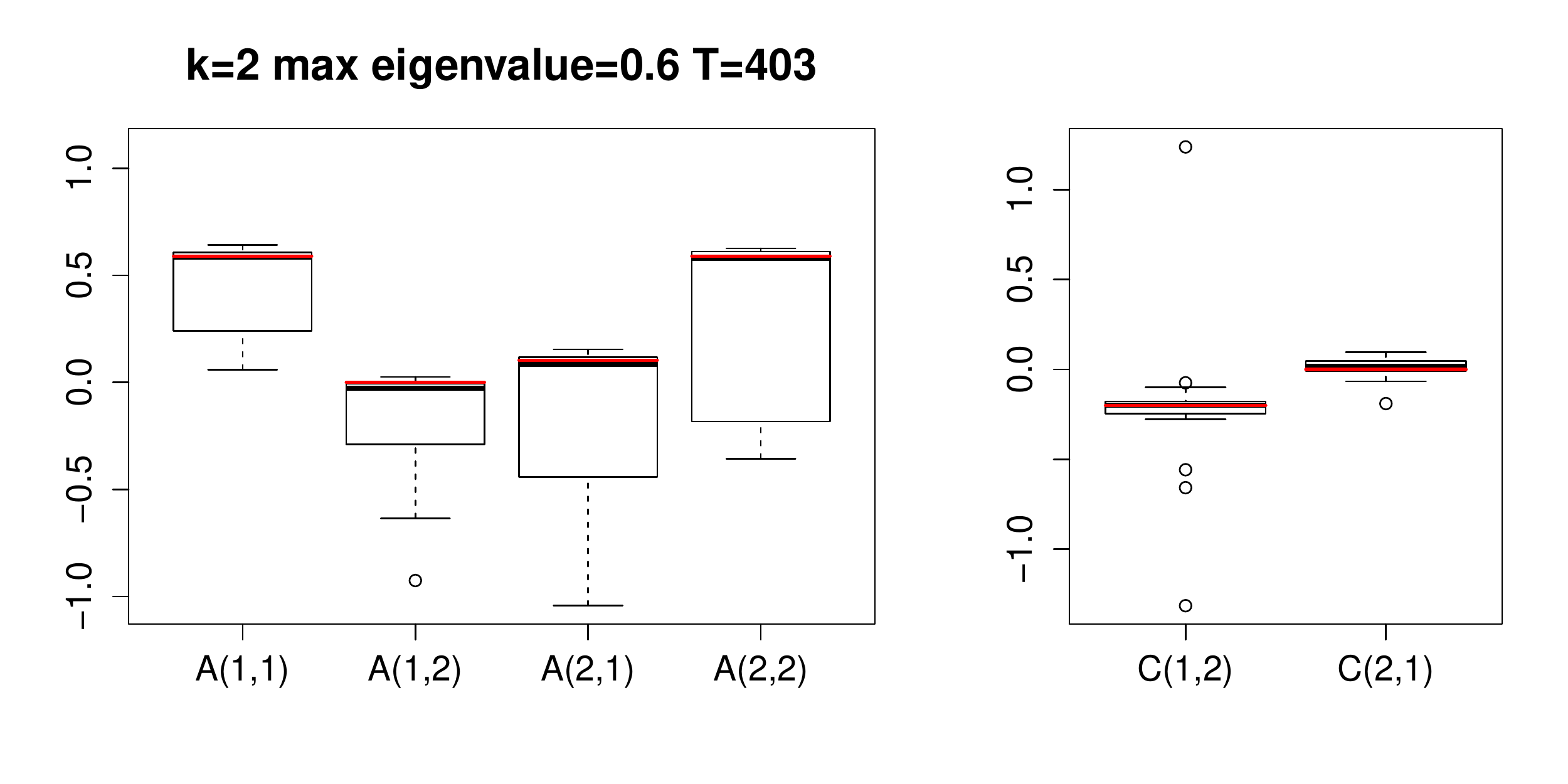}
\includegraphics[width=.47\textwidth]{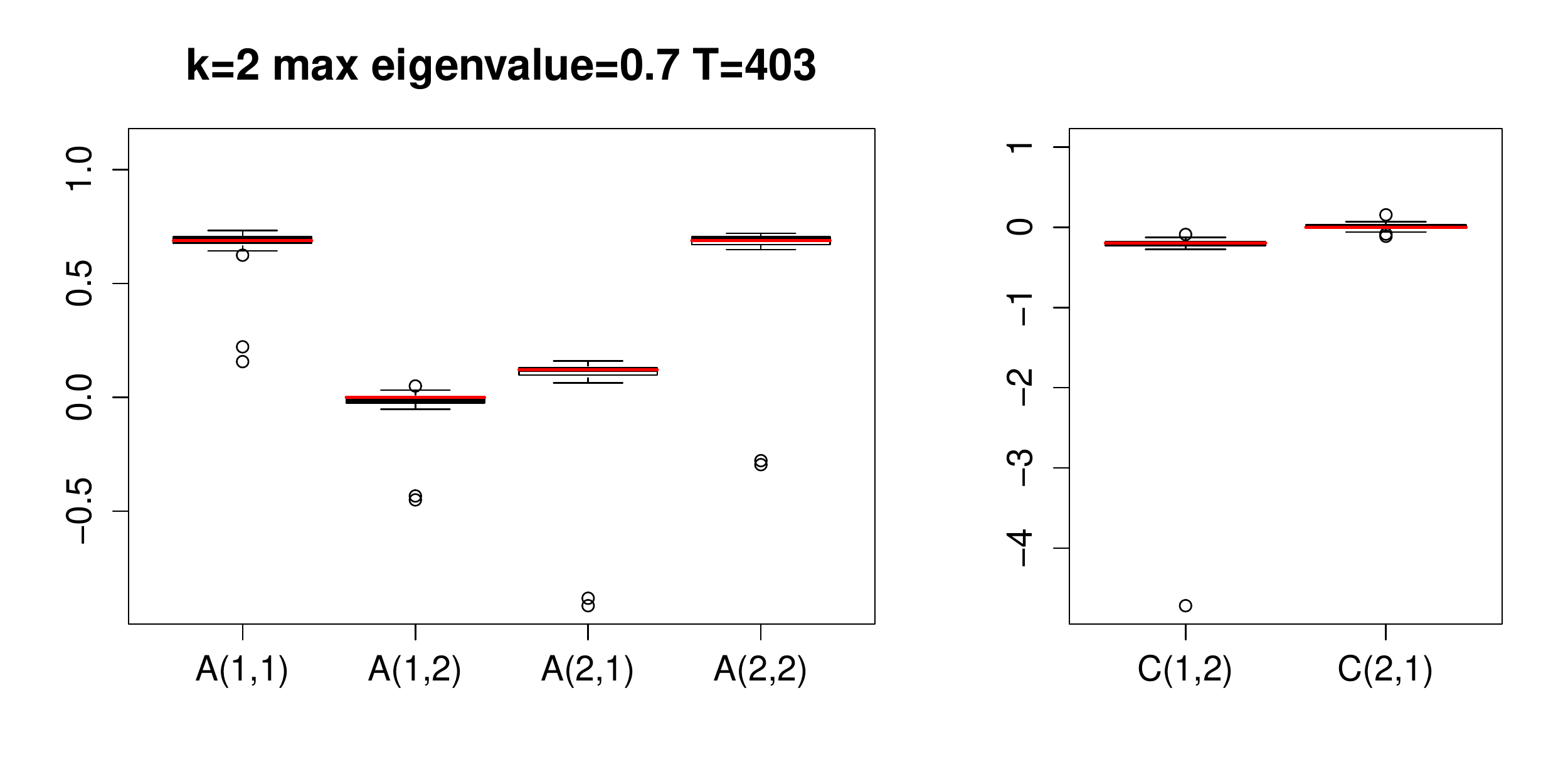} \hspace{.2 in}
\includegraphics[width=.47\textwidth]{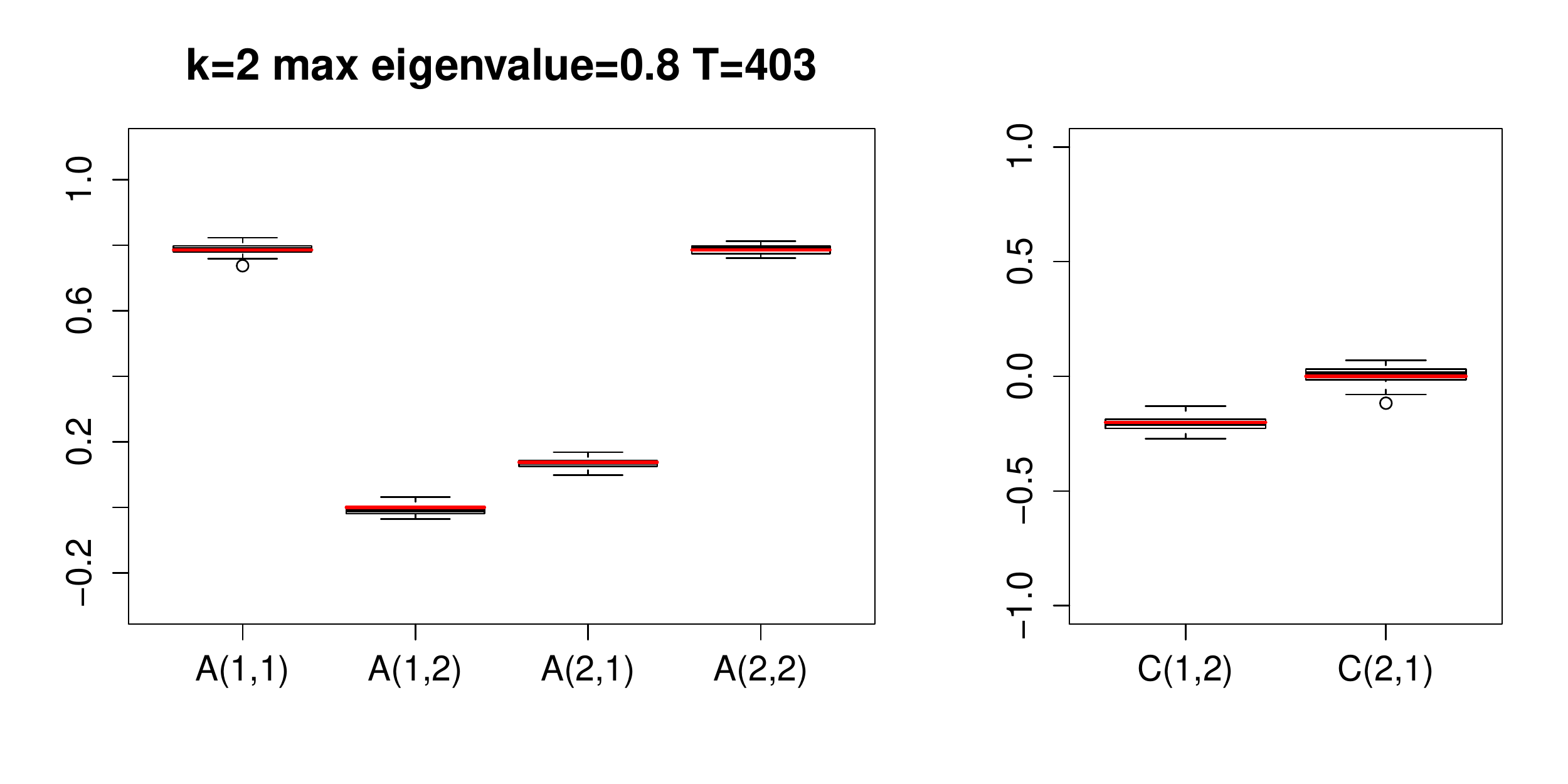}
\includegraphics[width=.47\textwidth]{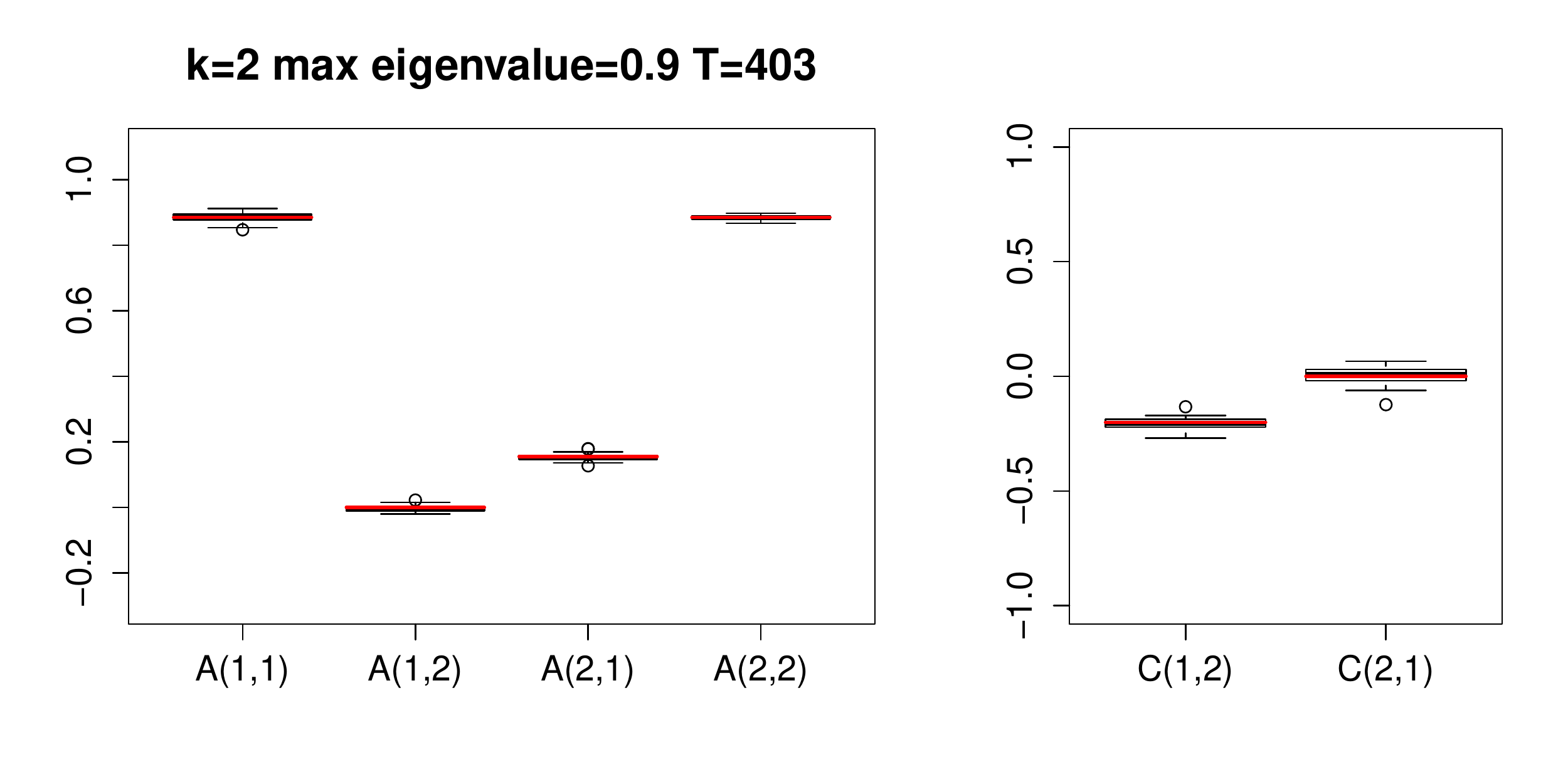} \hspace{.2 in}
\includegraphics[width=.47\textwidth]{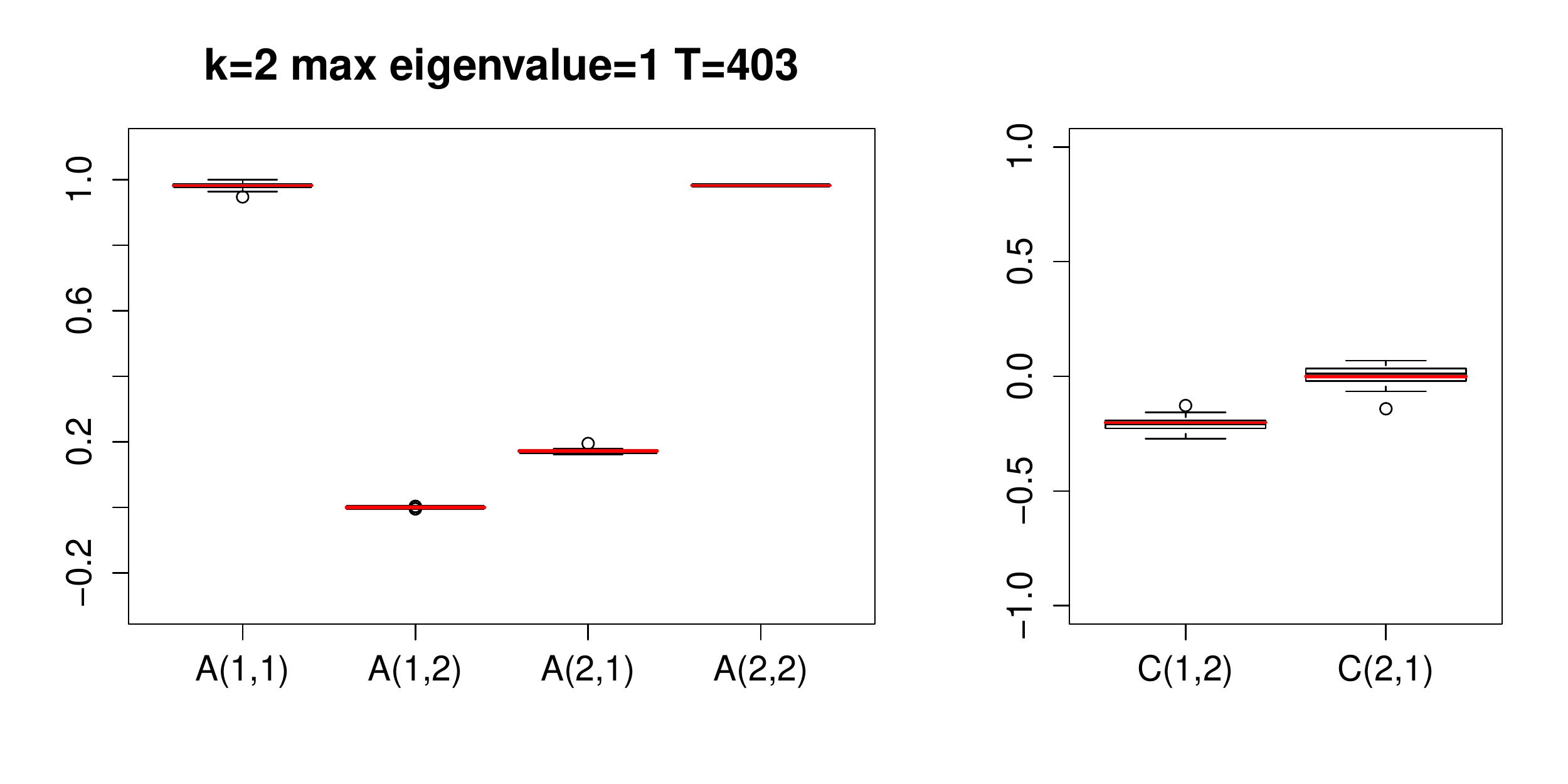}
\caption{Histogram plots of a scaled $\A^{(2)}$ and $\C^{(2)}$ parameter estimates over 40 random data samplings for differing maximum eigenvalue of $\A^{(2)}$ and subsampling factor of $k = 2$.}
\label{box2}
\end{figure}

\begin{figure}[!htb]
\includegraphics[width=.47\textwidth]{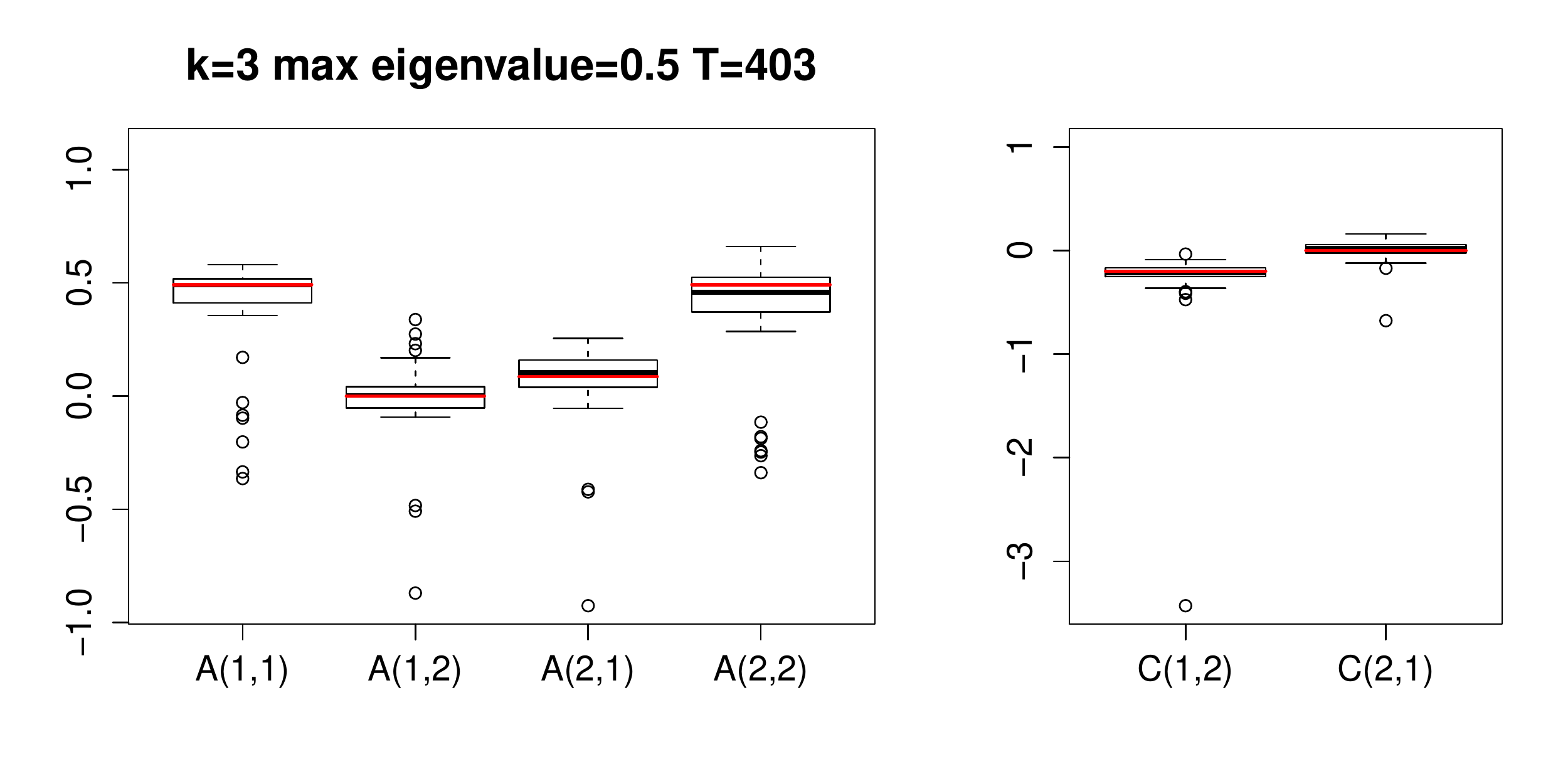} \hspace{.2 in}
\includegraphics[width=.47\textwidth]{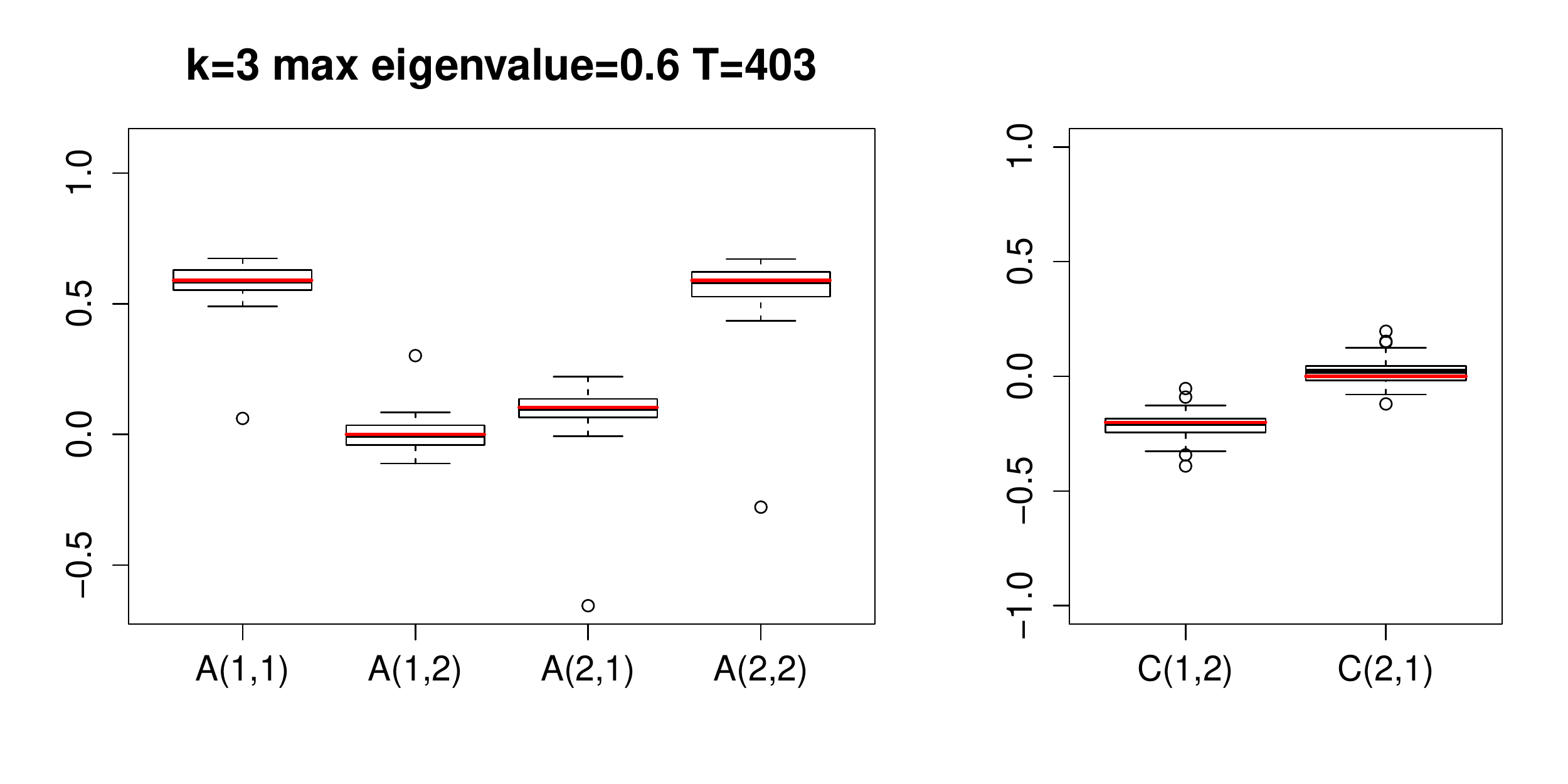}
\includegraphics[width=.47\textwidth]{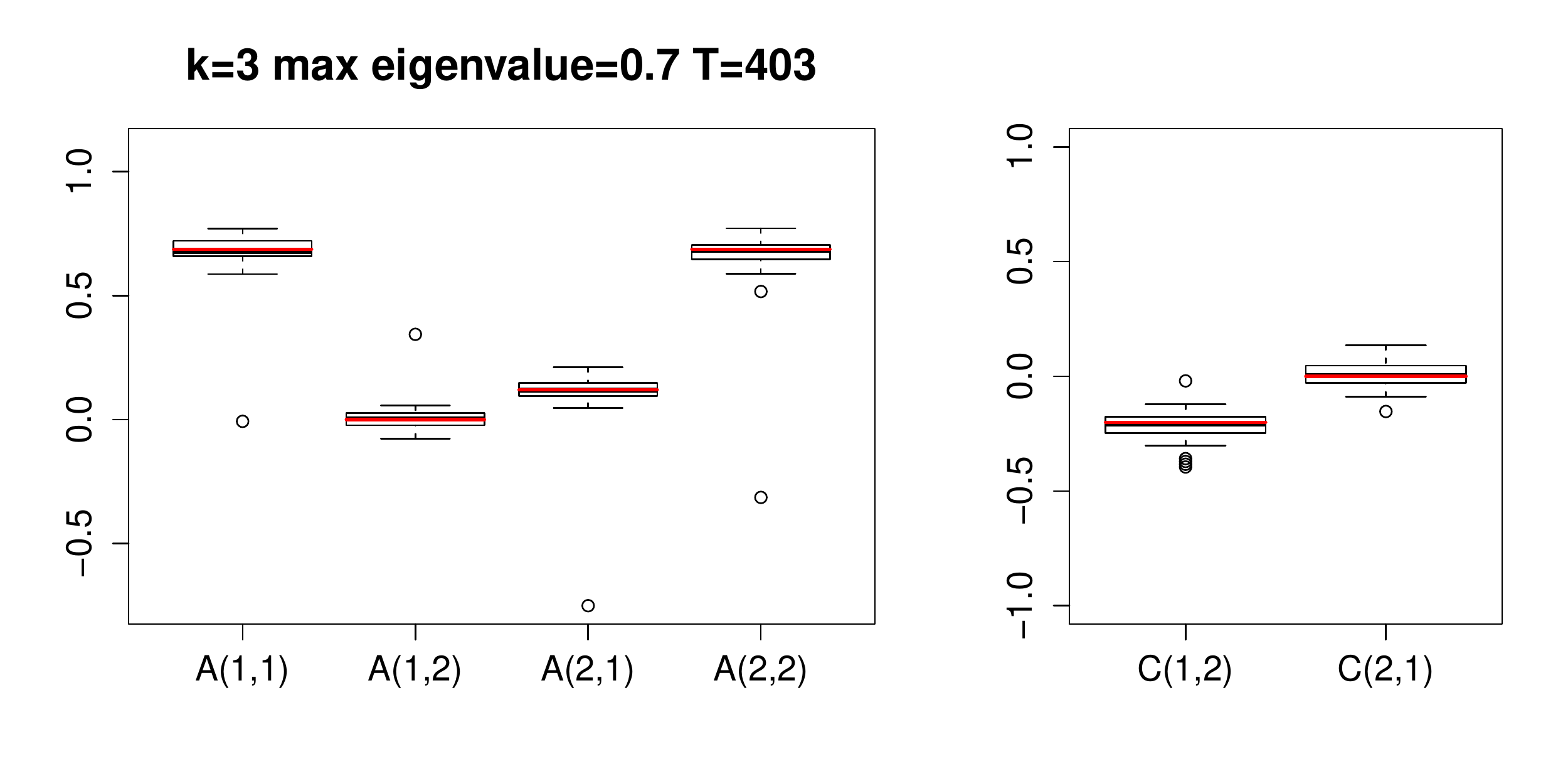} \hspace{.2 in}
\includegraphics[width=.47\textwidth]{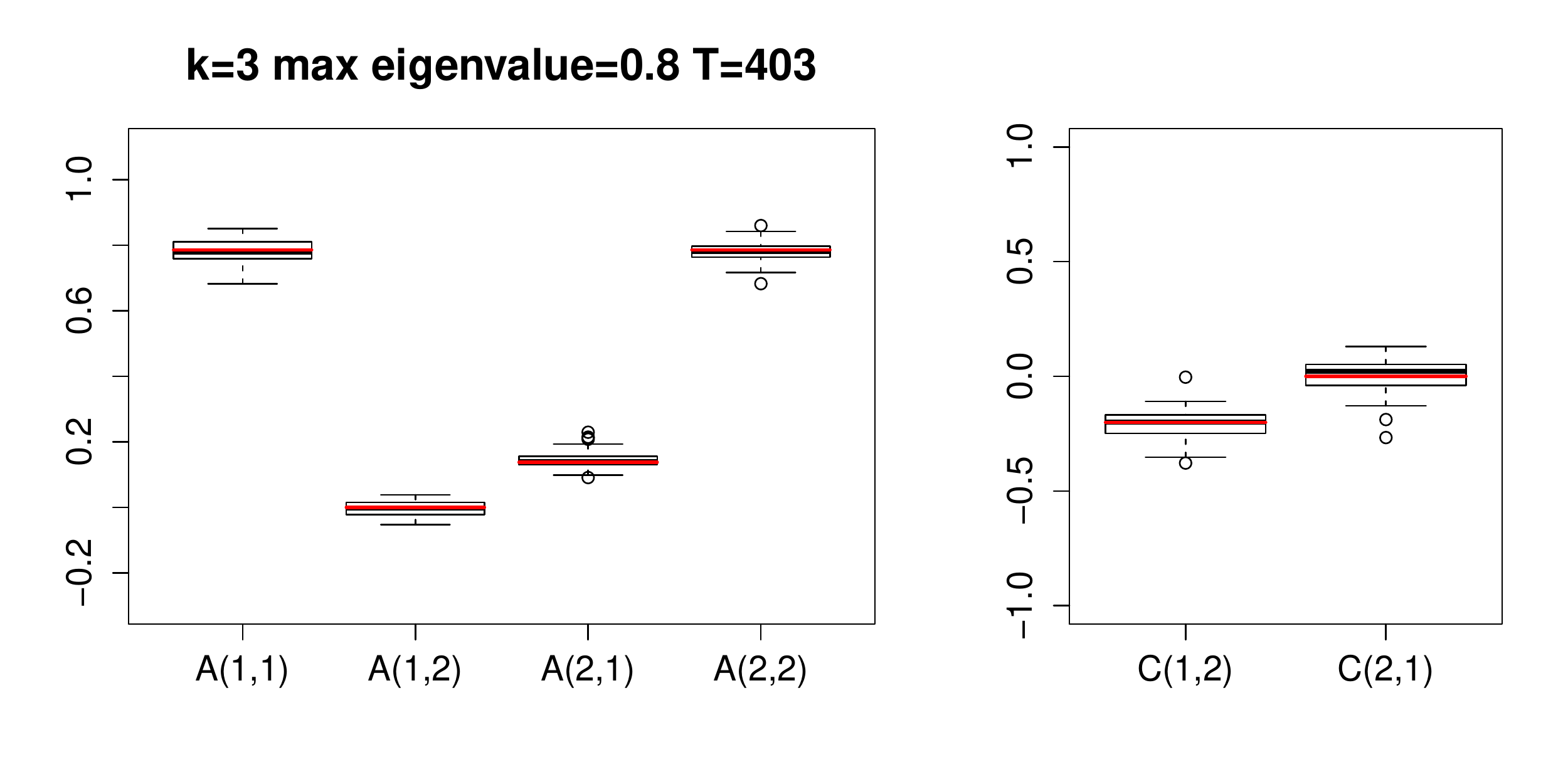}
\includegraphics[width=.47\textwidth]{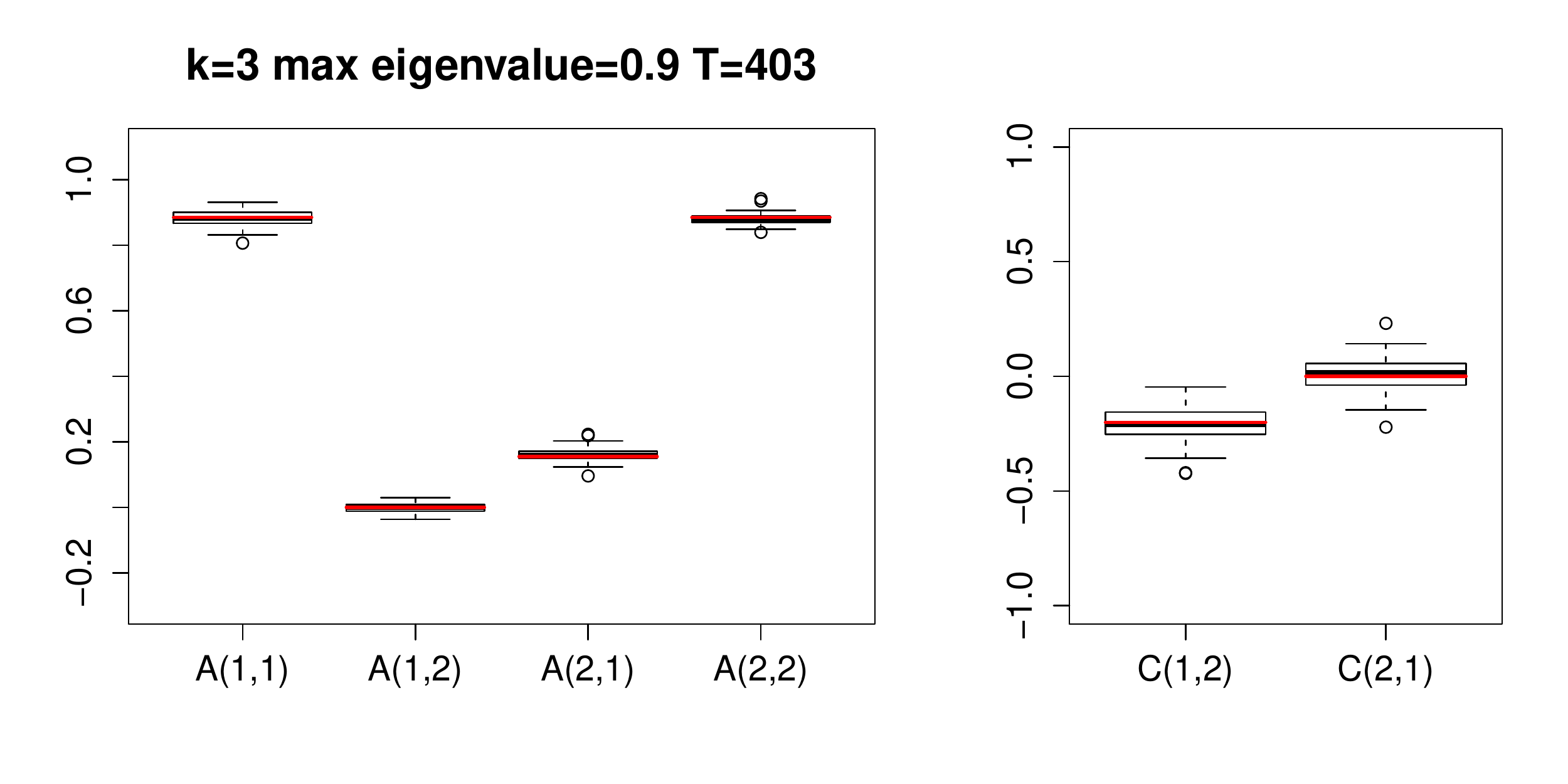} \hspace{.2 in}
\includegraphics[width=.47\textwidth]{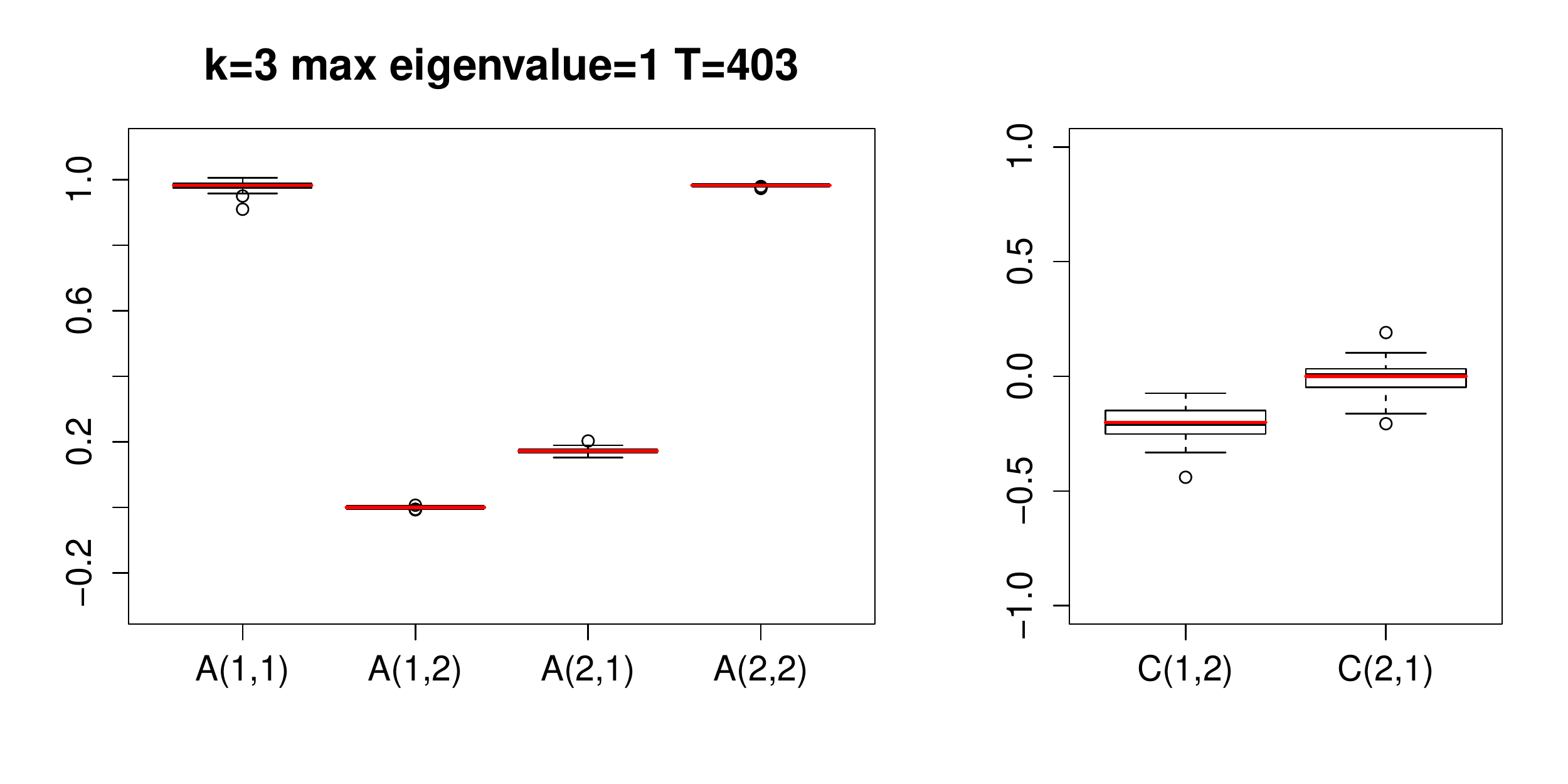}
\caption{As in Figure \ref{box2} but for $k = 3$.}
\label{box3}
\end{figure}

\end{document}